\documentclass[te,nameyear,final]{econsocart}
\RequirePackage[colorlinks,citecolor=blue,linkcolor=blue,urlcolor=blue,pagebackref]{hyperref}
\usepackage{booktabs} % For formal tables
\usepackage[ruled]{algorithm2e} % For algorithms
\usepackage[T1]{fontenc}
\usepackage{textcomp}
\usepackage{graphicx}
\usepackage{enumerate}
\usepackage{bm}
\usepackage{amsmath}
\usepackage{amsfonts}
\usepackage{amsthm}
\usepackage{amsfonts}
\usepackage{epsfig}
\usepackage{graphicx}
\usepackage{natbib}
\usepackage{subfigure}
\usepackage{color}
\usepackage{xcolor}
\usepackage{tikz}
\usepackage{accents}

\startlocaldefs

\theoremstyle{plain}

\newtheorem{theorem}{Theorem}
\newtheorem{proposition}{Proposition}

\newtheorem{observation}{Observation}

\theoremstyle{definition}

\SetAlFnt{\small}
\SetAlCapFnt{\small}
\SetAlCapNameFnt{\small}
\SetAlCapHSkip{0pt}
\IncMargin{-\parindent}

\newcommand{\lotto}{\mathcal{L}}
\newcommand{\blotto}{\mathcal{B}}
\newcommand{\general}{\mathcal{G}}

\newcommand{\uvec}[1]{\mathbf{1}_{#1}}

\newcommand{\uovec}[1]{\bm{u}_{\mathrm{O}}^{#1}}
\newcommand{\uoupvec}[1]{\bm{u}_{\mathrm{O}\uparrow 1}^{#1}}
\newcommand{\uevec}[1]{\bm{u}_{\mathrm{E}}^{#1}}
\newcommand{\wvec}[2]{\bm{w}_{#1}^{#2}}
\newcommand{\vvec}[2]{\bm{v}_{#1}^{#2}}
\newcommand{\vvvec}[1]{\bm{v}^{#1}}

\newcommand{\uo}[1]{U_{\mathrm{O}}^{#1}}
\newcommand{\uoup}[1]{U_{\mathrm{O}\uparrow 1}^{#1}}
\newcommand{\ue}[1]{U_{\mathrm{E}}^{#1}}
\newcommand{\w}[2]{W_{#1}^{#2}}
\newcommand{\vd}[2]{V_{#1}^{#2}}
\newcommand{\vv}[1]{V^{#1}}

\renewcommand{\Pr}{\mathbf{P}}
\newcommand{\Ex}{\mathbf{E}}
\newcommand{\sign}{\mathrm{sign}}
\newcommand{\val}[1]{\mathrm{val}\: #1}
\newcommand{\conv}[1]{\mathrm{conv}\!\left(#1\right)}

\newcommand{\matr}[1]{\mathbf{#1}}

\newcommand{\mset}[1]{\langle #1 \rangle}
\newcommand{\divid}[2]{{#1}\! \mid \!{#2}}
\newcommand{\ndivid}[2]{{#1}\! \nmid \!{#2}}
\newcommand{\sslash}{/\!\!/}
\newcommand{\divides}[2]{{#1}\! \mid \!{#2}}
\newcommand{\ndivides}[2]{{#1}\! \nmid \!{#2}}

\DeclareMathOperator{\lcm}{lcm}
\DeclareMathOperator{\card}{card}

\endlocaldefs

\begin{document}

\begin{frontmatter}

\title{Non-symmetric discrete Colonel Blotto game}
\runtitle{Non-symmetric discrete Colonel Blotto game}

\begin{aug}
% use \particle for den|der|de|van|von (only lc!)
% [add1]{\fnms{}~\snm{}\ead[label=e?]{}}
%
%% e-mail is mandatory for each author
%
%%% initials in fnms (if any) with spaces
%
\author[add]{\fnms{Marcin}~\snm{Dziubi\'{n}ski}\ead[label=e1]{m.dziubinski@mimuw.edu.pl}}
%%%%%%%%%%%%%%%%%%%%%%%%%%%%%%%%%%%%%%%%%%%%%%
%% Addresses                                %%
%%%%%%%%%%%%%%%%%%%%%%%%%%%%%%%%%%%%%%%%%%%%%%
\address[add]{%
\orgdiv{Faculty of Mathematics, Informatics, and Mechanics},
\orgname{University of Warsaw}}

\end{aug}

%% Put support info here. Reminder: do not thank the handling coeditor anonymously or by name
\begin{funding}
The author gratefully acknowledges financial support from the Polish National Science Centre through grant 2018/29/B/ST6/00174.
\end{funding}
%
%\coeditor{\fnm{[Name} \snm{Surname}; will be inserted later]}

\begin{abstract}
We study equilibrium strategies and the value of the asymmetric variant of the discrete Colonel Blotto game with $K \geq 2$ battlefields, $B \geq 1$ resources of the weaker player and $A > B$ resources of the stronger player. We derive equilibrium strategies and the formulas for the value of the game for the cases where the number of resources of the weaker player, $B$, is at least $2(\lceil A/K \rceil - 1)$ as well as for the cases where this number is at most $\lfloor A/K \rfloor$. In particular, we solve all the cases of the game which can be solved using the discrete General Lotto game of~\cite{Hart08}. We propose a constrained variant of the discrete General Lotto game and use it to derive equilibrium strategies in the discrete Colonel Blotto game, that go beyond the General Lotto solvable cases game.
\end{abstract}



\end{frontmatter}
%%%%%%%%%%%%%%%%%%%%%%%%%%%%%%%%%%%%%%%%%%%%%%%%%%%%%%%%%%%%%%%%%%%%%%%%%
%%%% Main text entry area:
%%%%%%%%%%%%%%%%%%%%%%%%%%%%%%%%%%%%%%%%%%%%%%%%%%%%%%%%%%%%%%%%%%%%%%%%%

\section{Introduction}
\label{sec:intro}

The Colonel Blotto game, proposed by~\cite{Borel21} is one of the oldest games studied in the modern game theory. It is a two player zero-sum game where two players, A and B, endowed with $A$ and $B$ units of resources ($A\geq B > 0$), respectively, compete over $K \geq 2$ battlefields. Each player, not observing the choice of the opponent, distributes his resources across the battlefields. A player wins all the battlefields where he has strictly more resources than the opponent, loses all the battlefields where he has strictly less resources than the opponent and the remaining battlefields are tied. In the classic formulation, all the battlefields have the same value and the score that a player gets at the end of the game is the average of the total number of the won minus the total number of the lost battlefields. In the discrete variant of the game, the units are not divisible and each battlefield must receive a discrete number of units.

Since its introduction, the game attracted interest of researchers in the areas of economics~\cite{Tukey49,LaslierPicard02}, operations research~\cite{BealeHeselden62,Penn71,BellCover80}  and, more recently, computer science~\cite{AhmadinejadDehghaniHajiaghayiLucierMahiniSeddighin19,
BehnezhadDehghaniDerakhshanHajiaghayiSeddighin22,VuLoiseau21}. The interest lasts to this day, which is mainly due to many applications of the game in the areas such as political competition~\cite{Laslier02,LaslierPicard02,Washburn13}, network security~\cite{ChiaChuang11,GuanWangYaoJiangHanRen20}, and military conflicts~\cite{Blackett54,ShubikWeber81}.
Still, the solution to the discrete Colonel Blotto games for all the possible values of parameters of the game is largely unknown.

The most significant contribution to obtaining an explicit solution to the game is the work of~\cite{Hart08}. Introducing a symmetric across battlefields variant of the game, called the Colonel Lotto game, and another game, called the discrete General Lotto game,\footnote{
The continuous variant  of the General Lotto game, without this name, was earlier defined by~\cite{Myerson93}.}
Hart found optimal strategies as well as the value of the discrete Colonel Blotto game in the symmetric case, where both players have the same number of units, as well as the asymmetric cases, where the numbers of units of the players are close, $\lfloor A/K \rfloor < B/K < A/K < \lceil A/K \rceil$. 
He also obtained the lower bound on the value of the game (payoff to player A) in the case where $\divides{K}{A}$ and $A \equiv K \pmod 2$, and the upper bound on the value of the game in the case where 
$\ndivides{K}{A}$, $\divides{K}{B}$, $B/K < \lfloor A/K \rfloor$, and $B$ is even.
\cite{Dziubinski17} constructed additional optimal strategies for the symmetric case and used that to obtain complete characterisation of an aggregate statistic of the optimal strategies in that case, called spectrum.
In a recent paper, \cite{LiangWangCaoYang23} studied the discrete Colonel Blotto game with two battlefields. They obtained complete characterization of the values of the game, provided complete direct characterization of optimal strategies of the players in the cases where the difference in resources between the players is $B = A-1$ as well as in the cases where $B \leq A-2$ and $A \bmod (A-B) = A-B-1$.\footnote{
Throughout the paper, given two natural numbers $X$ and $Y$, we use $X \bmod Y$ to denote the remainder of dividing $X$ by $Y$, i.e. the value $R \in \{0,\ldots,Y-1\}$ such that $X = K\cdot Y + R$, for a natural number $K$.
} In the remaining cases they provided complete but mainly indirect characterization, where the strategies of one of the players are characterized in terms of a set of conditions.

The continuous variant of the Colonel Blotto game is much better understood. \cite{BorelVille38} solved the game for the symmetric case, $A=B$, with $K=3$ battlefields. \cite{GrossWagner50} completed the solution of the symmetric case  and solved the asymmetric cases with $K = 2$ battlefields. Using an approach based on copula theory, \cite{Roberson06} completed the solution of the asymmetric cases. Most notably, he obtained complete characterization of the marginal distributions of equilibrium strategies.
\cite{MacdonellMastronardi15} obtained complete characterization of Nash equilibria in the case of $K=2$ battlefields.\footnote{
Other notable works that obtained solutions to different subcases of the continuous Colonel Blotto game include~\cite{LaslierPicard02}, \cite{Weinstein12}, and~\cite{Thomas18}.}
These advancements laid foundations for the study of various extensions of the Colonel Blotto game, such as the non-zero-sum variant studied by~\cite{RobersonKvasov12}; the generalization allowing heterogeneous values of battlefields, including the non-zero-sum variant, where the values are heterogeneous across the players, studied by~\cite{KovenockRoberson21}; the extension to more than two players, studied by~\cite{BoixAdseraEdelmanJayanti21}; and the extension allowing for head-starts studied by~\cite{VuLoiseau21}.\footnote{
Other notable variants of the Colonel Blotto game include the sequential non-zero-sum variant, studied by~\cite{Powell09}, and the coalitional variant, studied by~\cite{KovenockRoberson12}.}

The problem of computing optimal strategies in the Colonel Blotto game attracted attention from researchers in operations research and computer science from the early years. 
\cite{BealeHeselden62} proposed a method for finding approximate solutions of the Colonel Blotto game based on an auxiliary game, where the budget constraints of the players have to be satisfied in expectation only.
Partially using this approach, \cite{Penn71} proposed a generalized multiplier method, that can be applied to a large class of games, including the Colonel Blotto game. His technique, combined with fictitious play, was implemented in a computer program by~\cite{EisenLeMat68}. None of these papers is concerned with the actual computational complexity of the problem. The breakthrough in this regard was a paper by\cite{AhmadinejadDehghaniHajiaghayiLucierMahiniSeddighin19}, who showed that Colonel Blotto (and similar types of games with bilinear payoffs) can be solved in polynomial time with respect to the values of the parameters of the game. Their method, although polynomial, has high computational cost because it requires using the ellipsoid method to solve the linear program associated with the solution of the game. This was significantly improved by~\cite{BehnezhadDehghaniDerakhshanHajiaghayiSeddighin22} who, using a network-flow based representation of mixed strategies in Colonel Blotto games, constructed a linear program for solving the game, that is sufficiently small to allow for practical computation.
\cite{BoixAdseraEdelmanJayanti21} proposed a sampling algorithm that allows for generating mixed strategies in multiplayer Colonel Blotto games from the given marginal distributions. \cite{VuLoiseau21} proposed an efficient algorithm to compute good (but not necessarily equilibrium) marginal distributions in the Colonel Blotto game with favouritism. These marginal distributions can be then used to obtain approximate Nash equilibria in the game. \cite{BeagleholeHopkinsKaneLiuLovett23} proposed a fast sampling extension of the Multiplicative Weights Update algorithm and demonstrated how that can be used for solving Colonel Blotto game fast.

\subsection{Our contribution}
In this paper we provide optimal strategies and solve for the value of the discrete Colonel Blotto game for the asymmetric cases of the Colonel Blotto game with the parameters values
\begin{enumerate}
\item $2(\lceil A/K \rceil - 1) \leq B \leq K\lfloor A/K \rfloor$, except the cases of $K = 3$, $A \in \{7,13,19\}$ and odd $B$,\label{p:solved:1}
\item $1 \leq B\leq \lfloor A/K \rfloor$.\label{p:solved:3}
\end{enumerate}
The case~(\ref{p:solved:1}) constitutes the scenarios where the weaker player has at least twice the minimal number of resources that the stronger player can secure at every battlefield. We refer to this case as the \emph{high $B$ case}. The case~(\ref{p:solved:3}) constitutes the scenarios where the weaker player has at most the minimal number of resources that the stronger player can secure at every battlefield. We refer to this case as the \emph{low $B$ case}. Notice that in the low $B$ case, only the subcase where $B = \lfloor A/K\rfloor$ is non-trivial. 
In the case of $B < \lfloor A/K \rfloor$ it is immediate to see that any allocation of resources with at least $\lfloor A/K \rfloor$ resources at each battlefield is a dominant strategy for player $A$ and any strategy of player $B$ is a best response to that. The value of the game is $1$ in that case.

Given that the case of $K\lfloor A/K \rfloor < B \leq A$ was already solved by~\cite{Hart08}, this leaves the case of $\lfloor A/K \rfloor < B< 2(\lceil A/K \rceil - 1)$, the \emph{intermediate $B$ case}, as the major unsolved case.\footnote{
The analysis of the continuous variant of the game in~\cite{Roberson06} is also divided into similar three subcases: $2A/K \leq B < A$ (high $B$), $A/(K-1) \leq B < 2A/K$ (intermediate $B$), and $0 < B < A/(K-1)$ (low $B$).
}

Most of the analysis concerns the case~(\ref{p:solved:1}), of high values of $B$. We partition this analysis into three subcases
\begin{enumerate}[(a)]
\item $\ndivides{K}{A}$, $A > K$, $2\lfloor A/K \rfloor \leq B \leq K\lfloor A/K\rfloor$, and $B$ is even,\label{p:contrib:1}
\item $\divides{K}{A}$, $2A/K - 2 \leq B < A$, and $A \equiv K \pmod 2$,\label{p:contrib:2}
\item $\ndivides{K}{A}$, $A > K$, $2\lfloor A/K \rfloor \leq B \leq K\lfloor A/K\rfloor$, $B$ is odd, and either $K\neq 3$ or $A\notin\{7,13,19\}$,\label{p:contrib:3}
\end{enumerate}
We show that~(\ref{p:contrib:1}) and~(\ref{p:contrib:2}) constitute all the asymmetric cases that can be solved using the discrete General Lotto game of~\cite{Hart08} when $B \leq K\lfloor A/K\rfloor$ (the remaining cases constitute those already solved by~\cite{Hart08}: the symmetric case of $A = B$ and the asymmetric case of $K\lfloor A/K \rfloor < B < A$).
To solve case~(\ref{p:contrib:3}), we introduce and solve a constrained variant of the discrete General Lotto game and then use its solutions to construct optimal strategies in the Colonel Blotto game.

Our approach relies on the complete characterization of solutions of the discrete General Lotto game, obtained partially by~\cite{Hart08} and completed by~\cite{Dziubinski12}. Constructing optimal strategies for the players we rely on the approach of~\cite{Dziubinski17}, who constructs optimal strategies for the symmetric variant of the Colonel Blotto game by composing optimal strategies for $2$ and $3$ battlefields.
In our case, the optimal strategies are composed of optimal strategies for $2$ up to $5$ battlefields.

To our knowledge, our paper is the first to obtain equilibrium strategies as well as the values of the Colonel Blotto game with $K\geq 3$ battlefields and the number of resources of the weaker player at most $K\lfloor A/K \rfloor$. This covers a large part of parameter values, leaving the cases where the number of resources of the weaker player is between $\lfloor A/K \rfloor+1$ and $2\lceil A/K\rceil - 2$ as well as six cases with $K = 3$ and $(A,B) \in \{(7,5),(13,9),(13,11),(19,13),(19,15),(19,17)\}$ as the remaining unsolved fragments of the game.
%\footnote{
%In the appendix we present computational solutions of the six cases with $K = 3$.} 
We summarize the contribution for $K \geq 3$ battlefields in Table~\ref{tab:genk}.
\begin{table}[h!]
  \begin{center}
    \begin{tabular}{|c||c|c|}
      \hline      
      \textbf{$A$ and $B$} & \textbf{Value of the game} & \textbf{Equilibrium stratgies}\\
      \hline      
      \hline      
      $K\lfloor A/K \rfloor + 1 \leq B \leq A$ & \cite{Hart08} & \cite{Hart08} \\
      \hline
      $2(\lceil A/K \rceil - 1) \leq B \leq K\lfloor A/K \rfloor$ & \cite{Hart08}$^{\textrm{bounds}}$ & this paper \\
      (except $K = 3$, $A \in \{7,13,19\}$, odd $B$) & this paper$^{\textrm{exact}}$  & \\
      \hline
      $\lfloor A/K \rfloor + 1 \leq B \leq 2(\lceil A/K \rceil - 1)-1$ & open & open \\
      \hline
      $B = \lfloor A/K \rfloor$ & this paper & this paper \\
      \hline
      $1 \leq B \leq \lfloor A/K \rfloor-1$ & trivial & trivial \\
      \hline
    \end{tabular}
  \end{center}
  \caption{The values of the game and optimal strategies of the Colonel Blotto game in the cases of $K\geq 3$ battlefields. $^{\textrm{bounds}}$: lower bound in the case of $\divides{K}{A}$ and $A \equiv K \pmod 2$, and upper bound in the case of $\ndivides{K}{A}$, $\divides{K}{B}$, $B/K < \lfloor A/K \rfloor$, and $B$ even.}
  \label{tab:genk}
\end{table}

The case of $K = 2$ battlefields was already solved by~\cite{LiangWangCaoYang23}. However, in the case of even $A \geq 4$ and $B = A/2$, the characterization of equilibrium strategies for player A obtained by~\cite{LiangWangCaoYang23} is indirect, via a set of conditions, and we provide a direct description of an equilibrium strategy for player A for this case. Similarly, in the case of odd $A \geq 5$ and $B = (A-1)/2$ the characterization of equilibrium strategies for players A and B obtained by~\cite{LiangWangCaoYang23} is indirect and we provide direct description of equilibrium strategies for both players for this case.   
%
%We summarize the contribution for $K = 2$ battlefields in Table~\ref{tab:2k}. 
%\begin{table}[h!]
%  \begin{center}
%    \begin{tabular}{|c||c|c|}
%      \hline      
%      \textbf{$A$ and $B$} & \textbf{Value of the game} & \textbf{Equilibrium stratgies}\\
%      \hline      
%      \hline      
%      $A=B$ & \cite{Hart08} & \cite{Hart08} \\
%      \hline
%      $\lfloor A/2 \rfloor + 1 \leq B \leq A-1$ & \cite{LiangWangCaoYang23} & \cite{LiangWangCaoYang23} \\
%      \hline
%      $B = \lfloor A/2 \rfloor$ & \cite{LiangWangCaoYang23} & \cite{LiangWangCaoYang23} \\
%       & & this paper\\
%      \hline
%      $1 \leq B \leq \lfloor A/2 \rfloor-1$ & trivial & trivial \\
%      \hline
%    \end{tabular}
%  \end{center}
%  \caption{Values of the game and optimal strategies of the Colonel Blotto game in the cases of $K = 2$ battlefields.}
%  \label{tab:2k}
%\end{table}

The remainder of the paper is organized as follows. In Section~\ref{sec:blotto} we formally define the game. In Section~\ref{sec:analysis} we provide the analysis of the problem. In particular, in Section~\ref{sec:approach} we outline our approach, in Section~\ref{sec:genlot} we solve the General Lotto solvable cases and in Section~\ref{sec:beyond} we propose a constrained variant of the game and solve additional subcases of the high $B$ case. In Section~\ref{sec:low} we solve the low $B$ case. We conclude with a~discussion in Section~\ref{sec:disc}. All proofs are given in the Appendix.

\section{Definition of the problem}
\label{sec:blotto}

The asymmetric \emph{Colonel Blotto game} is defined as follows. There are two players A and B having $A > B \geq 1$ units (where $A$ and $B$ are natural numbers) respectively, to distribute simultaneously and independently over $K \geq 2$ battlefields. The game is denoted by $\blotto(A,B;K)$. Player A, with the larger number of resources, is referred to as the \emph{stronger player} and player B, with the smaller number of resources, is referred to as the \emph{weaker player}.

A pure strategy of a player with $X$ units is a $K$-partition, $\bm{x} = (x_1,\ldots,x_K)$, of $X$ so that $x_1 + \ldots + x_K = X$ and each $x_i$ is a natural number. The set of pure strategies of a player with $X \in \mathbb{Z}_{\geq 0}$ units is
\begin{equation*}
S_{\blotto}(X) = \left\{(z_i)_{i = 1}^K \in \mathbb{Z}_{\geq 0} : \sum_{i = 1}^K z_i = X \right\}.
\end{equation*}
A mixed strategy of a player with $X$ resources is a probability distribution on $S(X)$.

After the units are distributed, the payoff of each player is computed as follows. For each battlefield where a player assigned a strictly larger number of units than the opponent, he receives the score of $1$ and the opponent receives the score of $-1$. The score at the tied battlefields is $0$, for each player. The overall payoff is the average of payoffs obtained at all battlefields. Given the pure strategies $\bm{x}^{\mathrm{A}}\in S(A)$ and $\bm{x}^{\mathrm{B}} \in S(B)$ of A and B, respectively, the payoff to the stronger player, A, is given by
\begin{equation*}
h_{\mathcal{B}}\!\left(\bm{x}^{\mathrm{A}},\bm{x}^{\mathrm{B}}\right) = \frac{1}{K}\sum_{i = 1}^{K} \sign(x^{\mathrm{A}}_i - x^{\mathrm{B}}_i).
\end{equation*}
The payoff to the weaker player, B,  is minus the payoff to the stronger player. The payoffs from mixed strategy profiles are the expected values of the payoffs from the pure strategy profiles.

The Colonel Blotto game is a finite zero-sum game. Hence, for any $K \geq 2$ and $A \geq B \geq 1$, $K,A,B \in \mathbb{Z}_{\geq 0}$, 
$\blotto(A,B;K)$ has a Nash equilibrium in mixed strategies. In addition, all equilibria are payoff equivalent and interchangeable. Interchangeability means that if mixed strategy profiles $(\bm{\xi},\bm{\zeta})$ and $(\bm{\xi}',\bm{\zeta}')$ are Nash equilibria so are $(\bm{\xi}',\bm{\zeta})$ and $(\bm{\xi},\bm{\zeta}')$. Equilibrium strategies of a player are referred to as his \emph{optimal strategies}. The equilibrium payoff of the player receiving non-negative payoff in equilibrium (i.e. player A) is called the \emph{value of the game}.
We are interested in finding optimal strategies of the players and establishing the value of the Colonel Blotto game.

\section{Analysis}
\label{sec:analysis}

In the analysis we follow the approach of~\cite{Hart08}, who introduced two additional games, related to the Colonel Blotto game, the \emph{Colonel Lotto} game and the \emph{General Lotto} game.

The \emph{Colonel Lotto game} involves two players, A and B, endowed with $A$ and $B$ units, respectively. Each player chooses an unordered $K$-partition of his resources. A strategy of a player with $X$ resources is $\bm{x} = \mset{x_1,\ldots,x_K}$ such that $x_1 + \ldots + x_K = X$ and each $x_i$ is a natural number. 
After the partitions are chosen, the parts are matched uniformly at random. The result at each matched pair is computed like in the Colonel Blotto game. The payoff to a player is an average of his results over all possible pairings. Given the strategies $\bm{x}^{\mathrm{A}}$ and $\bm{x}^{\mathrm{B}}$ of A and B, respectively, the payoff to player A is equal to
\begin{equation*}
h_{\mathcal{L}}\!\left(\bm{x}^{\mathrm{A}},\bm{x}^{\mathrm{B}}\right) = \frac{1}{K^{2}}\sum_{i = 1}^{K} \sum_{j = 1}^{K} \sign(x^{\mathrm{A}}_i - x^{\mathrm{B}}_j)
\end{equation*}
and the payoff to player B is equal to minus this quantity. The game is finite and zero-sum.

The Colonel Lotto game is related to the Colonel Blotto in the following sense. Given a pure strategy $\bm{x}\in S(X)$ of a player with $X$ resources in the Colonel Blotto game, let $\sigma(\bm{x})$ denote a mixed strategy in the Colonel Blotto game that assigns equal probability, $\frac{1}{K!}$, to each permutation of $\bm{x}$. Similarly, given a mixed strategy $\bm{\xi}$ of a player with $X$ resources in the Colonel Blotto game, let $\sigma(\bm{\xi})$ denote a mixed
strategy obtained by replacing each pure strategy, $\bm{x}$, in the support of $\bm{\xi}$ by $\sigma(\bm{x})$. The strategies $\sigma(\bm{x})$ and $\sigma(\bm{\xi})$ are called \emph{symmetric across battlefields}. As was observed in~\cite{Hart08}, $h_{\mathcal{B}}(\sigma(\bm{\xi}), \bm{y}) = h_{\mathcal{L}}(\bm{\xi},\bm{y})$, for any pure strategy, $\bm{y}$,
of player B. Consequently, $h_{\mathcal{B}}(\sigma(\bm{\xi}), \bm{\zeta}) = h_{\mathcal{L}}(\bm{\xi},\bm{\zeta})$, for any mixed strategy $\bm{\zeta}$ of player B. Analogous facts hold for the strategies of player B. This leads to the following observation:

\begin{observation}[\cite{Hart08}]
\label{obs:1}
The Colonel Blotto game $\blotto(A, B; K)$ and the Colonel Lotto game $\lotto(A, B; K)$ have the same value. 
Moreover, the mapping $\sigma$ maps the optimal strategies in the Colonel Lotto game onto the optimal strategies in the Colonel Blotto game that are symmetric across battlefields.
\end{observation}

The \emph{General Lotto game} involves two players, $A$ and $B$, with budgets $a\in \mathbb{R}_{\geq 0}$ and $b \in \mathbb{R}_{\geq 0}$, respectively. A pure strategy of a player with budget $m$ is a probability distribution on natural numbers with the mean equal to $m$,\footnote{
Given a non-empty set $X$, $\Delta(X)$ denotes the set of probability distributions on $X$.}
\begin{equation*}
S_{\general}(m) = \left\{\bm{p} \in \Delta(\mathbb{Z}_{\geq 0}) : \sum_{i \geq 1} ip_i = m \right\}.
\end{equation*}
Given a strategy profile in the General Lotto game, $(\bm{p}^{\mathrm{A}},\bm{p}^{\mathrm{B}})$, and two non-negative integer valued random variables $X^{\mathrm{A}}$ and $X^{\mathrm{B}}$ such that, for any $i \in \mathbb{Z}_{\geq 0}$, $\Pr(X^{\mathrm{A}} = i) = p^{\mathrm{A}}_i$ and $\Pr(X^{\mathrm{B}} = i) = p^{\mathrm{B}}_i$, the payoff to player A is equal to
\begin{displaymath}
H\!\left(X^{\mathrm{A}},X^{\mathrm{B}}\right) = \Pr\!\left(X^{\mathrm{A}} > X^{\mathrm{B}}\right) - \Pr(X^{\mathrm{A}} < X^{\mathrm{B}})
\end{displaymath}
and the payoff to player B is equal to $H(X^{\mathrm{B}},X^{\mathrm{A}}) = -H(X^{\mathrm{A}},X^{\mathrm{B}})$.
The game is zero-sum and infinite.

The General Lotto game is connected to the Colonel Lotto game as follows. Any unordered $K$-partition $\bm{z} = \mset{z_1,\ldots,z_K}$ of a natural number $C$ corresponds to a probability distribution on non-negative integers, $\bm{q}$, such that the probability of value $i \in \mathbb{Z}_{\geq 0}$, $q_i$, is equal the frequency of $i$ in $\bm{z}$, i.e. to the number of times $i$ appears in $\bm{z}$ divided by $K$. The mean of this probability distribution is then $C/K$. We say that $K$-partition $\bm{z}$ \emph{$(C,K)$-implements} probability distribution $\bm{q}$.

This construction links pure strategies $\bm{x}^{\mathrm{A}}$ and $\bm{x}^{\mathrm{B}}$ of players A and B in the Colonel Lotto game with discrete non-negative integer valued probability distributions with means $A/K$ and $B/K$, respectively.  The payoff $h_{\mathcal{L}}(\bm{x}^{\mathrm{A}},\bm{x}^{\mathrm{B}})$ can be then written as
\begin{displaymath}
h_{\mathcal{L}}(\bm{x}^{\mathrm{A}},\bm{x}^{\mathrm{B}}) = H(X^{\mathrm{A}}, X^{\mathrm{B}}),
\end{displaymath}
where $X^{\mathrm{A}}$ and $X^{\mathrm{B}}$ are non-negative integer valued random variables bounded from above by $A$ and $B$, respectively, and distributed according to the probability distributions associated with $\bm{x}^{\mathrm{A}}$ and $\bm{x}^{\mathrm{B}}$, respectively.
General Lotto game could be seen as a generalization of Colonel Lotto game which allows for strategies of the players to be unbounded random variables. 
Every strategy of player A in the Colonel Lotto game $\lotto(A,B;K)$ $(A,K)$-implements a strategy of player A in the General Lotto game $\general(A/K, B/K)$. However, not every strategy of player A in $\general(A/K, B/K)$ has the corresponding strategy in $\lotto(A,B;K)$ that $(A,K)$-implements it.
The strategies of player A in the General Lotto game $\general(A/K, B/K)$ for which there exists a corresponding strategy in the Colonel Lotto game $\lotto(A,B;K)$ are called \emph{$(A,K)$-feasible}. The $(B,K)$-feasible strategies of $\general(A/K,B/K)$ are defined analogously.
A strategy profile in the General Lotto game $\general(A/K,B/K)$ is $(A,B,K)$-feasible if the strategy of player A in the profile is $(A,K)$-feasible and the strategy of player B in the profile is $(B,K)$-feasible.
This discussion leads to the following observation.

\begin{observation}[\cite{Hart08}]
\label{obs:2}
If an equilibrium in the General Lotto game $\general(A/K,B/K)$ is $(A,B,K)$-feasible then the implementing strategy profile is an equilibrium in the Colonel Lotto game $\lotto(A, B; K)$. In this case every strategy implementing an optimal strategy of player A in $\general(A/K,B/K)$ is an optimal strategy of A in $\lotto(A,B;K)$ and an analogous fact holds for the optimal strategies of player B. 
Moreover, in this case, the Colonel Lotto game, $\lotto(A, B; K)$, and the General Lotto game, $\general(A/K,B/K)$, have the same value.
\end{observation}

\subsection{The general approach}
\label{sec:approach}

By Observations~\ref{obs:1} and~\ref{obs:2}, solutions to the Colonel Blotto game $\blotto(A,B;K)$
can be obtained by first solving the General Lotto game $\general(A/K,B/K)$ and then $(A,K)$-implementing an optimal strategy for player A and $(B,K)$-implementing an optimal strategy for player B. 
We follow this approach in this paper. Before getting to finding implementations of optimal strategies in the General Lotto games, we recall the probability distributions that the optimal strategies in the General Lotto games are based on.
These strategies were defined in~\cite{Hart08} and~\cite{Dziubinski12}.

Given $m \geq 1$, let $\uo{m}$, $\ue{m}$, $\uoup{m}$ (in the case of $m \geq 2$),
$\w{j}{m}$ (for $j \in \{1,\ldots,m-1\}$) and $\vd{j}{m}$ (for $j \in \{1,\ldots,m\}$),
be vectors defined as follows:\footnote{
Throughout the paper we number the rows of matrices starting for $0$ and we number the columns starting from $1$.
}
\begin{align*}
\uo{m} & := [\underbrace{0,1,\ldots,0,1,0}_{2m+1}]^T,\\
\ue{m} & := [\underbrace{1,0,\ldots,1,0,1}_{2m+1}]^T,\\
\uoup{m} & := [\underbrace{0,0,1,0\ldots,1,0,0}_{2m+1}]^T,\\
\w{j}{m} & := [1,\underbrace{0,2,\ldots,0,2}_{2(j-1)},0,1,\underbrace{2,0,\ldots,2,0}_{2(m-j)} ]^T,\\
\vd{j}{m} & := [\underbrace{0,2,\ldots,0,2}_{2(j-1)},0,1,2\underbrace{0,2,\ldots,0,2}_{2(m-j)} ]^T.
\end{align*}
Given $m \geq 1$, let $\uovec{m}$, $\uevec{m}$, $\uoupvec{m}$ (in the case of $m \geq 2$),\footnote{
In the cases when $\uoup{m}$ or $\uoupvec{m}$ is multiplied by $0$, to simplify the notation, we will sometimes write $\uoup{m}$ or $\uoupvec{m}$, respectively, for $m = 1$.} $\wvec{j}{m}$ (for $j \in \{1,\ldots,m-1\}$) and $\vvec{j}{m}$ (for $j \in \{1,\ldots,m\}$),
be stochastic vectors defined as follows:
\begin{align*}
\uovec{m} & := \left(\frac{1}{m}\right) \uo{m}, \qquad\qquad
& \uevec{m} & := \left(\frac{1}{m+1}\right) \ue{m},\\
\uoupvec{m} & := \left(\frac{1}{m-1}\right) \uoup{m},
& \wvec{j}{m} & := \left(\frac{1}{2m}\right) \w{j}{m},\\
\vvec{j}{m} & := \left(\frac{1}{2m+1}\right) \vd{j}{m} & &.
\end{align*}

Lastly, let
\begin{equation*}
\mathcal{U}^m = \{\uevec{m}, \uovec{m}\}, \qquad\qquad
\mathcal{W}^m = \{\wvec{1}{m},\ldots, \wvec{m-1}{m}\},\qquad\qquad
\mathcal{V}^m = \{\vvec{1}{m},\ldots, \vvec{m}{m}\}
\end{equation*}
denote sets of stochastic vectors defined above (where $\mathcal{W}^1 = \varnothing$).
Given a vector $\bm{x}$, a set of vectors $\mathcal{Y}$, and $\lambda_1,\lambda_2 \in \mathbb{R}$, let
\begin{displaymath}
\lambda_1 \bm{x} + \lambda_2 \mathcal{Y} = \{\lambda_1 \bm{x} + \lambda_2 \bm{y} : \bm{y} \in \mathcal{Y} \}
\end{displaymath}
denote the set of distributions that can be obtained by linearly combining $\bm{x}$ and the vectors from $\mathcal{Y}$ with coefficients $\lambda_1$ and
$\lambda_2$, respectively. Similarly, given two sets of vectors $\mathcal{X}$ and $\mathcal{Y}$ as well as $\lambda_1,\lambda_2 \in \mathbb{R}$, let
\begin{displaymath}
\lambda_1 \mathcal{X} + \lambda_2 \mathcal{Y} = \{\lambda_1 \bm{x} + \lambda_2 \bm{y} : \bm{x} \in \mathcal{X} \textrm{ and } \bm{y} \in \mathcal{Y} \}
\end{displaymath}
denote the set of distributions obtained by linearly combining distributions from $\mathcal{X}$ and $\mathcal{Y}$ with coefficients $\lambda_1$ and $\lambda_2$,
respectively. Given a set of vectors $\mathcal{X}$ let $\conv{\mathcal{X}}$ denote the set of all convex combinations of vectors from $\mathcal{X}$.

A pure strategy of a player with $C$ units in a Colonel Lotto game with $K$ battlefields can be represented as a $1\times K$ row vector of non-negative integers that sum up to $C$. Let us denote the set of such vectors as $U(C)$.
A mixed strategy $\bm{\zeta}$ such that the probability $\zeta_{\bm{z}}$ of each unordered $K$-partition $\bm{z}$ in the support of $\bm{\zeta}$ is a rational number, $\zeta_{\bm{z}} = l_{\bm{z}}/k_{\bm{z}}$,  can be represented as a $L \times K$ matrix, where $L = \sum_{\bm{z}\in U(C)} q \lcm_{\bm{x}\in U(C)}(k_{\bm{x}}) l_{\bm{z}}/k_{\bm{z}}$, $q \geq 1$ is a natural number, $\lcm_{\bm{x}\in U(C)}(k_{\bm{x}})$ is the least common multiple of the denominators $(k_{\bm{x}})_{\bm{x}\in U(C)}$, and each row representing partition $\bm{x}$ is repeated $q \lcm_{\bm{x}\in U(C)}(l_{\bm{x}}) l_{\bm{z}} /k_{\bm{z}}$ times.

Given a $L \times K$ matrix of non-negative integers, $\mathbf{X}$, and a non-negative integer $i\in \mathbb{Z}_{\geq 0}$, let $\# \mathbf{X}(i)$ be the number of appearances of $i$ in $\mathbf{X}$. A vector $\bm{x}$ such that $x_i = \# \mathbf{X}(i)$ is called the \emph{cardinality vector} of $\mathbf{X}$ and is denoted by $\card(\mathbf{X})$.
A mixed strategy of a player with $C$ units in a Colonel Lotto game with $K$ battlefields, represented by a $L\times K$ matrix $\mathbf{X}$, gives rise to a strategy $\bm{z} = \card(\mathbf{X})/(LK)$ in a General Lotto game, where the player has budget $C/K$. In this case matrix $\mathbf{X}$ \emph{$(C,K)$-implements} $\bm{z}$. 

To implement optimal strategies in the General Lotto game we will construct matrices that implement them. 
Often these matrices will be constructed from simpler matrices by means of horizontal and vertical composition.
Given a $L \times K_1$ matrix $\mathbf{X}$ and a $L \times K_2$ matrix $\mathbf{Y}$, $\mathbf{X} | \mathbf{Y} = [\mathbf{X} | \mathbf{Y}]$ denotes a $L \times (K_1+K_2)$ matrix being the horizontal composition of matrices $\matr{X}$ and $\matr{Y}$. Similarly, given a $L_1 \times K$ matrix $\mathbf{X}$ and a $L_2 \times K$ matrix $\mathbf{Y}$, we use
\begin{displaymath}
\mathbf{X} \sslash \mathbf{Y} = \left[ \begin{array}{c} \mathbf{X} \\
                                                   \mathbf{Y} \end{array} \right]
\end{displaymath}
to denote a $(L_1+L_2)\times K$ matrix being the vertical composition of $\mathbf{X}$ and $\mathbf{Y}$. We also use the standard generalizations of these operators to multiple arguments,
$|_{i = 1}^{k} \mathbf{X}_i$ and $\sslash_{i = 1}^{k} \mathbf{X}_i$, where $\mathbf{X}_i$ are matrices satisfying the required dimensionality constraints.
Clearly $\card(\mathbf{X} | \mathbf{Y}) = \card(\mathbf{X}) + \card(\mathbf{Y})$ as well as 
$\card(\mathbf{X} \sslash \mathbf{Y}) = \card(\mathbf{X}) + \card(\mathbf{Y})$.

\cite{Hart08} gave necessary and sufficient conditions for $(C,K)$-implementability of strategies $\uevec{m}$ and
$\uovec{m}$.

\begin{proposition}[\cite{Hart08}]
\label{pr:hart}
Let $C \geq 1$ and $K \geq 2$ be natural numbers such that $\divides{K}{C}$ and let $m = C/K$.
\begin{enumerate}
\item $\uovec{m}$ is $(C,K)$-feasible if and only if $C \equiv K \pmod 2$.\label{p:hart:1}%\footnote{
%Given natural numbers $x$, $y$ and $z$, $x \equiv y \pmod z$ means that $x$ and $y$ are \emph{congruent} modulo $z$, that is $x \bmod z = y \bmod z$.}
\item $\uevec{m}$ is $(C,K)$-feasible if and only if $C$ is even.\label{p:hart:2}
\end{enumerate}
\end{proposition}

To prove this proposition, \cite{Hart08} provided explicit descriptions of the implementing strategies. We recall these strategies below. Our presentation is based on the representation proposed in~\cite{Dziubinski17}, as it is useful for representing and constructing more complex strategies out of other strategies. It will also help to familiarize the reader with our approach to constructing optimal strategies later in the paper.

Given a natural number $m \geq 1$, let $\bm{E}(m)$ be $(m+1) \times 2$ matrix defined as follows:
\begin{displaymath}
\bm{E}(m) = \left[ \begin{array}{c c}
                       0 & 2m \\
                       2 & 2m-2 \\
                       \vdots & \vdots \\
                       2m-2 & 2 \\
                       2m & 0
                       \end{array}
                 \right],
\end{displaymath}
that is the $i$'th row, $e_{i} = \left[ \begin{array}{c c} 2i & 2(m-i) \end{array} \right]$.
Each row of matrix $\bm{E}(m)$ represents a bipartition, a strategy in the Colonel Lotto game with $K=2$ battlefields and a player with $2m$ units.
The whole matrix represents a mixed strategy, where each partition (row) is chosen with probability $1/m$.
Notice that $\card(\bm{E}(m)) = 2\ue{m}$ and $\card\left(|_{i = 1}^{l} \bm{E}(m)\right) = 2l\ue{m}$, for all $l \in \mathbb{Z}_{\geq 0}$.
Clearly $\left(|_{i = 1}^{l} \bm{E}(m)\right)/(2l(m+1)) = \uevec{m}$, for all $l \in \mathbb{Z}_{\geq 0}$.
Thus $|_{i = 1}^{\frac{K}{2}} \bm{E}(m)$ $(mK,K)$-implements $\uevec{m}$ in the case of $K$ being even.

Let $\bm{RE}(m) = \bm{RE}^{\mathrm{I}}(m) \sslash \bm{RE}^{\mathrm{II}}(m)$ be a $(m+1) \times 3$ matrix consisting of two blocks,
$\bm{RE}^{\mathrm{I}}(m)$ and $\bm{RE}^{\mathrm{II}}(m)$, where
\begin{displaymath}
\bm{RE}^{\mathrm{I}}(m) = \left[ \begin{array}{c c c}
                       0 & m & 2m \\
                       2 & m+2 & 2m-4 \\
                       \vdots & \vdots & \vdots \\
                       m & 2m & 0
                       \end{array}
                 \right],
\end{displaymath}
that is the $i$'th row, $re^{\mathrm{I}}_{i} = \left[ \begin{array}{c c c} 2i & m + 2i & 2m-4i \end{array} \right]$, and
\begin{displaymath}
\bm{RE}^{\mathrm{II}}(m) = \left[ \begin{array}{c c c}
                       0 & m+2 & 2m-2 \\
                       2 & m+4 & 2m-6 \\
                       \vdots & \vdots & \vdots \\
                       m-2 & 2m & 2
                       \end{array}
                 \right],
\end{displaymath}
that is the $i$'th row $re^{\mathrm{II}}_{i} = \left[ \begin{array}{c c c} 2i & m + 2i + 2 & 2m-4i-2 \end{array} \right]$.
It can be easily checked that $\card\left(\bm{RE}(m)\right) = 3\ue{m}$.
Hence $\card\left(\bm{RE}(m) | \left(|_{i = 1}^{l} \bm{E}(m)\right)\right)/((2l+3)(m+1)) = \uevec{m}$, for all $l \in \mathbb{Z}_{\geq 0}$.
Thus $\left(\bm{RE}(m) | \left(|_{i = 1}^{\frac{K-3}{2}} \bm{E}(m)\right)\right)$ $(mK,K)$-implements $\uevec{m}$ in the case of $K$ being odd and greater or equal to $3$.

Similarly, let $\bm{O}(m)$ be a $m \times 2$ matrix defined as $\bm{O}(m) = \bm{E}(m-1) + \bm{1}$, that is
row $o_{i} = \left[ \begin{array}{c c} 2i+1 & 2(m-i)+1 \end{array}\right]$, and let $\bm{RO}(m) = \bm{RE}(m-1) + \bm{1}$.
The cardinality vector for $\bm{O}(m)$ is the cardinality vector for $\bm{E}(m-1)$, $2\ue{m-1}$, `shifted up' by one place,
which is exactly $2\uo{m}$. Similarly, $\card(\bm{RO}(m)) = 3\uo{m}$. Thus $|_{i = 1}^{\frac{K}{2}} \bm{O}(m)$ $(mK,K)$-implements 
$\uovec{m}$ in the case of $K$ being even and $\left(\bm{RO}(m) | \left(|_{i = 1}^{\frac{K-3}{2}} \bm{O}(m)\right)\right)$ $(mK,K)$-implements 
$\uovec{m}$ in the case of $K$ being odd.

\subsection{The General Lotto solvable cases}
\label{sec:genlot}

In this section we solve all the non-symmetric cases of the Colonel Blotto game that have solutions that correspond to solutions in the General Lotto game. We split the analysis into two subcases of $K$ not dividing $A$ and and $K$ dividing $A$. This is because the corresponding General Lotto games, $\general(A/K,B/K)$, have different sets of optimal strategies depending on whether $A/K$ is an integer or not.

\subsubsection{The case of $\ndivides{K}{A}$}
Below we recall the results from~\cite{Hart08} and~\cite{Dziubinski12}, characterizing the optimal strategies and the value of the game in the General Lotto game $\general(a,b)$ with $a > b > 0$ and non-integer $a$. We collect these results in a single statement below.

\begin{theorem}[\cite{Hart08} and~\cite{Dziubinski12}]
\label{th:hart4}
Let $a =m + \alpha$ and $b \leq m$, where $m \geq 1$ is an integer and $0 < \alpha < 1$. Then the value of General Lotto game $\general(a,b)$ is
\begin{displaymath}
\val{\general(a,b)} = (1-\alpha)\frac{\lfloor a \rfloor - b}{\lfloor a \rfloor} + \alpha \frac{\lceil a \rceil - b}{\lceil a \rceil} = 1 - \frac{(1-\alpha)b}{m} - \frac{\alpha b}{m+1}.
\end{displaymath}
The optimal strategies are as follows:
\begin{enumerate}
\item Strategy $\bm{y} = (1 - b/m) \uvec{0} + (b/m) \uevec{m}$ is the unique optimal strategy of Player $\mathrm{B}$.
\item Strategy $\bm{x}$ is optimal for Player $\mathrm{A}$ if and only if
\begin{displaymath}
\bm{x} \in \conv{\mathcal{U}^{m,\alpha} \cup \mathcal{X}^{m,\alpha}}, \textrm{where}
\end{displaymath}
\begin{itemize}
\item $\mathcal{U}^{m,\alpha} = (1-\alpha)\mathcal{U}^m + \alpha \uovec{m+1}$, if $b = m$,
\item $\mathcal{U}^{m,\alpha} = \left\{(1-\alpha)\uovec{m} + \alpha \uovec{m+1}\right\}$, if $b < m$
\end{itemize}
and
\begin{itemize}
\item $\mathcal{X}^{m,\alpha} = \alpha\delta \mathcal{V}^{m} + \left(1 - \alpha\delta \right) \mathcal{U}^{m}$, if $0 < \alpha \leq \frac{m+1}{2m+1}$ and $b = m$,
\item $\mathcal{X}^{m,\alpha} = \alpha\delta \mathcal{V}^{m} + \left(1 - \alpha\delta \right) \uovec{m}$, if $0 < \alpha \leq \frac{m+1}{2m+1}$ and $b < m$,
\item $\mathcal{X}^{m,\alpha} = (1-\alpha)\sigma \mathcal{V}^{m} + \left(1 - (1-\alpha)\sigma \right) \uovec{m+1}$, if $\frac{m+1}{2m+1} < \alpha < 1$, where
\end{itemize}
\begin{equation*}
\delta = \frac{2m+1}{m+1},\quad \sigma = \frac{2m+1}{m}.
\end{equation*} 
\end{enumerate}
\end{theorem}

We start by establishing sufficient and necessary conditions for the unique optimal strategy of the weaker player $B$, $\bm{y} = (1 - b/m) \uvec{0} + (b/m) \uevec{m}$, to be $(B,K)$-feasible when $b = B/K$ and $m = \lfloor A/K\rfloor$.

\begin{proposition}
\label{pr:aimpl3}
Let $m \geq 1$, $K \geq 2$, and $B \leq Km$ be natural numbers and let $b = B/K$. Then 
\begin{equation*}
\left(1-\frac{B}{Km}\right) \uvec{0} + \frac{B}{Km} \uevec{m}
\end{equation*}
is $(B,K)$-feasible if and only if $B$ is even and $B \geq 2m$.
\end{proposition}

By Proposition~\ref{pr:aimpl3}, if the Colonel Blotto game $\blotto(A,B;K)$ with $A > B$ and $\ndivides{K}{A}$ is solvable by the corresponding General Lotto game, then the number of units of the weaker player, $B$, is even
and it is at least twice the average number of units per battlefield, $m = \lfloor A/K\rfloor$, that the stronger player has.
In the following two propositions we establish optimal strategies in the General Lotto game for the stronger player A
which are $(A,K)$-feasible. 
Proposition~\ref{pr:aimpl1} covers the cases of $A$ and $K$ having the same parity.

\begin{proposition}
\label{pr:aimpl1}
Let $A > K \geq 2$ be natural numbers such that $K \nmid A$ and let $m = \lfloor A/K \rfloor$ and $\alpha = (A \bmod K)/K$. Then $(1-\alpha) \uovec{m} + \alpha \uovec{m+1}$ is $(A,K)$-feasible if and only if $A \equiv K \pmod 2$.
\end{proposition}

Proposition~\ref{pr:aimpl2} covers all the cases except those where $K = 3$ and $A \in \{7,13,19\}$. Since in the omitted cases $A$ and $K$ are both odd, these cases are covered by Proposition~\ref{pr:aimpl1}.
 
\begin{proposition}
\label{pr:aimpl2}
Let $A > K \geq 2$ be natural numbers such that $K \nmid A$ and let $m = \lfloor A/K \rfloor$ and $\alpha = (A \bmod K)/K$.
Let $\delta$ and $\sigma$ be defined as in Theorem~\ref{th:hart4} and let 
\begin{align*}
\bm{x} & = \alpha \delta \vvvec{m} + (1 - \alpha \delta) \uovec{m}\textrm{ and}\\
\bm{y} & = \left(\frac{1-\alpha}{\alpha}\right) \bm{x} + \left(\frac{2\alpha - 1}{\alpha}\right) \left( (1 - \alpha) \uovec{m} + \alpha \uovec{m+1} \right) \\
        & = \left(\frac{m}{m+1}\right)\left( (1-\alpha) \sigma \vvvec{m} + (1 - (1-\alpha) \sigma) \uovec{m+1} \right) + \left(\frac{1}{m+1}\right)\left( (1 - \alpha)\uovec{m} + \alpha \uovec{m+1} \right),
\end{align*}
where $\vvvec{m} = \frac{1}{m}\sum_{i = 1}^{m} \vvec{i}{m}$.

\begin{enumerate}
\item If $\alpha \leq \frac{1}{2}$ and either $K \neq 3$ or $m \notin \{2,4,6\}$, then $\bm{x}$ is $(A,K)$-feasible.\label{p:aimpl2:1}
\item If $\alpha > \frac{1}{2}$, then $\bm{y}$ is $(A,K)$-feasible.\label{p:aimpl2:2}
\end{enumerate}
\end{proposition} 

By Propositions~\ref{pr:aimpl1} and~\ref{pr:aimpl2}, for any $K \geq 2$, $B \geq 0$, and $A \geq K$ such that $\ndivides{K}{A}$, there exists an $(A,K)$-feasible optimal strategy for the stronger player in the General Lotto game $\general(A/K,B/K)$. Hence we have established sufficient and necessary conditions on the values of $A$, $B$, and $K$, when $A > B$ and $\ndivides{K}{A}$, under which optimal strategies in the Colonel Blotto game $\blotto(A,B;K)$ correspond to optimal strategies in the General Lotto game. In consequence, we are able to determine the value of the Colonel Lotto and the Colonel Blotto games for these parameter values.
%as well as some features of the optimal strategies in the Colonel Blotto game based on Theorem~\ref{th:hart4}.
%In particular, if the players use their optimal strategies then player $B$ never assigns an odd number of units to a battlefield and never assigns more than $2\lfloor A/K \rfloor$ units to a battlefield. Player $A$ never leaves a battlefield with no units assigned and never assigns more than $2\lfloor A/K \rfloor+1$ units to a battlefield.
%This is stated in the theorem below.

\begin{theorem}
\label{th:value:ndiv:even}
Let $K\geq 2$, $A > K$, and $1 \leq B \leq K\lfloor A/K \rfloor$ be natural numbers such that $\ndivides{K}{A}$. There exists an $(A,B,K)$-feasible profile of optimal strategies in General Lotto game $\general(A/K,B/K)$ if and only if $B$ is even and $B \geq 2\lfloor A/K \rfloor$. In this case the value of Colonel Lotto game $\lotto(A,B;K)$ and Colonel Blotto game $\blotto(A,B;K)$ is
\begin{align*}
\val{\blotto(A,B;K)} & = \frac{A-B}{A} - \frac{B}{A}\left(\frac{R(K-R)}{(A-R)(A+K-R)}\right),
\end{align*}
where $R = A\bmod K$.
%Moreover, if $\zeta$ is an optimal strategy of player B then $\spec(\zeta) = \left(1-\frac{B}{Km}\right) \uvec{0} + \frac{B}{Km} \uevec{m}$. Using their optimal strategies, player B assigns between $0$ and $2\lfloor A/K \rfloor$ units to a battlefield and player A assigns between $1$ and $2\lfloor A/K \rfloor + 1$ units to a battlefield.
\end{theorem}

The value of the game is equal to the value of the game in the continuous case, $(A-B)/A$, minus a reminder resulting from indivisibility of the number of units of the stronger player by the number of battlefields. If the number of units of the weaker player is even, it allows him to benefit from this indivisibility and secure higher payoff than in the case of continuous resources.  

\subsubsection{The case of $\divides{K}{A}$}
Like in the previous section, we first recall the results from~\cite{Hart08} and~\cite{Dziubinski12}, characterizing the optimal strategies and the value of the game in the General Lotto game $\general(a,b)$ with $a > b > 0$ and an integer $a$. We collect these results in a single statement below. To save space, we restrict the statement to parts of the characterization that we use in this paper.

\begin{theorem}[\cite{Hart08} and~\cite{Dziubinski12}]
\label{th:hart2}
Let $m > b > 0$, where $m$ is an integer. Then the value of General Lotto game $\general(m,b)$ is
\begin{displaymath}
\val{\general(m,b)} = \frac{m-b}{m} = 1 - \frac{b}{m}.
\end{displaymath}
The optimal strategies are as follows:
\begin{enumerate}
\item Strategy $\uovec{m}$ is the unique optimal strategy of Player $\mathrm{A}$.\label{p:hart2:1}
\item The strategies $(1-b/m)\uvec{0} + (b/m)\bm{z}$ with $\bm{z} \in \conv{\mathcal{U}^{m}}$ are optimal strategies of Player $\mathrm{B}$.\label{p:hart2:2}
\item If $b = (2m-1)/2$ and $m \geq 2$, then strategy $(1-b/m) \uvec{0} + (b/m)\bm{z}$ with $\bm{z} = m/(2m-1) \uovec{m} + (m-1)/(2m-1)\uoupvec{m}$ is optimal for Player $\mathrm{B}$.\label{p:hart2:3}
\item If $m - 1 \geq b > 0$, then the strategy $\bm{y}$ is optimal for Player $\mathrm{B}$ if and only if
\begin{displaymath}
\bm{y} = \left(1 - \frac{b}{m} \right) \uvec{0} + \left(\frac{b}{m}\right) \bm{z}, \textrm{ with $\bm{z} \in \conv{\mathcal{U}^m \cup \mathcal{W}^m \cup \{\uoupvec{m}\}}$.}
\end{displaymath}\label{p:hart2:4}
\item If $Y$ is a random variable distributed according to an optimal strategy of player B, then $\Pr(Y \geq 2m-2) > 0$.\footnote{
This part of the theorem is not directly stated in \cite{Hart08} or~\cite{Dziubinski12} but is an immediate consequence of the full characterization obtained in these works.}\label{p:hart2:5}
\end{enumerate}
\end{theorem}

Feasibility of optimal strategies for player A in the case of $\divides{K}{A}$ is fully characterized by point~\ref{p:hart:1} of Proposition~\ref{pr:hart}. It is necessary that $A$ and $K$ have the same parity.
In the case of player B, by point~\ref{p:hart2:5} of Theorem~\ref{th:hart2}, if an optimal strategy is $(B,K)$-feasible then $B \geq 2m-2$. 
In addition, if $B$ is even and $B \geq 2m$, where $m = A/K$ and $b = B/K$, then $(1 - b/m) \uvec{0} + (b/m) \uevec{m}$ is $(B,K)$-feasible, by Proposition~\ref{pr:aimpl3}.
In addition, since, for $m \geq 2$,
\begin{align*}
\left(1 - \frac{b}{m}\right)\uvec{0} + \frac{b}{m}\uoupvec{m} & = 
\left(1 - \frac{b}{m-1} + \frac{b}{m(m-1)}\right)\uvec{0} + \frac{b}{m(m-1)}\uoup{m}\\
& = \left(1 - \frac{b}{m-1}\right)\uvec{0} + \frac{b}{m(m-1)}\left(\uvec{0} + \uoup{m}\right)\\
& = \left(1 - \frac{b}{m-1}\right) + \frac{b}{m(m-1)}\left(\ue{m-1}\right)\\
& = \left(1 - \frac{b}{m-1}\right) + \frac{b}{m-1}\uevec{m-1}
\end{align*}
so, by Proposition~\ref{pr:aimpl3}, $(1 - b/m) \uvec{0} + (b/m) \uoupvec{m}$ is $(B,K)$-feasible 
if $B$ is even, $B \geq 2m-2$, and $m \geq 2$. Thus we have established that in the case of even $B$ there exist $(B,K)$-feasible optimal strategies for player B in the General Lotto game $\general(A/K,B/K)$ as long as either $B\geq 2A/K$ or $B\geq 2A/K-2$ and $A \geq 2K$.

For the case of odd $B$ we have the following two propositions. The first one of them establishes $(B,K)$-feasibility of an optimal strategy for player B in the case of $B = 2m-1$ and $m \geq 1$.

\begin{proposition}
\label{pr:aimpl5}
Let $m \geq 1$, $K \geq 2$ be natural numbers. If  $B = 2m-1$ then
\begin{equation*}
\bm{y} = \left(1-\frac{B}{Km}\right) \uvec{0} + \frac{B}{Km} \left(\frac{m}{B} \uovec{m} + \left(1-\frac{m}{B}\right)\uoupvec{m}\right)
\end{equation*}
is $(B,K)$-feasible.
\end{proposition}

Notice that in the case of $m = 1$, strategy $\bm{y}$ in Proposition~\ref{pr:aimpl5} is equal to
$(1 - b/m) \uvec{0} + (b/m) \uovec{m}$ and it is optimal by point~\ref{p:hart2:2} of Theorem~\ref{th:hart2}.
In the case of $K = 2$, $m \geq 2$, and $B = 2m-1$, we have $B/K = (2m-1)/2$ and, consequently, strategy $\bm{y}$ in Proposition~\ref{pr:aimpl5} is optimal by point~\ref{p:hart2:3} of Theorem~\ref{th:hart2}. 
In the case of $K \geq 2$, $m \geq 2$, and $B = 2m-1$, we have $B/K \leq m-1$ and, consequently, $\bm{y}$ is optimal by point~\ref{p:hart2:4} of Theorem~\ref{th:hart2}.

The second proposition establishes $(B,K)$-feasibility of an optimal strategy for player B in the case of odd $B$ such  that $B > 2m$.

\begin{proposition}
\label{pr:aimpl4}
Let $m \geq 1$, $K \geq 2$, and $2m < B \leq Km$ be natural numbers and let $b = B/K$. If $B$ is odd then
\begin{equation}
\bm{y} = \left(1-\frac{B}{Km}\right) \uvec{0} + \frac{B}{Km} \left(\frac{m}{B} \uovec{m} + \left(1-\frac{m}{B}\right)\uevec{m}\right)
\end{equation}
is $(B,K)$-feasible.
\end{proposition}

By Propositions~\ref{pr:aimpl5} and~\ref{pr:aimpl4}, for any $K \geq 2$, $B \geq 2A/K-2$, and $A \geq K$ such that $\divides{K}{A}$, there exists a $(B,K)$-feasible optimal strategy for the weaker player in the General Lotto game $\general(A/K,B/K)$. Hence we have established sufficient and necessary conditions on the values of $A$, $B$, and $K$, when $A > B$ and $\divides{K}{A}$, under which optimal strategies in the Colonel Blotto game $\blotto(A,B;K)$ correspond to optimal strategies in the General Lotto game. In consequence, we are able to determine the value of the Colonel Lotto and the Colonel Blotto games for these parameter values.

\begin{theorem}
\label{th:value:div}
Let $K\geq 2$, $A \geq K$, and $1 \leq B \leq A$ be natural numbers such that $\divides{K}{A}$. There exists an $(A,B,K)$-feasible profile of optimal strategies in General Lotto game $\general(A/K,B/K)$ if and only if $A \equiv K \pmod 2$ and $B \geq 2A/K-2$. The value of Colonel Lotto game $\lotto(A,B;K)$ and Colonel Blotto game $\blotto(A,B;K)$ is
\begin{equation}
\val{\blotto(A,B;K)} = \val{\lotto(A,B;K)} = \frac{A - B}{A} = 1 - \frac{B}{A}.
\end{equation}
\end{theorem}

When $\divides{A}{K}$ and $A$ has the same parity as $K$, the value of the game is equal to the value of the game in the continuous case. Divisibility of the number of units and the same parity of this number as the parity of the number of battlefields allows the stronger player to secure the same payoff as in the case of continuous resources.

\subsection{Beyond the General Lotto solvable cases}
\label{sec:beyond}
In this section we introduce a constrained variant of the General Lotto game that allows for solving additional cases of the Colonel Lotto and the Colonel Blotto game when $B \leq \lfloor A/K \rfloor$ and $\ndivides{K}{A}$. In this case, by Theorem~\ref{th:value:ndiv:even}, optimal strategies of player B in the General Lotto game $\general(A/K,B/K)$ require him to only assign even numbers to the battlefields with positive probability. When the number of units is odd, in any mixed strategy of the player, with probability $1$, at least one battlefield is assigned an odd number of units. Hence if $\bm{p}$ is the strategy in the General Lotto game associated with a mixed strategy of such a player, it satisfies $\sum_{i \geq 0} p_{2i+1} \geq 1/K$. This motivates introduction of the following constrained variant of the General Lotto game.

Given real numbers $a\in \mathbb{R}_{>0}$, $b\in \mathbb{R}_{>0}$, and $c\in \mathbb{R}_{>0}$, let $\widetilde{\general}(a,b;c)$ be a General Lotto game where the set of strategies of player B is restricted to probability distributions which pick an odd number with probability at least $c$ in total, that is
\begin{equation*}
S_{\widetilde{\general}}(b;c) = \left\{\bm{p} \in S_{\general}(b) : \sum_{i \geq 0} p_{2i+1} \geq c\right\}.
\end{equation*}

Since any strategy in the General Lotto game associated with a mixed strategy of player B with even number of units satisfies $\sum_{i \geq 0} p_{2i+1} \geq 1/K$, the following analogue of Observation~\ref{obs:2} follows.

\begin{observation}
\label{obs:3}
Let $B\geq 2$ be an even integer. 
If an equilibrium in the General Lotto game $\widetilde{\general}(A/K,B/K;1/K)$ is $(A,B,K)$-feasible then the implementing strategy profile is an equilibrium in the Colonel Lotto game $\lotto(A, B; K)$. In this case, every strategy implementing an optimal strategy of player A in $\widetilde{\general}(A/K,B/K;1/K)$ is an optimal strategy of A in $\lotto(A,B;K)$ and an analogous fact holds for the optimal strategies of player B. 
Moreover, in this case, the Colonel Lotto game, $\lotto(A, B; K)$, and the General Lotto game, $\general(A/K,B/K;1/K)$, have the same value.
\end{observation}

To solve the Colonel Blotto games $\blotto(A,B;K)$ with even $B \leq \lfloor A/K \rfloor$ and $\ndivides{K}{A}$, we first find the value and optimal strategies in the general Lotto games $\general(a,b;c)$
with non-integer $a>1$, $b < \lfloor a \rfloor$, and $c \in (0,b/(m+1))$.

\begin{proposition}
\label{pr:general:constr}
Let $a =m + \alpha$ and $b \leq m$, where $m \geq 1$ is an integer and $0 < \alpha < 1$, and $0 < c \leq b/(m+1)$. Then the value of the constrained General Lotto game $\widetilde{\general}(a,b;c)$ is
\begin{displaymath}
\val{\widetilde{\general}(a,b)} = 1 - \frac{(1-\alpha)b}{m} - \frac{\alpha b}{m+1} + \frac{c\max(\alpha,1-\alpha)}{m(m+1)}.
\end{displaymath}

Let $\delta$ and $\sigma$ be defined as in Theorem~\ref{th:hart4} and let $\vvvec{m} = \frac{1}{m}\sum_{i = 1}^{m} \vvec{i}{m}$.
\begin{enumerate}
\item If $0<\alpha < 1/2$ then strategy $\bm{x} = \alpha \delta \vvvec{m} + (1 - \alpha \delta) \uovec{m}$ is optimal for player A and strategy
\begin{equation*}
\bm{y} = \left(1 - \frac{b}{m}\right)\uvec{0} + \frac{b}{m}\left(\frac{cm}{b} \uovec{m} + \left(1-\frac{cm}{b}\right)\uevec{m}\right)
\end{equation*}
is optimal for player B.\label{p:general:constr:1}
\item If $1/2 \leq \alpha < 1$ then strategy 
\begin{equation*}
\bm{x} = \left(\frac{1-\alpha}{\alpha}\right) \left(\alpha \delta \vvvec{m} + (1 - \alpha \delta) \uovec{m}\right) + \left(\frac{2\alpha - 1}{\alpha}\right) \left( (1 - \alpha) \uovec{m} + \alpha \uovec{m+1} \right)
\end{equation*}
is optimal for player A and strategy
\begin{equation*}
\bm{y} = \left(1 - \frac{b-c}{m}\right)\uvec{0} + \frac{b-c}{m}\left(\frac{mc}{b-c} \uovec{m+1} + \left(1 - \frac{mc}{b-c}\right)\uevec{m}\right)
\end{equation*}
is optimal for player B.\label{p:general:constr:2}
\end{enumerate}
\end{proposition}

By Proposition~\ref{pr:aimpl2}, optimal strategies of player A in the constrained General Lotto game $\widetilde{\general}(A/K,B/K;1/K)$ are $(A,K)$-implementable for any $A > K \geq 2$ such that $K \nmid A$, except the cases of $K = 3$ and $A\in (7,13,19)$. Notice that these strategies are also optimal for player A when player B has an even number of units.

$(B,K)$-implementability of the optimal strategy of player B in the General Lotto game $\widetilde{\general}(A/K,B/K;1/K)$ in the case of $A \bmod K < K$ (strategy $\bm{y}$ defined in point~\ref{p:general:constr:1} of Proposition~\ref{pr:general:constr}), $K\geq 2$, and $2\lfloor A/K\rfloor < B \leq K\lfloor A/K\rfloor$ is established by Proposition~\ref{pr:aimpl4}. $(B,K)$-implementability of the optimal strategy of player B in the case of $A \bmod K \geq K$ (strategy $\bm{y}$ defined in point~\ref{p:general:constr:2} of Proposition~\ref{pr:general:constr}) is established in the proposition below.

\begin{proposition}
\label{pr:aimpl6}
Let $m \geq 1$, $K \geq 2$, and $2m+1 \leq B \leq Km$ be natural numbers and let $b = B/K$. If $B$ is odd then
\begin{equation}
\left(1-\frac{B-1}{Km}\right) \uvec{0} + \frac{B-1}{Km} \left(\left(\frac{m}{B-1}\right) \uovec{m+1} + \left(1-\frac{m}{B-1}\right)\uevec{m}\right)
\end{equation}
is $(B,K)$-feasible.
\end{proposition}

The results regarding optimal strategies obtained above allow us to complete the characterization of the value of the Colonel Blotto game $\blotto(A,B;K)$ for the cases where $2\lfloor A/K\rfloor < B \leq K \lfloor A/K\rfloor$, $\ndivides{K}{A}$, and $B$ odd. We state this in the theorem below.

\begin{theorem}
\label{th:value:ndiv:odd}
Let $K\geq 2$, $A > K$, and $2\lfloor A/K \rfloor < B \leq K\lfloor A/K \rfloor$ be natural numbers such that $\ndivides{K}{A}$, $B$ is odd, and either $K\neq 3$ or $A \notin \{7,13,19\}$. The value of Colonel Lotto game $\lotto(A,B;K)$ and Colonel Blotto game $\blotto(A,B;K)$ is
\begin{align*}
\val{\blotto(A,B;K)} & = \frac{A-B}{A} - \frac{B}{A}\left(\frac{R(K-R)}{(A-R)(A+K-R)}\right) + \frac{\max(R,K-R)}{(A-R)(A+K-R)},
\end{align*}
where $R = A\bmod K$.
\end{theorem}

The value of the game is equal to the value of the game in the continuous case, $(A-B)/A$, minus the reminder resulting from indivisibility of the number of units of the stronger player by the number of battlefields, already established in Theorem~\ref{th:value:ndiv:even}, plus the additional
quantity resulting from odd number of units of player B. When the number of units of player B
is sufficiently low,
\begin{equation*}
2\left\lfloor \frac{A}{K} \right\rfloor < B < A\frac{\max(R,K-R)}{R(K-R)},
\end{equation*}
where $R = A\bmod K$, then this quantity exceeds the loss due to the indivisibility of the number of units by the number of battlefields and payoff to player A is greater than in the continuous variant of the game.

\subsection{The case of $1 \leq  B \leq \lfloor A/K \rfloor$}
\label{sec:low}

In this case $A \geq K$, as $\lfloor A/K \rfloor \geq 1$. It is easy to see that when $B < \lfloor A/K \rfloor = m$ then
it is a dominant strategy for player A to assigned $m+1$ units at $A\bmod K$ battlefields and $m$ units at the remaining battlefields.
Player A wins at every battlefields regardless of the assignment chosen by player B. The value of the game is equal to $1$. 
If $B = \lfloor A/K \rfloor$ then player B is able to obtain payoff greater than $-1$ by assigning $B$ units to one battlefield, chosen uniformly at random, and assigning zero units to the remaining battlefields. Best response to this strategy by player A is choosing a set of $A\bmod K$ battlefields uniformly at random, assigning $m+1$ units there, and assigning $m$ units at the remaining battlefields (clearly this strategy is dominant in the case of $B < m$).
We characterize optimal strategies of the players and the value of the game in the proposition below.

\begin{proposition}
Let $A \geq K$ and $1 \leq  B \leq \lfloor A/K \rfloor$. Then the value of Colonel Blotto game $\blotto(A,B;K)$ is
\begin{equation*}
\val{\blotto(A,B;K)} = \begin{cases}
                        1 - \frac{1 - \frac{A\bmod K}{K}}{K} = \frac{A\bmod K}{K^2}, & \textrm{if $B = m$}\\
                        1, & \textrm{otherwise.}
                       \end{cases}
\end{equation*}
\begin{enumerate}
\item The mixed strategy which picks a subset of $K - A \bmod K$ out of $K$ battlefields uniformly at random and allocates 
$\lfloor A/K \rfloor$ units there and, in the case of $A \bmod K > 0$, allocates $\lceil A/K \rceil$ units in the remaining $A \bmod K$ battlefields is optimal for player A.
\item If $B = \lfloor A/K \rfloor$ then the mixed strategy which picks a battlefield uniformly at random and allocates $B$ units there and allocates $0$ units to the remaining battlefields is optimal for player B.
\item If $B < \lfloor A/K \rfloor$ then any assignments of units is an optimal stratetgy for player B.
\end{enumerate}
\end{proposition}

\begin{proof}
It is immediate to see that the strategy of player A is dominant in the case of $B < \lfloor A/K \rfloor$. Any response to that strategy of player B yields him the maximal possible payoff $0$. 

Suppose that $B = \lfloor A/K \rfloor$. Let $m = \lfloor A/K \rfloor$, $\alpha = \frac{A \bmod K}{K}$, and $b = B/K$. Consider the General Lotto game $\general(m+\alpha,b)$. The strategy in the General Lotto game that corresponds to the strategy of player A is
\begin{equation*}
\bm{x} = \frac{K-r}{K}\uvec{m} + \frac{r}{K}\uvec{m+1}
\end{equation*}
and the strategy that corresponds to the strategy of player B is
\begin{equation*}
\bm{y} = \frac{K-1}{K}\uvec{0} + \frac{r}{K}\uvec{m}.
\end{equation*}

Let $X$ be a non-negative integer valued random variable distributed with $\bm{x}$ and $Y$ be a non-negative integer valued random variable with $\Ex(Y) = m$ and $\Pr(Y > m) = 0$. Let $p_i = \Pr(Y = i)$. Then
\begin{align*}
H(X,Y) & = \frac{K-r}{K}H(\uvec{m},Y) + \frac{r}{K}H(\uvec{m+1},Y) = \frac{K-r}{K}\left(1 - p_m\right) + \frac{r}{K}\\
       & = 1 - \frac{K-r}{K}p_m,
\end{align*}
as  $\Pr(Y > m) = 0$ and, consequently, $H(\uvec{m+1},Y) = 1$. Since $\Ex(Y) = b = B/K = m/K$ so 
\begin{equation*}
p_m = \frac{\sum_{i\geq 1} ip_i - \sum_{\substack{i\geq 1 \\ i\neq m} ip_i}}{m} = \frac{\Ex(Y) - \sum_{\substack{i\geq 1 \\ i\neq m}} ip_i}{m} \leq \frac{b}{m} = \frac{1}{K}.
\end{equation*}
Hence 
\begin{align*}
H(X,Y) \geq 1 - \frac{K-r}{K}\frac{1}{K} = \frac{r}{K^2}.
\end{align*}

On the other hand, let $Y$ be a non-negative integer valued random variable distributed with $\bm{y}$ and $X$ be a non-negative integer valued random variable with $\Ex(X) = m + \alpha$. Let $q_i = \Pr(X = i)$. Then
\begin{align*}
H(Y,X) & = \frac{K-1}{K}\left(q_0 - 1\right) + \frac{1}{K}H(\uvec{m},X) = \frac{K-1}{K}\left(q_0 - 1\right) + \frac{1}{K}\left(\Pr(X < m) - \Pr(X > m)  \right)\\
       & = \frac{K-1}{K}\left(q_0 - 1\right) + \frac{1}{K}\left(1 - (\Pr(X \geq m) + \Pr(X \geq m+1))  \right),
\end{align*}
Since $\Ex(X) = m+\alpha$ so $\Pr(X \geq m) + \Pr(X \geq m+1)$ is bounded from above by
\begin{align*}
& \max\quad q_m + 2\sum_{i\geq m+1} q_i \quad \textrm{s.t.}\\
& \qquad \sum_{i \geq 1} iq_i = m+\alpha,\\
& \qquad \sum_{i \geq 0} q_i = 1,\\
& \qquad q_i \geq 0, \textrm{for all $i \geq 0$}.
\end{align*}
This is maximised by setting $q_i = 0$ for all $i \in \{1,\ldots,m-1\}$ and, since for all $i\geq m+2$, $q_i$ is multiplied by
a greater coefficient than $q_m$ and $q_{m+1}$ in the first constraint, by setting $q_i = 0$ for all $i \geq m+2$.
Given that the first and the second constraint yield $q_m = 1-\alpha$, $q_{m+1} = \alpha$, the solutions of the maximisation problem is $1 + \alpha$. Hence $\Pr(X \geq m) + \Pr(X \geq m+1) \leq 1 + \alpha$ and
\begin{align*}
H(Y,X) & \geq -\frac{K-1}{K} - \frac{\alpha}{K} = -1 + \frac{K-r}{K}\frac{1}{K} = -\frac{r}{K^2},
\end{align*}
with equality when $q_m = 1-\alpha$, $q_{m+1} = \alpha$, and $q_i = 0$, for all $i \notin \{m,m+1\}$.

This shows that the payoff that strategy $\bm{x}$ secures for player A is equal to minus the payoff that strategy $\bm{y}$ secures for player B. Since the game is zero-sum, this proves that $\bm{x}$ and $\bm{y}$ are optimal strategies of players A and B, respectively,
and that $r/K^2$ is the value of the game.
\end{proof}

%As we showed in Section~\ref{sec:genlot}, the General Lotto game does not allow to solve Colonel Blotto games $\blotto(A,B;K)$ when the number of units of player B is very low, $B < 2\lfloor A/K\rfloor - 2$.
%We conclude this section with a simple remark regarding the case where this number is very low.
%
%\begin{remark}
%When the number of units of player B is less than the number of units player A can secure at every battlefield, $B < \lfloor A/K \rfloor$, then any $K$-partition of resources of player A such that every part has at least $B+1$ units is a dominant strategy of player A and any $K$-partition of player B is a best response to that. The value of the game is equal to $1$.
%\end{remark}

\section{Discussion}
\label{sec:disc}

This paper provides explicit formulas for the value of the game as well as equilibrium strategies for a large range of parameters of the Colonel Blotto game in the asymmetric case of $B < A$. Explicit formulas for the value of the game help understanding the advangage/disadvantege related to the numbers and differences in resources as well as the marginal benefits of additional resources. In particular, the formula for the value of the game obtained in Theorems~\ref{th:value:ndiv:even}, \ref{th:value:ndiv:odd} and~\ref{th:value:ndiv:odd} reveals that divisibility of the number of resources by the number of battlefields is important for the stronger player and the parity of the number of resources is important for the weaker player. The inability of the stronger player to partition his resources evenly across the battlefields is beneficial to the weaker player, as long as he has an even number of resources. An odd number of resources of the weaker player is beneficial to the stronger player in this case. The formulas given in this theorems account exactly for these effects.
Explicit characterization of equilibrium strategies provides insight into marginal distributions of equilibrium strategies. From this we can conclude additional properties of optimal strategies, like maximal numbers of resources used by the players, the fractions of different numbers of resources used in equilibrium assignments, etc. 

Such qualitative properties of the value of the game and equilibrium strategies of the players would be hard two obtain by using numerical methods based on existing algorithms for solving the game. The value of the game can be approximated well by these algorithms, however, noticing the aforementioned effects of divisibility of the number resources by the number of battlefields or parity of the number of resources would require sampling the solutions of the game for a large space of game parameters. Uncovering the properties of equilibrium strategies is harder, because such numerical solutions allow us to obtain $\varepsilon$-equilibria, which do not have to be close to the actual equilibria of the game (c.f.~\cite{EtessamiYannakakis10} for the computational problem of finding strategy profiles which are $\varepsilon$-close to actual equilibria).

The matrices representing optimal strategies constructed in this paper are of the size at most $m(m+1) \times K$, where $m = \lfloor A/K \rfloor$. Hence these equilibrium strategies allow for relatively small representation. It is interesting to compare the size of these strategies with computational results regarding the Colonel Blotto game. The method of~\cite{BehnezhadDehghaniDerakhshanHajiaghayiSeddighin22} (which has the lowest proven bounds for the computation of solutions of the Colonel Blotto game) requires solving a linear program that has $\Theta(A^2/K)$ constraints and $\Theta(A^2/K)$ variables. They proved that this size is minimal in the sense that the extension complexity of the polytope of the optimal strategies of a player with $C$ units is $\Theta(C^2/K)$.\footnote{
It is important to note that the method of~\cite{BehnezhadDehghaniDerakhshanHajiaghayiSeddighin22} allows for solving Colonel Blotto in a much more general variant, where the values of battlefields are heterogeneous.}
Computing the optimal strategies requires solving this linear program. Solutions obtained in this paper show that in the case where the values of battlefields are homogeneous the problem can be solve in time that is proportional to the extension complexity of the polytope of optimal strategies of player A.

\bibliographystyle{te}
\bibliography{biblio}

\begin{thebibliography}{34}
\newcommand{\enquote}[1]{``#1''}
\providecommand{\natexlab}[1]{#1}
\providecommand{\url}[1]{\texttt{#1}}
\providecommand{\urlprefix}{URL }
\providecommand{\bibAnnoteFile}[1]{%
  \IfFileExists{#1}{\begin{quotation}\noindent\textsc{Key:} #1\\
  \textsc{Annotation:}\ \input{#1}\end{quotation}}{}}
\providecommand{\bibAnnote}[2]{%
  \begin{quotation}\noindent\textsc{Key:} #1\\
  \textsc{Annotation:}\ #2\end{quotation}}

\bibitem[{Ahmadinejad et~al.(2019)Ahmadinejad, Dehghani, Hajiaghayi, Lucier,
  Mahini, and Seddighin}]{AhmadinejadDehghaniHajiaghayiLucierMahiniSeddighin19}
Ahmadinejad, AM., S.~Dehghani, MT. Hajiaghayi, B.~Lucier, H.~Mahini, and
  S.~Seddighin (2019), \enquote{From duels to battlefields: {C}omputing
  equilibria of {B}lotto and other games.} \emph{Mathematics of Operations
  Research}, 44(4), 1304--1325.
\bibAnnoteFile{AhmadinejadDehghaniHajiaghayiLucierMahiniSeddighin19}

\bibitem[{Beaglehole et~al.(2023)Beaglehole, Hopkins, Kane, Liu, and
  Lovett}]{BeagleholeHopkinsKaneLiuLovett23}
Beaglehole, D., M.~Hopkins, D.~Kane, S.~Liu, and S.~Lovett (2023),
  \enquote{Sampling equilibria: Fast no-regret learning in structured games.}
  In \emph{Proceedings of the 2023 Annual ACM-SIAM Symposium on Discrete
  Algorithms (SODA)}, SODA'23, 3817--3855, Society for Industrial and Applied
  Mathematics, Philadelphia, PA.
\bibAnnoteFile{BeagleholeHopkinsKaneLiuLovett23}

\bibitem[{Beale and Heselden(1962)}]{BealeHeselden62}
Beale, E. and G.~Heselden (1962), \enquote{An approximate method of solving
  {B}lotto games.} \emph{Naval Research Logistics}, 9(2), 65--79.
\bibAnnoteFile{BealeHeselden62}

\bibitem[{Behnezhad et~al.(2022)Behnezhad, Dehghani, Derakhshan, Hajiaghayi,
  and Seddighin}]{BehnezhadDehghaniDerakhshanHajiaghayiSeddighin22}
Behnezhad, S., S.~Dehghani, M.~Derakhshan, M.~Hajiaghayi, and S.~Seddighin
  (2022), \enquote{Fast and simple solutions of {B}lotto games.}
  \emph{Operations Research}, 71(2), 506--516.
\bibAnnoteFile{BehnezhadDehghaniDerakhshanHajiaghayiSeddighin22}

\bibitem[{Bell and Cover(1980)}]{BellCover80}
Bell, R.~M. and T.~M. Cover (1980), \enquote{Competitive optimality of
  logarithmic investment.} \emph{Mathematics Of Operations Research}, 5(2),
  161--166.
\bibAnnoteFile{BellCover80}

\bibitem[{Blackett(1954)}]{Blackett54}
Blackett, D. (1954), \enquote{Some {B}lotto games.} \emph{Naval Research
  Logistics Quarterly}, 1(1), 55--60.
\bibAnnoteFile{Blackett54}

\bibitem[{Boix-Adser\`{a} et~al.(2021)Boix-Adser\`{a}, Edelman, and
  Jayanti}]{BoixAdseraEdelmanJayanti21}
Boix-Adser\`{a}, E., B.~Edelman, and S.~Jayanti (2021), \enquote{The
  multiplayer {C}olonel {B}lotto game.} \emph{Games and Economic Behavior},
  129, 15--31.
\bibAnnoteFile{BoixAdseraEdelmanJayanti21}

\bibitem[{Borel(1921)}]{Borel21}
Borel, E. (1921), \enquote{La {T}h{\'e}orie du {J}eu et les {{\'E}}quations
  {I}nt{\'e}grales {\`a} {N}oyau {S}ym{\'e}trique.} \emph{Comptes Rendus de
  l\textquoteright Acad{\'e}mie des Sciences}, 173, 1304--1308. {T}ranslated by
  {S}avage {LJ}, {T}he theory of play and integral equations with skew
  symmetric kernels. {E}conometrica 21 (1953) 97--100.
\bibAnnoteFile{Borel21}

\bibitem[{Borel and Ville(1938)}]{BorelVille38}
Borel, \'{E}. and J.~Ville (1938), \emph{Applications de la th\'{e}orie des
  probabilit\'{e}s aux jeux de hasard}, 1991 edition. J. Gabay, Paris.
\bibAnnoteFile{BorelVille38}

\bibitem[{Chia and Chuang(2011)}]{ChiaChuang11}
Chia, P.~H. and J.~Chuang (2011), \enquote{Colonel {B}lotto in the phishing
  war.} In \emph{Proceedings of the Second International Conference on Decision
  and Game Theory for Security}, GameSec'11, 201--218.
\bibAnnoteFile{ChiaChuang11}

\bibitem[{Dziubi{\'n}ski(2012)}]{Dziubinski12}
Dziubi{\'n}ski, M. (2012), \enquote{Non-symmetric discrete {G}eneral {L}otto
  games.} \emph{International Journal of Game Theory}, 1--33.
\bibAnnoteFile{Dziubinski12}

\bibitem[{Dziubi\'{n}ski(2017)}]{Dziubinski17}
Dziubi\'{n}ski, M. (2017), \enquote{The spectrum of equilibria for the colonel
  {B}lotto and the colonel {L}otto games.} In \emph{Algorithmic Game Theory}
  (V.~Bil{\`o} and M.~Flammini, eds.), 292--306, Springer International
  Publishing, Cham.
\bibAnnoteFile{Dziubinski17}

\bibitem[{Eisen and Le~Mat(1968)}]{EisenLeMat68}
Eisen, D. and M.~Le~Mat (1968), \enquote{Defense models, {XIII}: {BLOTTO}, a
  constrained matrix game solver.}
\bibAnnoteFile{EisenLeMat68}

\bibitem[{Etessami and Yannakakis(2010)}]{EtessamiYannakakis10}
Etessami, K. and M.~Yannakakis (2010), \enquote{On the complexity of nash
  equilibria and other fixed points.} \emph{SIAM Journal on Computing}, 39(6),
  2531--2597.
\bibAnnoteFile{EtessamiYannakakis10}

\bibitem[{Gross and Wagner(1950)}]{GrossWagner50}
Gross, O. and A.~Wagner (1950), \enquote{A continuous {C}olonel {Blotto} game.}
  {RM}-408, RAND Corporation, Santa Monica, CA.
\bibAnnoteFile{GrossWagner50}

\bibitem[{Guan et~al.(2020)Guan, Wang, Yao, Jiang, Han, and
  Ren}]{GuanWangYaoJiangHanRen20}
Guan, S., J.~Wang, H.~Yao, C.~Jiang, Z.~Han, and Y.~Ren (2020),
  \enquote{Colonel {B}lotto games in network systems: Models, strategies, and
  applications.} \emph{IEEE Transactions on Network Science and Engineering},
  7(2), 637--649.
\bibAnnoteFile{GuanWangYaoJiangHanRen20}

\bibitem[{Hart(2008)}]{Hart08}
Hart, S. (2008), \enquote{Discrete {C}olonel {B}lotto and {G}eneral {L}otto
  games.} \emph{International Journal of Game Theory}, 36, 441--460.
\bibAnnoteFile{Hart08}

\bibitem[{Kovenock and Roberson(2012)}]{KovenockRoberson12}
Kovenock, D. and B.~Roberson (2012), \enquote{Coalitional {C}olonel {B}lotto
  games with application to the economics of alliances.} \emph{Journal of
  Public Economic Theory}, 14(4), 653--676.
\bibAnnoteFile{KovenockRoberson12}

\bibitem[{Kovenock and Roberson(2021)}]{KovenockRoberson21}
Kovenock, D. and B.~Roberson (2021), \enquote{Generalizations of the {G}eneral
  {L}otto and {C}olonel {B}lotto games.} \emph{Economic Theory}, 71(3),
  997--1032.
\bibAnnoteFile{KovenockRoberson21}

\bibitem[{Laslier(2002)}]{Laslier02}
Laslier, J.-F. (2002), \enquote{How two-party competition treats minorities.}
  \emph{Review of Economic Design}, 7(3), 297--307.
\bibAnnoteFile{Laslier02}

\bibitem[{Laslier and Picard(2002)}]{LaslierPicard02}
Laslier, J.-F. and N.~Picard (2002), \enquote{Distributive politics and
  electoral competition.} \emph{Journal of Economic Theory}, 103(1), 106--130.
\bibAnnoteFile{LaslierPicard02}

\bibitem[{Liang et~al.(2023)Liang, Wang, Cao, and Yang}]{LiangWangCaoYang23}
Liang, D., Y.~Wang, Z.~Cao, and X.~Yang (2023), \enquote{Discrete colonel
  blotto games with two battlefields.} \emph{International Journal of Game
  Theory}, 52(4), 1111--–1151.
\bibAnnoteFile{LiangWangCaoYang23}

\bibitem[{Macdonell and Mastronardi(2015)}]{MacdonellMastronardi15}
Macdonell, S. and N.~Mastronardi (2015), \enquote{Waging simple wars: a
  complete characterization of two-battlefield {B}lotto equilibria.}
  \emph{Economic Theory}, 58(1), 183--216.
\bibAnnoteFile{MacdonellMastronardi15}

\bibitem[{Myerson(1993)}]{Myerson93}
Myerson, R.~B. (1993), \enquote{Incentives to cultivate favored minorities
  under alternative electoral systems.} \emph{The American Political Science
  Review}, 87(4), 856--869.
\bibAnnoteFile{Myerson93}

\bibitem[{Penn(1971)}]{Penn71}
Penn, A. (1971), \enquote{A generalized lagrange-multiplier method for
  constrained matrix games.} \emph{Operations Research}, 19(4), 933--945.
\bibAnnoteFile{Penn71}

\bibitem[{Powell(2009)}]{Powell09}
Powell, R. (2009), \enquote{Sequential, nonzero-sum ``{B}lotto'': {A}llocating
  defensive resources prior to attack.} \emph{Games and Economic Behavior},
  67(2), 611--615.
\bibAnnoteFile{Powell09}

\bibitem[{Roberson(2006)}]{Roberson06}
Roberson, B. (2006), \enquote{The {C}olonel {B}lotto game.} \emph{Economic
  Theory}, 29(1), 1--24.
\bibAnnoteFile{Roberson06}

\bibitem[{Roberson and Kvasov(2012)}]{RobersonKvasov12}
Roberson, B. and D.~Kvasov (2012), \enquote{The non-constant-sum {C}olonel
  {B}lotto game.} \emph{Economic Theory}, 51(2), 397--433.
\bibAnnoteFile{RobersonKvasov12}

\bibitem[{Shubik and Weber(1981)}]{ShubikWeber81}
Shubik, M. and R.~Weber (1981), \enquote{Systems defense games: Colonel
  {B}lotto, command and control.} \emph{Naval Research Logistics Quarterly},
  28(2), 281--287.
\bibAnnoteFile{ShubikWeber81}

\bibitem[{Thomas(2018)}]{Thomas18}
Thomas, C. (2018), \enquote{$n$-dimensional {B}lotto game with heterogeneous
  battlefield values.} \emph{Economic Theory}, 65(3), 509--544.
\bibAnnoteFile{Thomas18}

\bibitem[{Tukey(1949)}]{Tukey49}
Tukey, J.~W. (1949), \enquote{A problem of strategy.} \emph{Econometrica}, 17,
  73.
\bibAnnoteFile{Tukey49}

\bibitem[{Vu and Loiseau(2021)}]{VuLoiseau21}
Vu, D.~Q. and P.~Loiseau (2021), \enquote{{C}olonel {B}lotto games with
  favoritism: {C}ompetitions with pre-allocations and asymmetric
  effectiveness.} In \emph{Proceedings of the 22nd ACM Conference on Economics
  and Computation}, EC'21, 862--863, ACM, New York, NY, USA.
\bibAnnoteFile{VuLoiseau21}

\bibitem[{Washburn(2013)}]{Washburn13}
Washburn, A. (2013), \enquote{{OR} forum -- {B}lotto politics.}
  \emph{Operations Research}, 61(3), 532--543.
\bibAnnoteFile{Washburn13}

\bibitem[{Weinstein(2012)}]{Weinstein12}
Weinstein, J. (2012), \enquote{Two notes on the {B}lotto game.} \emph{The B.E.
  Journal of Theoretical Economics}, 12(1), 0000101515193517041893.
\bibAnnoteFile{Weinstein12}

\end{thebibliography}

\newpage

\begin{appendix}

\section{Proofs for the General Lotto solvable cases}

\subsection{Proofs for the case of $\ndivides{K}{A}$}

\begin{proof}[Proof of Proposition~\ref{pr:aimpl3}]
Let $\bm{x} = (1-b/m) \uvec{0} + (b/m)  \uevec{m}$. 

For the left to right implication notice that all values that appear with positive frequencies under $\bm{x}$ are even. Hence for any integer $K$, the sum of $K$ even numbers is even as well. Thus if $\bm{x}$ is $(B,K)$-feasible then $B$ is even. In addition, the value $2m$ appears with positive frequency under $\bm{x}$. Hence if $\bm{x}$ is $(B,K)$-feasible then $B\geq 2m$.

For the right to left implication suppose that $B$ is even and let $r = B \bmod m$ and $L = \lfloor B/m \rfloor$ so that $B = Lm + r$. In addition, since $B \geq 2m$ so $L\geq 2$ and since $B \leq Km$ so either $L < K$ or $L = K$ and $r = 0$.

Notice that in the case of $L = K$ and $r = 0$, $\bm{x} = \uevec{m}$ and, since $B$ is even, so by point~\ref{p:hart:2} of Proposition~\ref{pr:hart} $\bm{x}$ is $(B,K)$-feasible.

For the remaining part of the proof we assume that $L < K$. We start with rewriting $\bm{x}$ as follows:
\begin{equation*}
\bm{x} = \left(\frac{1}{Km(m+1)}\right) \left((m(K-L)-r)(m+1)  \uvec{0} + (Lm+r) \ue{m} \right).
\end{equation*}
Let $X^m = (m(K-L)-r)(m+1)  \uvec{0} + (Lm+r) \ue{m}$.
To show that $\bm{x}$ is $(B,K)$-feasible, we will construct $m(m+1) \times K$ matrices such that
$X^m$ is a cardinality vector for them. 

We consider two cases separately: (i)~$L\geq 3$ and $L$ odd and (ii)~$L\geq 2$ and $L$ even.

\noindent\textbf{Case~(i): $L \geq 3$ and $L$ odd.}

In this case $X^m$ can be rewritten as $X^m = (m-r)(m+1)  \uvec{0} + (3m+r) \ue{m} + m(m+1)(K-L-1)  \uvec{0} + (L-3)m\ue{m}$.
Since $L \geq 3$ and $L$ is odd so, by point~\ref{p:hart:2} of Proposition~\ref{pr:hart}, $(L-3)m\ue{m}$ is $((L-3)m,L-3)$-feasible and can be implemented by $m(m+1) \times (L-3)$ matrix $\sslash_{i = 1}^{m} |_{j = 1}^{\frac{L-3}{2}} \bm{E}(m)$.
Also, since $K > L$ so $m(m+1)(K-L-1)  \uvec{0}$ is $(0,K-L-1)$-feasible and can be implemented by $m(m+1) \times (K-L-1)$ matrix that contains zeros only.
Thus it is enough to construct a $m(m+1) \times 3$ matrix $\bm{S}(m)$ that $(3m+r,4)$-implements $Z(m) = (m-r)(m+1)  \uvec{0} + (3m+r) \ue{m}$.
Notice that $Z(m) = [(m-r)(m+1) + 3m+r,0,3m+r,0,3m+r,\ldots,0,3m+r]$, that is $Z_{2i+1} = 0$, for $i \in \{0,\ldots,m-1\}$, $Z_{0} = (m-r)(m+1) + 3m+r$, and $Z_{2i} = 3m+r$, for $i \in \{1,\ldots,m\}$.
Moreover, since $B$ is even and $L$ is odd so $r \equiv m \pmod 2$.

Let 
$\bm{S}(m,r) = \bm{S}^{\mathrm{I}}(m,r) \sslash \bm{S}^{\mathrm{II}}(m,r) \sslash \bm{S}^{\mathrm{III}}(m,r) \sslash \bm{S}^{\mathrm{IV}}(m,r) \sslash \bm{S}^{\mathrm{V}}(m,r) \sslash \left( \sslash_{j = 1}^{\frac{m-r+2}{2}} \bm{S}^{\mathrm{VI},j}(m,r) \right)$ $\sslash \left( \sslash_{j = 1}^{\frac{m-r+2}{2}} \bm{S}^{\mathrm{VII},j}(m,r) \right)$ be a $m(m+1) \times 4$ matrix consisting of $\frac{m}{2} + 3$ blocks, defined as follows.

A $(m+r-2) \times 4$ block
\begin{equation}
\bm{S}^{\mathrm{I}}(m,r) = \left[ \begin{array}{l c c c c l}
                             \{ & 2 & m+r-2 & 0 & 2m & \}^{2} \\
                             \{ & 4 & m+r-4 & 0 & 2m & \}^{2} \\
                             & \vdots & \vdots & \vdots& \vdots & \\
                             \{ & m+r-2 & 2 & 0 & 2m & \}^{2}
                       \end{array} \right],
\end{equation}
with the $i$'th part ($i = 1,\ldots,\frac{m+r-2}{2}$) of the form $\left[ \begin{array}{c c c c} 2i & m+r-2i & 0 & 2m \end{array}\right]$, repeated $2$ times.

A $\frac{(m+r-2)(m+r-4)}{8} \times 4$ block
\begin{equation}
\bm{S}^{\mathrm{II}}(m,r) = \left[ \begin{array}{l c c c c l}
                             \{ & 2 & m+r-2 & 2 & 2m-2 & \}^{\frac{m+r}{2}-2} \\
                             \{ & 4 & m+r-4 & 4 & 2m-4 & \}^{\frac{m+r}{2}-3} \\
                             & \vdots & \vdots & \vdots& \vdots & \\
                             \{ & m+r-4 & 4 & m+r-4 & m-r+4 & \}^{1}
                       \end{array} \right],
\end{equation}
with the $i$'th part ($i = 1,\ldots,\frac{m+r-4}{2}$) of the form $\left[ \begin{array}{c c c c} 2i & m+r-2i & 2i & 2m-2i \end{array}\right]$, repeated $\frac{m+r}{2}-i-1$ times.

A $\frac{(m+r-2)(m+r-4)}{8} \times 4$ block
\begin{equation}
\bm{S}^{\mathrm{III}}(m,r) = \left[ \begin{array}{l c c c c l}
                             \{ & m+r-2 & 2 & 2 & 2m-2 & \}^{\frac{m+r}{2}-2} \\
                             \{ & m+r-4 & 4 & 4 & 2m-4 & \}^{\frac{m+r}{2}-3} \\
                             & \vdots & \vdots & \vdots& \vdots & \\
                             \{ & 4 & m+r-4 & m+r-4 & m-r+4 & \}^{1}
                       \end{array} \right],
\end{equation}
with the $i$'th part ($i = 1,\ldots,\frac{m+r-4}{2}$) of the form $\left[ \begin{array}{c c c c} m+r-2i & 2i & 2i & 2m-2i \end{array}\right]$, repeated $\frac{m+r}{2}-i-1$ times.

A $\frac{(m-r+2)(m-r+4)}{8} \times 4$ block
\begin{equation}
\bm{S}^{\mathrm{IV}}(m,r) = \left[ \begin{array}{l c c c c l}
                             \{ & m+r & 2m & 0 & 0 & \}^{\frac{m-r}{2}+1} \\
                             \{ & m+r+2 & 2m-2 & 0 & 0 & \}^{\frac{m-r}{2}} \\
                             & \vdots & \vdots & \vdots& \vdots & \\
                             \{ & 2m & m+r & 0 & 0 & \}^{1}
                       \end{array} \right],
\end{equation}
with the $i$'th part ($i = 0,\ldots,\frac{m-r}{2}$) of the form $\left[ \begin{array}{c c c c} m+r+2i & 2m-2i & 0 & 0 \end{array}\right]$, repeated $\frac{m-r}{2}-i+1$ times.

A $\frac{(m-r+2)(m-r+4)}{8} \times 4$ block
\begin{equation}
\bm{S}^{\mathrm{V}}(m,r) = \left[ \begin{array}{l c c c c l}
                             \{ & 2m & m+r & 0 & 0 & \}^{\frac{m-r}{2}+1} \\
                             \{ & 2m-2 & m+r+2 & 0 & 0 & \}^{\frac{m-r}{2}} \\
                             & \vdots & \vdots & \vdots& \vdots & \\
                             \{ & m+r & 2m & 0 & 0 & \}^{1}
                       \end{array} \right],
\end{equation}
with the $i$'th part ($i = 0,\ldots,\frac{m-r}{2}$) of the form $\left[ \begin{array}{c c c c} 2m-2i & m+r+2i & 0 & 0 \end{array}\right]$, repeated $\frac{m-r}{2}-i+1$ times.

A $\frac{m-r+2}{2} \times 4$ block
\begin{equation}
\bm{S}^{\mathrm{VI,j}}(m,r) = \left[ \begin{array}{l c c c c l}
                             \{ & 2 & 2m-2j & 0 & m+r+2j-2 & \}^{1} \\
                             \{ & 4 & 2m-2j-2 & 0 & m+r+2j-2 & \}^{1} \\
                             & \vdots & \vdots & \vdots& \vdots & \\
                             \{ & m+r-2 & m-r-2j+4 & 0 & m+r+2j-2 & \}^{1}
                       \end{array} \right],
\end{equation}
with the $i$'th part ($i = 1,\ldots,\frac{m+r-2}{2}$) of the form $\left[ \begin{array}{c c c c} 2i & 2m-2j-2i+2 & 0 & m+r+2j-2 \end{array}\right]$, repeated $1$ time.

A $\frac{m-r+2}{2} \times 4$ block
\begin{equation}
\bm{S}^{\mathrm{VII,j}}(m,r) = \left[ \begin{array}{l c c c c l}
                             \{ & m+r+2j-2 & 0 & 2m-2j & 2  & \}^{1} \\
                             \{ & m+r+2j-2 & 0 & 2m-2j-2 & 4 & \}^{1} \\
                             & \vdots & \vdots & \vdots& \vdots & \\
                             \{ & m+r+2j-2 & 0 & m-r-2j+4 & m+r-2 & \}^{1}
                       \end{array} \right],
\end{equation}
with the $i$'th part ($i = 1,\ldots,\frac{m+r-2}{2}$) of the form $\left[ \begin{array}{c c c c} m+r+2j-2 & 0 & 2m-2j-2i+2 & 2i \end{array}\right]$, repeated $1$ time.

\noindent\textbf{Case~(ii): $L\geq 2$ and $L$ even.}

In this case $X^m$ can be rewritten as $X^m = (m-r)(m+1)  \uvec{0} + (2m+r) \ue{m} + m(m+1)(K-L-1)  \uvec{0} + (L-2)m\ue{m}$.
Since $L \geq 2$ and $L$ is even so, by point~\ref{p:hart:2} of Proposition~\ref{pr:hart}, $(L-2)m\ue{m}$ is $((L-2)m,L-2)$-feasible and can be implemented by $m(m+1) \times (L-2)$ matrix $\sslash_{i = 1}^{m} |_{j = 1}^{\frac{L-2}{2}} \bm{E}(m)$.
Also, since $K > L$ so $m(m+1)(K-L-1)  \uvec{0}$ is $(0,K-L-1)$-feasible and can be implemented by $m(m+1) \times (K-L-1)$ matrix that contains zeros only.
Thus it is enough to construct a $m(m+1) \times 3$ matrix $\bm{S}(m)$ that $(2m+r,3)$-implements $Z(m) = (m-r)(m+1)  \uvec{0} + (2m+r) \ue{m}$.
Notice that $Z(m) = [(m-r)(m+1) + 2m+r,0,2m+r,0,2m+r,\ldots,0,2m+r]$, that is $Z_{2i+1} = 0$, for $i \in \{0,\ldots,m-1\}$, $Z_{0} = (m-r)(m+1) + 2m+r$, and $Z_{2i} = 2m+r$, for $i \in \{1,\ldots,m\}$.
Moreover, since $B$ is even and $L$ is even so $r$ is even as well.

Let 
$\bm{T}(m,r) = \bm{T}^{\mathrm{I}}(m,r) \sslash \left( \sslash_{j=1}^{\frac{r}{2}-1} \bm{T}^{\mathrm{II},j}(m,r) \right) \sslash \left( \sslash_{j = 1}^{\frac{r}{2}-1} \bm{T}^{\mathrm{III},j}(m,r) \right)$ $\sslash \left( \sslash_{j = 1}^{\lfloor\frac{r}{4}\rfloor - 1} \bm{T}^{\mathrm{IV},j}(m,r) \right)$ $\sslash \left( \sslash_{j = 1}^{\lceil\frac{r}{4}\rceil - 1} \bm{T}^{\mathrm{V},j}(m,r) \right)\allowbreak \sslash \left( \sslash_{j = 1}^{\lceil\frac{r}{4}\rceil-1} \bm{T}^{\mathrm{VI},j}(m,r) \right) \sslash \left(\sslash_{j=\lceil\frac{r}{4}\rceil}^{\frac{r}{2}-2} \bm{T}^{\mathrm{VII},j}(m,r) \right)$ be a $m(m+1) \times 3$ matrix consisting of blocks defined as follows.

A $m(m-r+3) \times 3$ block (defined for $r \geq 2$)
\begin{equation}
\bm{T}^{\mathrm{I}}(m,r) = \left[ \begin{array}{l c c c l}
                             \{ & 0 & r & 2m & \}^{l(0)} \\
                             \{ & 0 & r+2 & 2m-2 & \}^{l(1)} \\
                             & \vdots & \vdots & \vdots & \\
                             \{ & 0 & 2m & r & \}^{l(m-\frac{r}{2})}
                       \end{array} \right],
\end{equation}
with the $i$'th part ($i = 0,\ldots,m-\frac{r}{2}$) of the form $\left[ \begin{array}{c c c} 0 & r+2i & 2m-2i \end{array}\right]$, repeated $l(i)$ times, where

\begin{equation*}
l(i) = \begin{cases}
       m-r+i+4, & \textrm{if $i \leq \frac{r}{2}-3$ and $i \leq m-r+2$}\\
       2m-2r+6, & \textrm{if $m-r+3 \leq i \leq \frac{r}{2}-3$}\\
       m-\frac{r}{2}+2, & \textrm{if $\frac{r}{2}-2 \leq i \leq m-r+2$}\\
       2m-\frac{3r}{2}-i+4, & \textrm{if $i \geq \frac{r}{2}-2$ and $i \geq m-r+3$}\\
       \end{cases}
\end{equation*}

A $(2m-r+2) \times 3$ block (defined for $r \geq 4$ and $j \in \{1,\ldots,\frac{r}{2}-1\}$
\begin{equation}
\bm{T}^{\mathrm{II,j}}(m,r) = \left[\begin{array}{l c c c l}
                             \{ & 2j & r & 2m-2j & \}^{2} \\
                             \{ & 2j & r+2 & 2m-2j-2 & \}^{2} \\
                             & \vdots & \vdots & \vdots & \\
                             \{ & 2j & 2m & r-2j & \}^{2}
                       \end{array} \right],
\end{equation}
with the $i$'th part ($i = 0,\ldots,m-\frac{r}{2}$) of the form $\left[ \begin{array}{c c c} 2j & r+2i & 2m-2j-2i \end{array}\right]$, each repeated $2$ times.

A $2\lceil\frac{j}{2}\rceil \times 3$ block (defined for $r \geq 4$ and $j \in \{1,\ldots,\frac{r}{2}-1\}$)
\begin{equation}
\bm{T}^{\mathrm{III},j}(m,r) = \left[ \begin{array}{l c c c l}
                             \{ & 2j & r-2j & 2m & \}^2 \\
                             \{ & 2j & r-2j+2 & 2m-2 & \}^2 \\
                             & \vdots & \vdots & \vdots & \\
                             \{ & 2j & r-2\lfloor\frac{j}{2}\rfloor-2 & 2m-2\lceil\frac{j}{2}\rceil+2 & \}^2
                       \end{array} \right],
\end{equation}
with the $i$'th part ($i = 0,\ldots,\lceil\frac{j}{2}\rceil-1$) of the form $\left[ \begin{array}{c c c} 2j & r-2j+2i & 2m-2i \end{array}\right]$, each repeated $2$ times.

A $(r-4j-2) \times 3$ block (defined for $r \geq 6$ and $j \in \{1,\ldots,\lfloor\frac{r}{4}\rfloor-1\}$)
\begin{equation}
\bm{T}^{\mathrm{IV},j}(m,r) = \left[ \begin{array}{l c c c l}
                             \{ & 2 & 2m+r-2j-2 & 2j & \}^2 \\
                             \{ & 4 & 2m+r-2j-4 & 2j & \}^2 \\
                             & \vdots & \vdots & \vdots & \\
                             \{ & 2r-4j-4 & 2m-r+2j+4 & 2j & \}^2
                       \end{array} \right],
\end{equation}
with the $i$'th part ($i = 0,\ldots,\frac{r}{2}-2j-2$) of the form $\left[ \begin{array}{c c c} r+2i & 2m-2j-2i & 2j \end{array}\right]$, each repeated $2$ times.

A $(r-4j) \times 3$ block (defined for $r \geq 6$ and $j \in \{1,\ldots,\lceil\frac{r}{4}\rceil-1\}$)
\begin{equation}
\bm{T}^{\mathrm{V},j}(m,r) = \left[ \begin{array}{l c c c l}
                             \{ & 2j & 2 & 2m+r-2j-2 & \}^2 \\
                             \{ & 2j & 4 & 2m+r-2j-4 & \}^2 \\
                             & \vdots & \vdots & \vdots & \\
                             \{ & 2j & 2r-4j-2 & 2m-r+2j+4 & \}^2
                       \end{array} \right],
\end{equation}
with the $i$'th part ($i = 0,\ldots,\frac{r}{2}-2j-1$) of the form $\left[ \begin{array}{c c c} 2j & r+2i & 2m-2j-2i \end{array}\right]$, each repeated $2$ times.

A $2j \times 3$ block (defined for $r \geq 6$ and $j \in \{1,\ldots,\lceil\frac{r}{4}\rceil-1\}$
\begin{equation}
\bm{T}^{\mathrm{VI,j}}(m,r) = \left[ \begin{array}{l c c c l}
                             \{ & r-2j & 2m & 2j & \}^{2} \\
                             \{ & r-2j+2 & 2m-2 & 2j & \}^{2} \\
                             & \vdots & \vdots & \vdots & \\
                             \{ & r-2 & 2m-2j+2 & 2j & \}^{2}
                       \end{array} \right],
\end{equation}
with the $i$'th part ($i = 0,\ldots,j-1$) of the form $\left[ \begin{array}{c c c} r-2j+2i & 2m-2i & 2j \end{array}\right]$, each repeated $2$ times.

A $(r-2j-2) \times 3$ block (defined for $r \geq 8$ and $j\in \{\lceil\frac{r}{4}\rceil,\ldots,\frac{r}{2}-2\}$
\begin{equation}
\bm{T}^{\mathrm{VII,j}}(m,r) = \left[ \begin{array}{l c c c l}
                             \{ & r-2j & 2m & 2j & \}^{2} \\
                             \{ & r-2j+2 & 2m-2 & 2j & \}^{2} \\
                             & \vdots & \vdots & \vdots & \\
                             \{ & 2r-4j-4 & 2m-r+2j+4 & 2j & \}^{2}
                       \end{array} \right],
\end{equation}
with the $i$'th part ($i = 0,\ldots,\frac{r}{2}-j-2$) of the form $\left[ \begin{array}{c c c} r-2j+2i & 2m-2i & 2j \end{array}\right]$, each repeated $2$ times.

\end{proof}

\begin{proof}[Proof of Proposition~\ref{pr:aimpl1}]
Let $\bm{x} = (1-\alpha) \uovec{m} + \alpha \uovec{m+1}$. 

For the left to right implication notice that all values that appear with positive frequencies under $\bm{x}$ are odd. Hence if $K$ is even, the sum of $K$ odd numbers is even as well. Thus if $\bm{x}$ is $(A,K)$-feasible then $A$ is even. On the other hand, if $K$ is odd then the sum of $K$ odd numbers is odd and so, if $\bm{x}$ is $(A,K)$-feasible then $A$ is odd. Hence $A \equiv K \pmod 2$

For the right to left implication, suppose that $A\equiv K \pmod 2$. Let $r = A \bmod K$, so that $\alpha = \frac{r}{K}$.
We start with rewriting $\bm{x}$ as follows:
\begin{equation}
\bm{x} = \left(\frac{1}{Km(m+1)}\right) \left((K-r)(m+1)\uo{m} + rm\uo{m+1} \right).
\end{equation}
Let $X^m = (K-r)(m+1)\uo{m} + rm\uo{m+1}$.
To show that $\bm{x}$ is $(A,K)$-feasible we will construct $m(m+1) \times K$ matrices such that
$X^m$ is a cardinality vector for them. We consider three cases separately: (i)~$K$ is even and $r$ is even, (ii)~$K$ is odd and $r$ is odd, and (iii)~$K$ is odd and $r$ is even.

\noindent\textbf{Case~(i): $K$ is even and $r$ is even.}

Since $K$ is even and $r$ is even, so is $K-r$. Moreover, $m(K-r) \equiv (K - r) \pmod 2$ and $(m+1)r \equiv r \pmod 2$. By point~\ref{p:hart:1} of Proposition~\ref{pr:hart}, $(K - r)(m+1)\uo{m}$ is $(m(K-r),K-r)$-feasible 
and can be implemented by $m(m+1) \times (K-r)$ matrix $S(m) = \sslash_{i = 1}^{m+1} |_{j = 1}^{\frac{K-r}{2}} \bm{O}(m)$. 
Similarly, by point~\ref{p:hart:1} of Proposition~\ref{pr:hart}, $rm\uo{m+1}$ is $((m+1)r,r)$-feasible 
and can be implemented by $m(m+1) \times r$ matrix $T(m) = \sslash_{i = 1}^{m} |_{j = 1}^{\frac{r}{2}} \bm{O}(m+1)$. Matrix $S(m) | T(m)$ implements $X^m$, which shows that $X^m$ is $(Km+r,K)$-feasible. Since $A = Km+r$ so it shows that $X^m$ is $(A,K)$-feasible.

\noindent\textbf{Case~(ii): $K$ is odd, $r$ is odd.}
Since $A = Km + r$ and $A\equiv K \pmod 2$ so $A$ is odd and $m$ is even.
Since $K$ is odd and $r$ is odd so $K-r$ is even. Moreover, since $1 \leq r \leq K-1$ so $K-r \geq 2$.
Rewrite $X^m$ as $X^m = 2(m+1)\uo{m} + m\uo{m+1} + (K-r-2)(m+1)\uo{m} + (r-1)m\uo{m+1}$.
Since $K-r-2$ is even so, by point~\ref{p:hart:1} of Proposition~\ref{pr:hart}, $(K-r-2)(m+1)\uo{m}$ is $(m(K-r-2),K-r-2)$-feasible and can be implemented by $m(m+1) \times (K-r-2)$ matrix $S(m) = \sslash_{i = 1}^{m+1} |_{j = 1}^{\frac{K-r}{2}-1} \bm{O}(m)$. Similarly, since $r-1$ is even so, by point~\ref{p:hart:1} of Proposition~\ref{pr:hart}, $(r-1)m\uo{m+1}$ is $((r-1)(m+1),r-1)$-feasible and can be implemented by $m(m+1) \times (r-1)r$ matrix $T(m) = \sslash_{i = 1}^{m} |_{j = 1}^{\frac{r-1}{2}} \bm{O}(m+1)$. Thus it is enough to construct a $m(m+1) \times 3$ matrix $\bm{R}(m)$ that $(3m+1,3)$-implements $Z^m = 2(m+1)\uo{m} + m\uo{m+1}$, and then
$\bm{R}(m) | S(m) | T(m)$ would $(Km+r,K)$-implement $X^m$.
Notice that $Z^m = [0,3m+2,0,3m+2,0,\ldots,3m+2,0,m]$, that is $Z_{2i+1} = 3m+2$, for $i \in \{0,\ldots,m-1\}$,
$Z_{2m+1} = m$, and  $Z_{2i} = 0$, for $i \in \{0,\ldots,m\}$.
Let $\bm{R}(m)$ be a $m(m+1) \times 3$ matrix consisting of four blocks, 
$\bm{R}(m) = \bm{R}^{\mathrm{I}}(m) \sslash \bm{R}^{\mathrm{II}}(m) \sslash \bm{R}^{\mathrm{III}}(m) \sslash \bm{R}^{\mathrm{IV}}(m) \sslash \bm{R}^{\mathrm{V}}(m) \sslash \bm{R}^{\mathrm{VI}}(m)$, defined below. 

A $\frac{m(m-2)}{2} \times 3$ block
\begin{equation}
\bm{R}^{\mathrm{I}}(m) = \left[ \begin{array}{l c c c r}
                             \{ & 1 & 2m-1 & m+1 & \}^{m-2} \\
                             \{ & 3 & 2m-5 & m+3 & \}^{m-2} \\
                             & \vdots & \vdots & \vdots & \\
                             \{ & m-1 & 3 & 2m-1 & \}^{m-2}
                       \end{array} \right],
\end{equation}
with the $i$'th part of the form $\left[ \begin{array}{c c c} 2i+1 & 2m - 4i - 1 & m + 2i + 1 \end{array}\right]$,
where $i \in \{0,\ldots,\frac{m}{2}-1\}$.

A $\frac{(m-2)^2}{2} \times 3$ block
\begin{equation}
\bm{R}^{\mathrm{II}}(m) = \left[ \begin{array}{l c c c l}
                             \{ & 3 & 2m-3 & m+1 & \}^{m-2} \\
                             \{ & 5 & 2m-7 & m+3 & \}^{m-2} \\
                             & \vdots & \vdots & \vdots & \\
                             \{ & m-1 & 5 & 2m-3 & \}^{m-2}
                       \end{array} \right],
\end{equation}
with the $i$'th part of the form $\left[ \begin{array}{c c c} 2i + 3 & 2m - 4i - 3 & m + 2i + 1 \end{array}\right]$,
where $i \in \{0,\ldots,\frac{m}{2}-2\}$. 

A $(m-2) \times 3$ block
\begin{equation}
\bm{R}^{\mathrm{III}}(m) = \left[ \begin{array}{l c c c l}
                             \{ & 3 & m-1 & 2m-1 & \}^2 \\
                             \{ & 5 & m-3 & 2m-1 & \}^2 \\
                             & \vdots & \vdots & \vdots & \\
                             \{ & m-1 & 3 & 2m-1 & \}^2
                       \end{array} \right],
\end{equation}
with the $i$'th part of the form $\left[ \begin{array}{c c c} 2i + 3 & m - 2i - 1 & 2m-1 \end{array}\right]$,
where $i \in \{0,\ldots,\frac{m}{2}-2\}$.

A $m \times 3$ block
\begin{equation}
\bm{R}^{\mathrm{IV}}(m) = \left[ \begin{array}{l c c c r}
                             \{ & 1 & 2m-1 & m+1 & \}^2 \\
                             \{ & 1 & 2m-3 & m+3 & \}^2 \\
                             & \vdots & \vdots & \vdots & \\
                             \{ & 1 & m+1 & 2m-1 & \}^2
                       \end{array} \right],
\end{equation}
with the $i$'th part of the form $\left[ \begin{array}{c c c} 1 & 2m - 2i - 1 & m + 2i + 1 \end{array}\right]$,
where $i \in \{0,\ldots,\frac{m}{2}-1\}$.

A $m \times 3$ block
\begin{equation}
\bm{R}^{\mathrm{V}}(m) = \left[ \begin{array}{l c c c l}
                             \{ & 1 & m-1 & 2m+1 & \}^2 \\
                             \{ & 3 & m-3 & 2m+1 & \}^2 \\
                             & \vdots & \vdots & \vdots & \\
                             \{ & m-1 & 1 & 2m+1 & \}^2
                       \end{array} \right],
\end{equation}
with the $i$'th part of the form $\left[ \begin{array}{c c c} 2i + 1 & m - 2i - 1 & 2m+1 \end{array}\right]$,
where $i \in \{0,\ldots,\frac{m}{2}-1\}$

A $m \times 3$ block
\begin{equation}
\bm{R}^{\mathrm{VI}}(m) = \left[ \begin{array}{l c c c r}
                             \{ & 1 & 2m-1 & m+1 & \}^2 \\
                             \{ & 1 & 2m-3 & m+3 & \}^2 \\
                             & \vdots & \vdots & \vdots & \\
                             \{ & 1 & m+1 & 2m-1 & \}^2
                       \end{array} \right],
\end{equation}
with the $i$'th part of the form $\left[ \begin{array}{c c c} 1 & 2m - 2i - 1 & m + 2i + 1 \end{array}\right]$,
where $i \in \{0,\ldots,\frac{m}{2}-1\}$.

Clearly, each row of $\bm{R}(m)$ sums up to $3m+1$. It is straightforward to verify that $\card(\bm{R}(m)) = Z^m$.

\noindent\textbf{Case~(iii): $K$ is odd, $r$ is even.}
Since $A = Km + r$ and $A\equiv K \pmod 2$ so $A$ is odd and $m$ is odd.
Since $K$ is odd, $r$ is even, and $r \leq K-1$ so $K-r$ is odd and $K - r \geq 1$. Moreover, since $r$ is even and $r \geq 1$ so $r \geq 2$.
Rewrite $X^m$ as $X^m = (m+1)\uo{m} + 2m\uo{m+1} + (K-r-1)(m+1)\uo{m} + (r-2)m\uo{m+1}$.
Since $K-r-1$ is even so, by point~\ref{p:hart:1} of Proposition~\ref{pr:hart}, $(K-r-1)(m+1)\uo{m}$ is $(m(K-r-1),K-r-1)$-feasible and can be implemented by $m(m+1) \times (K-r-1)$ matrix $S(m) = \sslash_{i = 1}^{m+1} |_{j = 1}^{\frac{K-r-1}{2}} \bm{O}(m)$. Similarly, since $r-2$ is even so, by point~\ref{p:hart:1} of Proposition~\ref{pr:hart}, $(r-2)m\uo{m+1}$ is $((r-2)(m+1),r-2)$-feasible and can be implemented by $m(m+1) \times (r-1)r$ matrix $T(m) = \sslash_{i = 1}^{m} |_{j = 1}^{\frac{r}{2}-1} \bm{O}(m+1)$. Thus it is enough to construct a $m(m+1) \times 3$ matrix $\bm{R}(m)$ that $(3m+2,3)$-implements $Z^m = (m+1)\uo{m} + 2m\uo{m+1}$, and then
$\bm{R}(m) | S(m) | T(m)$ would $(Km+r,K)$-implement $X^m$.
Notice that $Z^m = [0,3m+1,0,3m+1,0,\ldots,3m+1,0,2m]$, that is $Z_{2i+1} = 3m+1$, for $i \in \{0,\ldots,m-1\}$,
$Z_{2m+1} = 2m$, and  $Z_{2i} = 0$, for $i \in \{0,\ldots,m\}$.
Let $\bm{R}(m)$ be a $m(m+1) \times 3$ matrix consisting of four blocks, 
$\bm{R}(m) = \bm{R}^{\mathrm{I}}(m) \sslash \bm{R}^{\mathrm{II}}(m) \sslash \bm{R}^{\mathrm{III}}(m) \sslash \bm{R}^{\mathrm{IV}}(m)$, defined below. 

A $\frac{(m+1)(m-1)}{2} \times 3$ block
\begin{equation}
\bm{R}^{\mathrm{I}}(m) = \left[ \begin{array}{l c c c r}
                             \{ & 1 & 2m-1 & m+2 & \}^{m-1} \\
                             \{ & 3 & 2m-5 & m+4 & \}^{m-1} \\
                             & \vdots & \vdots & \vdots & \\
                             \{ & m & 1 & 2m+1 & \}^{m-1}
                       \end{array} \right],
\end{equation}
with the $i$'th part of the form $\left[ \begin{array}{c c c} 2i+1 & 2m - 4i - 1 & m + 2i + 1 \end{array}\right]$, where $i \{0,\ldots,\frac{m-1}{2}\}$.
 
A $\frac{(m-1)^2}{2} \times 3$ block
\begin{equation}
\bm{R}^{\mathrm{II}}(m) = \left[ \begin{array}{l c c c l}
                             \{ & 3 & 2m-3 & m+2 & \}^{m-1} \\
                             \{ & 5 & 2m-7 & m+4 & \}^{m-1} \\
                             & \vdots & \vdots & \vdots & \\
                             \{ & m & 3 & 2m-1 & \}^{m-1}
                       \end{array} \right],
\end{equation}
with the $i$'th part of the form $\left[ \begin{array}{c c c} 2i + 3 & 2m - 4i - 3 & m + 2i + 2 \end{array}\right]$, where $i \in \{0,\ldots,\frac{m-3}{2}\}$. 

A $(m-1) \times 3$ block
\begin{equation}
\bm{R}^{\mathrm{III}}(m) = \left[ \begin{array}{l c c c l}
                             \{ & 3 & m-2 & 2m+1 & \}^2 \\
                             \{ & 5 & m-4 & 2m+1 & \}^2 \\
                             & \vdots & \vdots & \vdots & \\
                             \{ & m & 1 & 2m+1 & \}^2
                       \end{array} \right],
\end{equation}
with the $i$'th part of the form $\left[ \begin{array}{c c c} 2i + 3 & m - 2i - 2 & 2m+1 \end{array}\right]$, where $i\in \{0,\ldots,\frac{m-3}{2}\}$.

A $(m+1) \times 3$ block
\begin{equation}
\bm{R}^{\mathrm{IV}}(m) = \left[ \begin{array}{l c c c r}
                             \{ & 1 & 2m-1 & m+2 & \}^2 \\
                             \{ & 1 & 2m-3 & m+4 & \}^2 \\
                             & \vdots & \vdots & \vdots & \\
                             \{ & 1 & m & 2m+1 & \}^2
                       \end{array} \right],
\end{equation}
with the $i$'th part of the form $\left[ \begin{array}{c c c} 1 & 2m - 2i - 1 & m + 2i + 2 \end{array}\right]$, where $i \in \{0,\ldots,\frac{m-1}{2}\}$.

Clearly, each row of $\bm{R}(m)$ sums up to $3m+2$. It is straightforward to verify that $\card(\bm{R}(m)) = Z^m$.
\end{proof}

\begin{proof}[Proof of Proposition~\ref{pr:aimpl2}]
Let $r = A \bmod K$, so that $\alpha = \frac{r}{K}$.

For point~\ref{p:aimpl2:1}, suppose that $\alpha \leq 1/2$, in which case $K \geq 2r$.
We start with rewriting $\bm{x}$ as follows:
\begin{equation}
\bm{x} = \left(\frac{1}{Km(m+1)}\right) \left(r\vv{m} + (K(m+1) - r(2m+1))\uo{m} \right),
\end{equation}
where $\vv{m} = \sum_{i = 1}^m \vd{i}{m}$. Let $X^m = r\vv{m} + (K(m+1) - r(2m+1))\uo{m}$.
To show that $\bm{x}$ is $(A,K)$-feasible we will construct $m(m+1) \times K$ matrices such that
$X^m$ is a cardinality vector for them. We consider five cases separately: (i)~either $K$ is even, or $K$ is odd and $K \neq 2r + 1$ and $m$ is odd, (ii)~$K = 2r+1$ and $m$ is odd, (iii)~$K$ is odd, $K\geq 5$,  $r=1$, and $m$ is even, (iv)~$K$ is odd, $K\geq 5$, $r\geq 2$, and $m$ is even, and (v)~$K$ is odd, $m$ is even, and $m \geq 8$. This covers all the cases of $K\geq 2$, $m = \lfloor A/K \rfloor$, and $r = A\bmod K$ except the case of $K = 3$ and $m \in \{2,4,6\}$.

\noindent\textbf{Case~(i): either $K$ is even, or $K$ is odd and $K \neq 2r + 1$ and $m$ is odd}

Rewrite $X^m$ as $X^m = r(\vv{m} + \uo{m}) + (K - 2r)(m+1)\uo{m}$. By the assumptions of this case, either $K - 2r = 0$ or $K - 2r \geq 2$.
Moreover, $m(K-2r) \equiv K - 2r \pmod 2$. By point~\ref{p:hart:1} of Proposition~\ref{pr:hart}, $(K - 2r)\uo{m}$ is $(m(K-2r),K-2r)$-feasible
and can be implemented by $m(m+1) \times (K-2r)$ matrix $\sslash_{i = 1}^{m+1} |_{j = 1}^{\frac{K}{2} - r} \bm{O}(m)$, if $K$ is even,
or $m(m+1) \times (K-2r)$ matrix $\sslash_{i = 1}^{m+1} \left(\bm{RO}(m) | \left(|_{i = 1}^{\frac{K-3}{2} - r} \bm{O}(m)\right)\right)$,
if $K$ is odd.
Thus it is enough to construct a $m(m+1) \times 2$ matrix $\bm{R}(m)$ that $(2m+1,2)$-implements $Z(m) = \vv{m} + \uo{m}$, and then
$|_{i = 1}^{r} \bm{R}(m)$ would $((2m+1)r,2r)$-implement $rZ(m)$.
Notice that $Z(m) = [0,2m,2,2m-2,4,\ldots,2,2m]$, that is $Z_{2i+1} = 2m - 2i$, for $i \in \{0,\ldots,m-1\}$,
and  $Z_{2i} = 2i$, for $i \in \{0,\ldots,m\}$.
Let $\bm{R}(m)$ be a $m(m+1) \times 2$ matrix defined as follows:
\begin{equation}
\bm{R}(m) = \left[ \begin{array}{l c c l}
                       \{ & 2 & 2m-1 & \}^{2} \\
                       \{ & 4 & 2m-3 & \}^{4} \\
                       & \vdots & \vdots & \\
                       \{ & 2m-2 & 3 & \}^{2m-2} \\
                       \{ & 2m & 1 & \}^{2m}
                       \end{array}
                 \right],
\end{equation}
with the $i$'th part of the form $\left[ \begin{array}{c c} 2i & 2m - 2i + 1 \end{array}\right]$. It is easy to see that $\bm{R}(m)$
$(2m+1,2)$-implements $Z(m)$. This completes proof of case~(i).

\noindent\textbf{Case~(ii): $K = 2r+1$ and $m$ is odd}

In this case $X^m$ can be rewritten as $X^m = \vv{m} + (m+2)\uo{m} + (r-1)(\vv{m} + \uo{m})$. As was shown above, $(r-1)(\vv{m} + \uo{m})$
is $((2m+1)(r-1), 2(r-1))$-feasible and can be implemented in a $m(m+1) \times 2(r-1)$ matrix $|_{i = 1}^{r-1} \bm{R}(m)$
Thus we need to show that $Z(m) = \vv{m} + (m+2)\uo{m}$ is $(3m+1,3)$-feasible and can be implemented in a $m(m+1) \times 3$ matrix.
Notice that $Z(m) = [0,3m+1,2,3m-1,4,\ldots,m+3,2m]$, that is $Z_{2i+1} = 3m - 2i + 1$, for $i \in \{0,\ldots,m-1\}$, and $Z_{2i} = 2i$, for $i \in \{0,\ldots,m\}$.
Let $\bm{S}(m) = \bm{S}^{\mathrm{I}}(m) \sslash \bm{S}^{\mathrm{II}}(m) \sslash \left( \sslash_{j = 0}^{\frac{m-1}{2}} \bm{S}^{\mathrm{III},j}(m) \right) \sslash \left( \sslash_{j = 0}^{\frac{m-3}{2}} \bm{S}^{\mathrm{IV},j}(m) \right)$
be a $m(m+1) \times 3$ matrix consisting of $m+1$ blocks, defined as follows.
A $\frac{(m-1)(m+1)}{4} \times 3$ block
\begin{equation}
\bm{S}^{\mathrm{I}}(m) = \left[ \begin{array}{l c c c l}
                             \{ & 2m-3 & m-2 & 6 & \}^2 \\
                             \{ & 2m-5 & m-4 & 10 & \}^4 \\
                             & \vdots & \vdots & \vdots & \\
                             \{ & m & 1 & 2m & \}^{m-1}
                       \end{array} \right],
\end{equation}
with the $i$'th part ($i = 1,\ldots,\frac{m-1}{2}$) of the form $\left[ \begin{array}{c c c} 2m-2i-1 & m - 2i & 4i+2 \end{array}\right]$, repeated $2i$ times.
A $\frac{(m-1)(m+1)}{4} \times 3$ block
\begin{equation}
\bm{S}^{\mathrm{II}}(m) = \left[ \begin{array}{l c c c l}
                             \{ & 2m-1 & m-2 & 4 & \}^2 \\
                             \{ & 2m-3 & m-4 & 8 & \}^4 \\
                             & \vdots & \vdots & \vdots & \\
                             \{ & m+2 & 1 & 2m-2 & \}^{m-1}
                       \end{array} \right],
\end{equation}
with the $i$'th part ($i = 1,\ldots,\frac{m-1}{2}$) of the form $\left[ \begin{array}{c c c} 2m-2i-1 & m - 2i & 4i \end{array}\right]$, repeated $2i$ times.
A $(m-2j+1) \times 3$ block
\begin{equation}
\bm{S}^{\mathrm{III},j}(m) = \left[ \begin{array}{l c c c l}
                             \{ & 2m-2j-1 & m-2j & 4j+2 & \}^2 \\
                             \{ & 2m-2j-1 & m-2j-2 & 4j+4 & \}^2 \\
                             & \vdots & \vdots & \vdots & \\
                             \{ & 2m-2j-1 & 1 & m+2j+1 & \}^2
                       \end{array} \right],
\end{equation}
with the $i$'th part ($i = 0,\ldots,\frac{m-1}{2}-j$) of the form $\left[ \begin{array}{c c c} 2m-2j-1 & m - 2j-2i & 4j+2i+2 \end{array}\right]$.
A $(m-2j-1) \times 3$ block
\begin{equation}
\bm{S}^{\mathrm{IV},j}(m) = \left[ \begin{array}{l c c c l}
                             \{ & m & 2j+1 & 2m-2j & \}^2 \\
                             \{ & m-2 & 2j+3 & 2m-2j & \}^2 \\
                             & \vdots & \vdots & \vdots & \\
                             \{ & 2j+3 & m-2 & 2m-2j & \}^2
                       \end{array} \right],
\end{equation}
with the $i$'th part ($i = 1,\ldots,\frac{m-1}{2}-j$) of the form $\left[ \begin{array}{c c c} m-2i+2 & 2j+2i-1 & 2m-2j \end{array}\right]$.
Clearly each row of $\bm{S}(m)$ sums up to $3m+1$ and it is straightforward to verify that $\bm{S}(m)$ $m(m+1) \times 3$ implements
$Z(m)$ (we provide the detailed calculations in the Appendix).

\noindent\textbf{Case~(iii): $K$ is odd, $K \geq 5$, $r = 1$, and $m$ is even}

In this case $X^m$ can be rewritten as $X^m = \vv{m} + (3m+4)\uo{m} + (K - 5)(m+1)\uo{m}$.
Moreover, since $K$ is odd so  $m(K-5) \equiv K - 5 \pmod 2$. By point~\ref{p:hart:1} of Proposition~\ref{pr:hart}, $(K - 5)\uo{m}$ is $(m(K-5),K-5)$-feasible and can be implemented by $m(m+1) \times (K-5)$ matrix $\sslash_{i = 1}^{m+1} |_{j = 1}^{\frac{K-5}{2}} \bm{O}(m)$.
Thus it is enough to construct a $m(m+1) \times 5$ matrix $\bm{S}(m)$ that $(5m+1,5)$-implements $Z(m) = \vv{m} + (3m+4)\uo{m}$.
Notice that $Z(m) = [0,5m+3,2,5m+1,4,\ldots,3m+5,2m]$, that is $Z_{2i+1} = 5m + 3 - 2i$, for $i \in \{0,\ldots,m-1\}$,
and  $Z_{2i} = 2i$, for $i \in \{0,\ldots,m\}$.
Let $\bm{S}(m) = \bm{S}^{\mathrm{I}}(m) \sslash \bm{S}^{\mathrm{II}}(m) \sslash \bm{S}^{\mathrm{III}}(m) \sslash \bm{S}^{\mathrm{IV}}(m) \sslash \bm{S}^{\mathrm{V}}(m) \sslash \bm{S}^{\mathrm{VI}}(m) \sslash \bm{S}^{\mathrm{VII}}(m) \sslash \bm{S}^{\mathrm{VIII}}(m)$ be a $m(m+1) \times 5$ matrix consisting of $8$ blocks, defined as follows:

A $m \times 5$ block
\begin{equation}
\bm{S}^{\mathrm{I}}(m) = \left[ \begin{array}{l c c c c c r}
                             \{ & 1 & 2m - 1 & 1 & 2m-1 & m+1 & \}^2 \\
                             \{ & 1 & 2m - 1 & 1 & 2m-3 & m+3 & \}^2 \\
                             & \vdots & \vdots & \vdots & \vdots & \vdots & \\
                             \{ & 1 & 2m - 1 & 1 & m+1 & 2m-1 & \}^2
                       \end{array} \right],
\end{equation}
with the $i$'th part ($i = 0,\ldots,(m-2)/2$) of the form $\left[ \begin{array}{c c c c c} 1 & 2m - 1 & 1 & 2m - 2i - 1 & m + 2i + 1 \end{array}\right]$, repeated $2$ times.

A $m \times 5$ block
\begin{equation}
\bm{S}^{\mathrm{II}}(m) = \left[ \begin{array}{l c c c c c l}
                             \{ & 2m & 1 & 1 & m-1 & 2m & \}^2 \\
                             \{ & 2m & 1 & 3 & m-3 & 2m & \}^2 \\
                             & \vdots & \vdots & \vdots & \vdots & \vdots & \\
                             \{ & 2m & 1 & m-1 & 1 & 2m & \}^2
                       \end{array} \right],
\end{equation}
with the $i$'th part ($i = 0,\ldots,(m-2)/2$) of the form $\left[ \begin{array}{c c c c c} 2m & 1 & 2i + 1 & m - 2i - 1 & 2m \end{array}\right]$, repeated $2$ times.

A $\frac{m^2}{4} \times 5$ block
\begin{equation}
\bm{S}^{\mathrm{III}}(m) = \left[ \begin{array}{l c c c c c r}
                             \{ & 2m-1 & 1 & 1 & 2m-1 & m+1 & \}^{1} \\
                             \{ & 2m-5 & 5 & 3 & 2m-5 & m+3 & \}^{3} \\
                             & \vdots & \vdots & \vdots & \vdots & \vdots & \\
                             \{ & 3 & 2m-3 & m-1 & 3 & 2m-1 & \}^{m-1}
                       \end{array} \right],
\end{equation}
with the $i$'th part ($i = 0,\ldots,(m-2)/2$) of the form $\left[ \begin{array}{c c c c c} 2m - 4i - 1 & 4i + 1 & 2i + 1 & 2m - 4i - 1 & m + 2i + 1 \end{array}\right]$, repeated $2i + 1$ times. 

A $\frac{(m-2)^2}{4} \times 5$ block
\begin{equation}
\bm{S}^{\mathrm{IV}}(m) = \left[ \begin{array}{l c c c c c r}
                             \{ & 2m-2 & 3 & 1 & 2m-2 & m+1 & \}^{m-3} \\
                             \{ & 2m-6 & 7 & 3 & 2m-6 & m+3 & \}^{m-5} \\
                             & \vdots & \vdots & \vdots & \vdots & \vdots & \\
                             \{ & 6 & 2m-5 & m-3 & 6 & 2m-3 & \}^{1}
                       \end{array} \right],
\end{equation}
with the $i$'th part ($i = 0,\ldots,(m-4)/2$) of the form $\left[ \begin{array}{c c c c c} 2m - 4i - 2 & 4i + 3 & 2i + 1 & 2m - 4i - 2 & m + 2i + 1 \end{array}\right]$, repeated $m - 2i - 3$ times. 

A $\frac{m(m-2)}{4} \times 3$ block
\begin{equation}
\bm{S}^{\mathrm{V}}(m) = \left[ \begin{array}{l c c c c c l}
                             \{ & 2m-3 & 3 & 3 & 2m-3 & m+1 & \}^{2} \\
                             \{ & 2m-7 & 7 & 5 & 2m-7 & m+3 & \}^{4} \\
                             & \vdots & \vdots & \vdots & \\
                             \{ & 5 & 2m-5 & m-1 & 5 & 2m-3 & \}^{m-2}
                       \end{array} \right],
\end{equation}
with the $i$'th part ($i = 0,\ldots,(m-4)/2$) of the form $\left[ \begin{array}{c c c c c} 2m - 4i - 3 & 4i + 3 & 2i + 3 & 2m - 4i - 3 & m + 2i + 1 \end{array}\right]$, repeated $2i+2$ times. 

A $\frac{(m-2)(m-4)}{4} \times 3$ block
\begin{equation}
\bm{S}^{\mathrm{VI}}(m) = \left[ \begin{array}{l c c c c c l}
                             \{ & 2m-4 & 5 & 3 & 2m-4 & m+1 & \}^{m-4} \\
                             \{ & 2m-8 & 9 & 5 & 2m-8 & m+3 & \}^{m-6} \\
                             & \vdots & \vdots & \vdots & \\
                             \{ & 8 & 2m-7 & m-3 & 8 & 2m-5 & \}^{2}
                       \end{array} \right],
\end{equation}
with the $i$'th part ($i = 0,\ldots,(m-6)/2$) of the form $\left[ \begin{array}{c c c c c} 2m - 4i - 4 & 4i + 5 & 2i + 3 & 2m - 4i - 4 & m + 2i + 1 \end{array}\right]$, repeated $m-2i-4$ times. 

A $(m-1) \times 3$ block
\begin{equation}
\bm{S}^{\mathrm{VII}}(m) = \left[ \begin{array}{l c c c c c l}
                             \{ & 2 & 2m-1 & m-1 & 2 & 2m-1 & \}^1\\
                             \{ & 4 & 2m-3 & m-3 & 4 & 2m-1 & \}^2 \\
                             & \vdots & \vdots & \vdots & \vdots & \vdots & \\
                             \{ & m-2 & m+3 & 3 & m-2 & 2m-1 & \}^2 \\
                             \{ & m & m+1 & 1 & m & 2m-1 & \}^2
                       \end{array} \right],
\end{equation}
with the $i$'th part ($i = 0,\ldots,(m-2)/2$) of the form $\left[ \begin{array}{c c c c c} 2i+2 & 2m-2i-1 & m-2i-1 & 2i+2 & 2m - 1 \end{array}\right]$, repeated $2$ times except the first one, which is repeated once. 

A $(m-2) \times 3$ block
\begin{equation}
\bm{S}^{\mathrm{VIII}}(m) = \left[ \begin{array}{l c c c c c r}
                             \{ & m+2 & m-1 & 1 & m+2 & 2m-3 & \}^2 \\
                             \{ & m+4 & m-3 & 1 & m+4 & 2m-5 & \}^2 \\
                             & \vdots & \vdots & \vdots & \\
                             \{ & 2m-2 & 3 & 1 & 2m-2 & m+1 & \}^2
                       \end{array} \right],
\end{equation}
with the $i$'th part ($i = 0,\ldots,(m-4)/2$) of the form $\left[ \begin{array}{c c c c c} m+2i+2 & m-2i-1 & 1 & m+2i+2 & 2m-2i-3 \end{array}\right]$, repeated $2$ times.

\noindent\textbf{Case~(iv): $K$ is odd, $K \geq 5$, $r \geq 2$ and $m$ is even}

In this case $X^m$ can be rewritten as 
$X^m = 2\vv{m} + (m+3)\uo{m} + (r-2)(\vv{m} + \uo{m}) + (K - 2r - 1)(m+1)\uo{m}$.
As was shown above, $(r-2)(\vv{m} + \uo{m})$ is $((2m+1)(r-2), 2(r-2))$-feasible and can be implemented in a $m(m+1) \times 2(r-2)$ matrix $|_{i = 1}^{r-2} \bm{R}(m)$. Moreover, by point~\ref{p:hart:1} of Proposition~\ref{pr:hart}, $(K - 2r - 1)\uo{m}$ is $(m(K-2r-1),K-2r-1)$-feasible and can be implemented by $m(m+1) \times (K-2r-1)$ matrix $\sslash_{i = 1}^{m+1} |_{j = 1}^{\frac{K-1}{2}-r} \bm{O}(m)$.
Thus it is enough to construct a $m(m+1) \times 5$ matrix $\bm{S}(m)$ that $(5m+2,5)$-implements $Z(m) = 2\vv{m} + (m+3)\uo{m}$, in the cases of $m \in \{2,4,6,8\}$.
Notice that $Z(m) = [0,5m+1,4,5m-3,8,\ldots,m+5,4m]$, that is $Z_{2i+1} = 5m + 1 - 2i$, for $i \in \{0,\ldots,m-1\}$,
and  $Z_{2i} = 4i$, for $i \in \{0,\ldots,m\}$.
Let $\bm{T}(m) = \bm{T}^{\mathrm{I}}(m) \sslash \bm{T}^{\mathrm{II}}(m) \sslash \bm{T}^{\mathrm{III}}(m) \sslash \bm{T}^{\mathrm{IV}}(m) \sslash \bm{T}^{\mathrm{V}}(m) \sslash \bm{T}^{\mathrm{VI}}(m) \sslash \bm{T}^{\mathrm{VII}}(m) \sslash \bm{T}^{\mathrm{VIII}}(m)$ be a $m(m+1) \times 5$ matrix consisting of $8$ blocks, defined as follows:

Block $\bm{T}^{q}(m)$, with $q \in \{\mathrm{I},\mathrm{VI},\mathrm{VII},\mathrm{VIII}\}$ is obtained from $\bm{S}^{q}(m)$, defined for case~(iii), by increasing each of the values in the second column by $1$. Block $\bm{T}^{\mathrm{II}}(m)$ is obtained from $\bm{S}^{\mathrm{II}}(m)$, defined for case~(iii), by increasing each of the values in the third column by $1$.
Block $\bm{T}^{q}(m)$, with $q \in \{\mathrm{III},\mathrm{IV},\mathrm{V}\}$, is obtained from $\bm{S}^{q}(m)$, defined for case~(iii), by increasing each of the values in the fifth column by $1$.
It is elementary to verify that $\bm{T}(m)$ $(5m+2,5)$-implements $Z(m)$ in the cases of $m \in \{2,4,6\}$. For the case of $m = 8$, we redefine $\bm{T}^{\mathrm{V}}(8)$, $\bm{T}^{\mathrm{VI}}(8)$, and $\bm{T}^{\mathrm{VIII}}(8)$ as follows:
\begin{equation}
\bm{T}^{\mathrm{V}}(8) = \left[ \begin{array}{l c c c c c l}
                             \{ & 13 & 3 & 3 & 13 & 10 & \}^{2} \\
                             \{ & 9 & 8 & 5 & 9 & 11 & \}^{4} \\
                             \{ & 5 & 11 & 7 & 5 & 14 & \}^{6}
                       \end{array} \right],
\end{equation}

\begin{equation}
\bm{T}^{\mathrm{VI}}(8) = \left[ \begin{array}{l c c c c c l}
                             \{ & 12 & 6 & 3 & 12 & 9 & \}^{4} \\
                             \{ & 8 & 9 & 5 & 8 & 12 & \}^{2}
                       \end{array} \right],
\end{equation}

\begin{equation}
\bm{T}^{\mathrm{VIII}}(8) = \left[ \begin{array}{l c c c c c r}
                             \{ & 14 & 4 & 1 & 14 & 9 & \}^2 \\
                             \{ & 12 & 5 & 1 & 12 & 12 & \}^2 \\
                             \{ & 10 & 8 & 1 & 10 & 13 & \}^2
                       \end{array} \right].
\end{equation}
The definitions of the remaining blocks of $\bm{T}(8)$ remain unchanged.
It is elementary to verify that $\bm{T}(8)$ $(5m+2,5)$-implements $Z(8)$.

\noindent\textbf{Case~(v): $K$ is odd, $m$ is even, and $m \geq 8$}

In this case $X^m$ can be rewritten as $X^m = \vv{m} + (m+2)\uo{m} + (K - 3)(m+1)\uo{m}$.
Moreover, since $K$ is odd so  $m(K-3) \equiv K - 3 \pmod 2$. By point~\ref{p:hart:1} of Proposition~\ref{pr:hart}, $(K - 3)\uo{m}$ is $(m(K-3),K-3)$-feasible and can be implemented by $m(m+1) \times (K-3)$ matrix $\sslash_{i = 1}^{m+1} |_{j = 1}^{\frac{K-3}{2}} \bm{O}(m)$.
Thus it is enough to construct a $m(m+1) \times 3$ matrix $\bm{S}(m)$ that $(3m+1,3)$-implements $Z(m) = \vv{m} + (m+2)\uo{m}$.
As we noticed earlier, $Z(m) = [0,3m+1,2,3m-1,4,\ldots,m+3,2m]$, that is $Z_{2i+1} = 3m - 2i + 1$, for $i \in \{0,\ldots,m-1\}$, and $Z_{2i} = 2i$, for $i \in \{0,\ldots,m\}$.
Let $\bm{S}(m) = \bm{P}(m) \sslash \bm{R}(m)$ where matrix $\bm{P}$ contains $m(m+1)/2$ rows with two even numbers and one odd number and matrix $\bm{R}$ contains $m(m+1)/2$ rows with three odd numbers. Let 
$P(m) = \left( \sslash_{j = 1}^{\frac{m}{2}-1} \bm{P}^{\mathrm{I},j}(m) \right) \sslash \bm{P}^{\mathrm{II}}(m) \sslash \bm{P}^{\mathrm{III}}(m) \sslash \bm{P}^{\mathrm{IV}}(m) \sslash \bm{P}^{\mathrm{V}}(m)$ be a $m(m+1)/2 \times 3$ matrix consisting of $\frac{m}{2} + 3$ blocks, defined as follows.

A $4j \times 3$ block
\begin{equation}
\bm{P}^{\mathrm{I},j}(m) = \left[ \begin{array}{l c c c l}
                             \{ & 2j+2 & 2m & m-2j-1 & \}^4 \\
                             \{ & 2j+4 & 2m-2j & m-2j-1 & \}^4 \\
                             & \vdots & \vdots & \vdots & \\
                             \{ & 4j & 2m-2j+2 & m-2j-1 & \}^4
                       \end{array} \right],
\end{equation}
with the $i$'th part ($i = 0,\ldots,j-1$) of the form $\left[ \begin{array}{c c c} 2j+2i+2 & 2m-2i & m-2j-1 \end{array}\right]$.

A $\frac{m}{2} \times 3$ block
\begin{equation}
\bm{P}^{\mathrm{II}}(m) = \left[ \begin{array}{l c c c l}
                             \{ & 6 & 2m & m-5 & \}^2 \\
                             \{ & 10 & 2m-2 & m-7 & \}^2 \\
                             & \vdots & \vdots & \vdots & \\
                             \{ & 4\lceil \frac{m}{4}\rceil-2 & 2m-2\lceil \frac{m}{4}\rceil + 4 & m-2\lceil\frac{m}{4}\rceil-1 & \}^{2} \\
                             \{ & 4\lceil \frac{m}{4}\rceil+2 & 2m-2\lceil \frac{m}{4}\rceil + 2 \rfloor & m-2\lceil\frac{m}{4}\rceil-3 & \}^{2 - \frac{m\bmod 4}{2}}
                       \end{array} \right],
\end{equation}
with the $i$'th part ($i = 0,\ldots,\lceil\frac{m}{4}\rceil-1$) of the form $\left[ \begin{array}{c c c} 6+4i & 2m - 2i & m-2i-5 \end{array}\right]$ and each but the last part is repeated $2$ times and the last part is repeated twice when $\divid{4}{m}$ and once when $\ndivid{4}{m}$.

A $\frac{m}{2} \times 3$ block
\begin{equation}
\bm{P}^{\mathrm{III}}(m) = \left[ \begin{array}{l c c c l}
                             \{ & m+2 & 2m-2 & 1 & \}^2 \\
                             \{ & m+4 & 2m-4 & 1 & \}^2 \\
                             & \vdots & \vdots & \vdots & \\
                             \{ & m+2\lceil \frac{m}{4}\rceil-2 & 2m-2\lceil \frac{m}{4}\rceil+2 & 1 & \}^{2} \\
                             \{ & m+2\lceil \frac{m}{4}\rceil & 2m-2\lceil \frac{m}{4}\rceil & 1 & \}^{2 - \frac{m\bmod 4}{2}}
                       \end{array} \right],
\end{equation}
with the $i$'th part ($i = 0,\ldots,\lceil\frac{m}{4}\rceil-1$) of the form $\left[ \begin{array}{c c c} m+2i+2 & 2m - 2i-2 & 1 \end{array}\right]$ and each but the last part is repeated $2$ times and the last part is repeated twice when $\divid{4}{m}$ and once when $\ndivid{4}{m}$.

A $\frac{m-4}{2} \times 3$ block
\begin{equation}
\bm{P}^{\mathrm{IV}}(m) = \left[ \begin{array}{l c c c l}
                             \{ & m+2 & 2m-2 & 1 & \}^2 \\
                             \{ & m+4 & 2m-6 & 3 & \}^2 \\
                             & \vdots & \vdots & \vdots & \\
                             \{ & m+2\lfloor \frac{m}{4}\rfloor-2 & 2m-4\lfloor \frac{m}{4}\rfloor+6 & 2\lfloor\frac{m}{4}\rfloor-3 & \}^{2} \\
                             \{ & m+2\lfloor \frac{m}{4}\rfloor & 2m-4\lfloor \frac{m}{4}\rfloor+2 & 2\lfloor\frac{m}{4}\rfloor-1 & \}^{\frac{m\bmod 4}{2}}
                       \end{array} \right],
\end{equation}
with the $i$'th part ($i = 0,\ldots,\lfloor\frac{m}{4}\rfloor-1$) of the form $\left[ \begin{array}{c c c} m+2i+2 & 2m - 4i-2 & 2i+1 \end{array}\right]$ and each but the last part is repeated $2$ times and the last part is repeated twice when $\divid{4}{m}$ and once when $\ndivid{4}{m}$.

A $2 \times 3$ block
\begin{equation}
\bm{P}^{\mathrm{V}}(m) = \left[ \begin{array}{l c c c l}
                             \{ & 2 & 2m & m-1 & \}^2
                       \end{array} \right].
\end{equation}

Let $R(m) = \left( \sslash_{j = 3}^{\frac{m}{2}} \bm{R}^{\mathrm{I},j}(m) \right) \sslash \left( \sslash_{j = 3}^{\frac{m}{2}} \bm{R}^{\mathrm{II},j}(m) \right) \sslash \left( \sslash_{j = 3}^{\frac{m}{2}} \bm{R}^{\mathrm{III},j}(m) \right) \sslash \left( \sslash_{j = 3}^{\frac{m}{2}} \bm{R}^{\mathrm{IV},j}(m) \right)  \sslash \bm{R}^{\mathrm{V}}(m)\sslash \bm{R}^{\mathrm{VI}}(m)$

\noindent $\sslash \bm{R}^{\mathrm{VII}}(m)\sslash \bm{R}^{\mathrm{VIII}}(m)\sslash \bm{R}^{\mathrm{IX}}(m)$ 
 be a $m(m+1)/2 \times 3$ matrix consisting of $2m-3$ blocks, defined as follows.

A $j \times 3$ block
\begin{equation}
\bm{R}^{\mathrm{I},j}(m) = \left[ \begin{array}{l c c c l}
                             \{ & m-2j+1 & m+1 & m+2j-1 & \} \\
                             \{ & m-2j+1 & m+3 & m+2j-3 & \} \\
                             & \vdots & \vdots & \vdots & \\
                             \{ & m-2j+1 & m+2j-1 & m+1 & \}
                       \end{array} \right],
\end{equation}
with the $i$'th part ($i = 0,\ldots,j-1$) of the form $\left[ \begin{array}{c c c} m-2j+1 & m+2i+1 & m+2j-2i-1 \end{array}\right]$.

A $(j+1) \times 3$ block
\begin{equation}
\bm{R}^{\mathrm{II},j}(m) = \left[ \begin{array}{l c c c l}
                             \{ & m-2j+3 & m-1 & m+2j-1 & \} \\
                             \{ & m-2j+2 & m+1 & m+2j-3 & \} \\
                             & \vdots & \vdots & \vdots & \\
                             \{ & m-2j+2 & m+2j-1 & m-1 & \}
                       \end{array} \right],
\end{equation}
with the $i$'th part ($i = 0,\ldots,j$) of the form $\left[ \begin{array}{c c c} m-2j+3 & m+2i-1 & m+2j-2i-1 \end{array}\right]$.

A $(j-1) \times 3$ block
\begin{equation}
\bm{R}^{\mathrm{III},j}(m) = \left[ \begin{array}{l c c c l}
                             \{ & m+2j-1 & m-2j+3 & m-1 & \} \\
                             \{ & m+2j-1 & m-2j+5 & m-3 & \} \\
                             & \vdots & \vdots & \vdots & \\
                             \{ & m+2j-1 & m-1 & m-2j+3 & \}
                       \end{array} \right],
\end{equation}
with the $i$'th part ($i = 0,\ldots,j-2$) of the form $\left[ \begin{array}{c c c} m+2j-1 & m-2j+2i+3 & m-2i-1 \end{array}\right]$.

A $(j-3) \times 3$ block
\begin{equation}
\bm{R}^{\mathrm{IV},j}(m) = \left[ \begin{array}{l c c c l}
                             \{ & m+2j-1 & m-2j+5 & m-3 & \} \\
                             \{ & m+2j-1 & m-2j+7 & m-5 & \} \\
                             & \vdots & \vdots & \vdots & \\
                             \{ & m+2j-1 & m-3 & m-2j+5 & \}
                       \end{array} \right],
\end{equation}
with the $i$'th part ($i = 1,\ldots,j-3$) of the form $\left[ \begin{array}{c c c} m+2j-1 & m-2j+2i+3 & m-2i-1 \end{array}\right]$.

A $(\lceil \frac{m}{4} \rceil-2) \times 3$ block
\begin{equation}
\bm{R}^{\mathrm{V}}(m) = \left[ \begin{array}{l c c c l}
                             \{ & 1 & m+1 & 2m-1 & \} \\
                             \{ & 3 & m+3 & 2m-5 & \} \\
                             & \vdots & \vdots & \vdots & \\
                             \{ & 2\lceil\frac{m}{4}\rceil-5 & m+2\lceil\frac{m}{4}\rceil-5 & 2m-4\lceil\frac{m}{4}\rceil+11 & \}
                       \end{array} \right],
\end{equation}
with the $i$'th part ($i = 0,\ldots,\lceil \frac{m}{4} \rceil-3$) of the form $\left[ \begin{array}{c c c} 2i+1 & m+2i+1 & 2m-4i-1 \end{array}\right]$.

A $(\lfloor\frac{m}{4}\rfloor-2) \times 3$ block
\begin{equation}
\bm{R}^{\mathrm{VI}}(m) = \left[ \begin{array}{l c c c l}
                             \{ & 1 & m+3 & 2m-3 & \} \\
                             \{ & 3 & m+5 & 2m-7 & \} \\
                             & \vdots & \vdots & \vdots & \\
                             \{ & 2\lfloor\frac{m}{4}\rfloor-5 & m+2\lfloor\frac{m}{4}\rfloor-3 & 2m-4\lfloor\frac{m}{4}\rfloor+9 & \}
                       \end{array} \right],
\end{equation}
with the $i$'th part ($i = 0,\ldots,\lfloor\frac{m}{4}\rfloor-3$) of the form $\left[ \begin{array}{c c c} 2i+1 & m+2i+3 & 2m-4i-3 \end{array}\right]$.

A $(\lceil\frac{m}{4}\rceil+1) \times 3$ block
\begin{equation}
\bm{R}^{\mathrm{VII}}(m) = \left[ \begin{array}{l c c c l}
                             \{ & m-1 & 2m-1 & 3 & \} \\
                             \{ & m-3 & 2m-3 & 7 & \} \\
                             & \vdots & \vdots & \vdots & \\
                             \{ & m-2\lceil\frac{m}{4}\rceil-1 & 2m-2\lceil\frac{m}{4}\rceil-1 & 4\lceil\frac{m}{4}\rceil+3 & \}
                       \end{array} \right],
\end{equation}
with the $i$'th part ($i = 0,\ldots,\lceil\frac{m}{4}\rceil$) of the form $\left[ \begin{array}{c c c} m-2i-1 & 2m-2i-1 & 4i+3 \end{array}\right]$.

A $(\lfloor\frac{m}{4}\rfloor+2) \times 3$ block
\begin{equation}
\bm{R}^{\mathrm{VIII}}(m) = \left[ \begin{array}{l c c c l}
                             \{ & m+1 & 2m-1 & 1 & \} \\
                             \{ & m-1 & 2m-3 & 5 & \} \\
                             & \vdots & \vdots & \vdots & \\
                             \{ & m-2\lfloor\frac{m}{4}\rfloor-1 & 2m-2\lfloor\frac{m}{4}\rfloor-3 & 4\lfloor\frac{m}{4}\rfloor+5 & \}
                       \end{array} \right],
\end{equation}
with the $i$'th part ($i = 0,\ldots,\lfloor\frac{m}{4}\rfloor+1$) of the form $\left[ \begin{array}{c c c} m-2i+1 & 2m-2i-1 & 4i+1 \end{array}\right]$.

A $7 \times 3$ block
\begin{equation}
\bm{R}^{\mathrm{IX}}(m) = \left[ \begin{array}{l c c c l}
                             \{ & m-1 & m-1 & m+3 & \}^2 \\
                             \{ & m-1 & m+1 & m+1 & \}^2 \\
                             \{ & m-3 & m+1 & m+3 & \}^2 \\
                             \{ & m-3 & m-3 & m+7 & \}
                       \end{array} \right].
\end{equation}

For point~\ref{p:aimpl2:2} we start by rewriting $\bm{y}$ as follows:
\begin{equation}
\bm{y} = \left(\frac{1}{Km(m+1)}\right) \left((K-r)(\vv{m} + \uo{m}) + (2r-K)m\uo{m+1} \right).
\end{equation}
Let $Y^m = (K-r)(\vv{m} + \uo{m}) + (2r-K)\uo{m+1}$. Similarly to point~\ref{p:aimpl2:1}, to show that $\bm{y}$ is $(A,K)$-feasible we construct
$m(m+1) \times K$ matrices such that $Y^m$ is a cardinality vector for them. We consider three cases separately: 
(i)~$(m+1)K \equiv K \pmod 2$ and $K \neq 2r - 1$, (ii)~$K = 2r-1$ and $m$ is even, and~(iii)~$K$ is odd and $m$ is odd.

\noindent\textbf{Case~(i): $(m+1)K \equiv K \pmod 2$ and $2r \neq K + 1$}

By the assumptions of this case, either $2r - K = 0$ or $2r - K \geq 2$.
Moreover, $(m+1)(2r-K) \equiv (2r-K) \pmod 2$. By point~\ref{p:hart:1} of Proposition~\ref{pr:hart}, $(2r-K)\uo{m+1}$ is $((m+1)(2r-K),2r-K)$-feasible
and can be implemented by $m(m+1) \times (2r-K)$ matrix $\sslash_{i = 1}^{m} |_{j = 1}^{\frac{K}{2} - r} \bm{O}(m+1)$, if $K$ is even,
or $m(m+1) \times (K-2r)$ matrix $\sslash_{i = 1}^{m} \left(\bm{RO}(m+1) | \left(|_{i = 1}^{\frac{K-3}{2} - r} \bm{O}(m+1)\right)\right)$,
if $K$ is odd. As was shown in case~(i) of point~\ref{p:aimpl2:1}, $(K-r)(\vv{m} + \uo{m})$ is $((2m+1)r,2r)$-implemented by $|_{i = 1}^{K-r} \bm{R}(m)$.
This completes proof of case~(i).

In the remaining two cases we rewrite $X^m$ as $X^m = \vv{m} + \uo{m} + m\uo{m+1} + (K-r-1)(\vv{m} + \uo{m})$. As was shown above, $(K - r-1)(\vv{m} + \uo{m})$
is $((2m+1)(K-r-1), 2(K-r-1))$-feasible and can be implemented in a $m(m+1) \times 2(K-r-1)$ matrix $|_{i = 1}^{K-r-1} \bm{R}(m)$.
Thus we need to show that $Z(m) = \vv{m} + \uo{m} + m\uo{m+1}$ is $(3m+2,3)$-feasible and can be implemented in a $m(m+1) \times 3$ matrix.
Notice that $Z(m) = [0,3m,2,3m-2,4,\ldots,m+2,2m,m]$, that is $Z_{2i+1} = 3m - 2i$ and $Z_{2i} = 2i$.

\noindent\textbf{Case~(ii): $2r = K + 1$ and $m$ is even}
Let $\bm{S}(m) = \bm{S}^{\mathrm{I}}(m) \sslash \bm{S}^{\mathrm{II}}(m) \sslash \left( \sslash_{j = 0}^{\frac{m}{2}-2} \bm{S}^{\mathrm{III},j}(m) \right) \sslash \left( \sslash_{j = 0}^{\frac{m}{2}-1} \bm{S}^{\mathrm{IV},j}(m) \right)$
be a $m(m+1) \times 3$ matrix consisting of $m+1$ blocks, defined as follows.
A $\frac{m(m+2)}{4} \times 3$ block
\begin{equation}
\bm{S}^{\mathrm{I}}(m) = \left[ \begin{array}{l c c c l}
                             \{ & 2m+1 & m-1 & 2 & \}^2 \\
                             \{ & 2m-1 & m-3 & 6 & \}^4 \\
                             & \vdots & \vdots & \vdots & \\
                             \{ & m+3 & 1 & 2m-2 & \}^{m}
                       \end{array} \right],
\end{equation}
with the $i$'th part ($i = 1,\ldots,\frac{m}{2}$) of the form $\left[ \begin{array}{c c c} 2m-2i+3 & m-2i+1 & 4i-2 \end{array}\right]$, repeated $2i$ times.
A $\frac{m(m+2)}{4} \times 3$ block
\begin{equation}
\bm{S}^{\mathrm{II}}(m) = \left[ \begin{array}{l c c c l}
                             \{ & 2m-1 & m-1 & 4 & \}^2 \\
                             \{ & 2m-3 & m-3 & 8 & \}^4 \\
                             & \vdots & \vdots & \vdots & \\
                             \{ & m+1 & 1 & 2m & \}^{m}
                       \end{array} \right],
\end{equation}
with the $i$'th part ($i = 1,\ldots,\frac{m}{2}$) of the form $\left[ \begin{array}{c c c} 2m-2i+1 & m-2i+1 & 4i \end{array}\right]$, repeated $2i$ times.
A $(m-2j-2) \times 3$ block
\begin{equation}
\bm{S}^{\mathrm{III},j}(m) = \left[ \begin{array}{l c c c l}
                             \{ & 2m-2j+1 & m-2j-3 & 4j+4 & \}^2 \\
                             \{ & 2m-2j+1 & m-2j-5 & 4j+6 & \}^2 \\
                             & \vdots & \vdots & \vdots & \\
                             \{ & 2m-2j+1 & 1 & m+2j & \}^2
                       \end{array} \right],
\end{equation}
with the $i$'th part ($i = 1,\ldots,\frac{m}{2}-j-1$) of the form $\left[ \begin{array}{c c c} 2m-2j+1 & m-2j-2i-1 & 4j+2i+2 \end{array}\right]$.
A $(m-2j) \times 3$ block
\begin{equation}
\bm{S}^{\mathrm{IV},j}(m) = \left[ \begin{array}{l c c c l}
                             \{ & m+1 & 2j+1 & 2m-2j & \}^2 \\
                             \{ & m-1 & 2j+3 & 2m-2j & \}^2 \\
                             & \vdots & \vdots & \vdots & \\
                             \{ & 2j+3 & m-1 & 2m-2j & \}^2
                       \end{array} \right],
\end{equation}
with the $i$'th part ($i = 0,\ldots,\frac{m}{2}-j-1$) of the form $\left[ \begin{array}{c c c} m-2i+1 & 2j+2i+1 & 2m-2j \end{array}\right]$.
Clearly each row of $\bm{S}(m)$ sums up to $3m+2$ and it is straightforward to verify that $\bm{S}(m)$ $m(m+1) \times 3$ implements $Z(m)$.

\noindent\textbf{Case~(iii): $K$ is odd and $m$ is odd}

Let $\bm{S}(m) = \bm{S}^{\mathrm{I}}(m) \sslash \left( \sslash_{j = 1}^{\frac{m-1}{2}} \bm{S}^{\mathrm{II},j}(m) \right) \sslash \bm{S}^{\mathrm{III}}(m) \sslash \left( \sslash_{j = 0}^{\frac{m-3}{2}} \bm{S}^{\mathrm{IV},j}(m) \right) \sslash \left( \sslash_{j = 0}^{\frac{m-3}{2}} \bm{S}^{\mathrm{V},j}(m) \right)$ be a $m(m+1) \times 3$ matrix consisting of $(3m+1)/2$ blocks, defined as follows.
A $\frac{m+1}{2} \times 3$ block
\begin{equation}
\bm{S}^{\mathrm{I}}(m) = \left[ \begin{array}{l c c c l}
                             \{ & 2m & 2 & m & \} \\
                             \{ & 2m-4 & 6 & m & \} \\
                             & \vdots & \vdots & \vdots & \\
                             \{ & 2 & 2m & m & \}
                       \end{array} \right],
\end{equation}
with the $i$'th part ($i = 0,\ldots,\frac{m-1}{2}$) of the form $\left[ \begin{array}{c c c} 2m-4i & 4i+2 & m \end{array}\right]$.
A $(2m-4j+2) \times 3$ block
\begin{equation}
\bm{S}^{\mathrm{II},j}(m) = \left[ \begin{array}{l c c c l}
                             \{ & 2m-2j+2 & 4j & m-2j & \}^4 \\
                             \{ & 2m-2j+2 & 4j+2 & m-2j-2 & \}^4 \\
                             & \vdots & \vdots & \vdots & \\
                             \{ & 2m-2j+2 & m+2j-1 & 1 & \}^4
                       \end{array} \right],
\end{equation}
with the $i$'th part ($i = 0,\ldots,\frac{m-1}{2}-j$) of the form $\left[ \begin{array}{c c c} 2m-2j+2 & 4j+2i & m-2j-2i \end{array}\right]$, repeated $2i$ times.
A $\frac{m+1}{2} \times 3$ block
\begin{equation}
\bm{S}^{\mathrm{III}}(m) = \left[ \begin{array}{l c c c l}
                             \{ & 2m+1 & 1 & m & \} \\
                             \{ & 2m-1 & 3 & m & \} \\
                             & \vdots & \vdots & \vdots & \\
                             \{ & m+2 & m & m & \}
                       \end{array} \right],
\end{equation}
with the $i$'th part ($i = 0,\ldots,\frac{m-1}{2}$) of the form $\left[ \begin{array}{c c c} 2m-2i+1 & 2i+1 & m \end{array}\right]$.
A $(m-2j-1) \times 3$ block
\begin{equation}
\bm{S}^{\mathrm{IV},j}(m) = \left[ \begin{array}{l c c c l}
                             \{ & 2m-2j+1 & m & 2j+1 & \}^2 \\
                             \{ & 2m-2j+1 & m-2 & 2j+3 & \}^2 \\
                             & \vdots & \vdots & \vdots & \\
                             \{ & 2m-2j+1 & 2j+3 & m-2 & \}^2
                       \end{array} \right],
\end{equation}
with the $i$'th part ($i = 0,\ldots,\frac{m-3}{2}-j$) of the form $\left[ \begin{array}{c c c} 2m-2j+1 & m-2i & 2j+2i+1 \end{array}\right]$.
A $(m-2j-1) \times 3$ block
\begin{equation}
\bm{S}^{\mathrm{V},j}(m) = \left[ \begin{array}{l c c c l}
                             \{ & 2m-2j-1 & 2j+1 & m+2 & \}^2 \\
                             \{ & 2m-2j-3 & 2j+1 & m+4 & \}^2 \\
                             & \vdots & \vdots & \vdots & \\
                             \{ & m+2 & 2j+1 & 2m-2j-1 & \}^2
                       \end{array} \right],
\end{equation}
with the $i$'th part ($i = 0,\ldots,\frac{m-3}{2}-j$) of the form $\left[ \begin{array}{c c c} 2m-2j-2i-1 & 2j+1 & m+2i+2 \end{array}\right]$.

Clearly each row of $\bm{S}(m)$ sums up to $3m+2$ and it is straightforward to verify that $\bm{S}(m)$ $m(m+1) \times 3$ implements $Z(m)$.
\end{proof}

\subsection{Proofs for the case of $\divides{K}{A}$}

\begin{proof}[Proof of Proposition~\ref{pr:aimpl5}]
Suppose that $K\geq 2$, $m \geq 1$, and $B = 2m-1$.

We start with rewriting $\bm{y}$ as follows:
\begin{align*}
\bm{y} & = \left(1-\frac{B}{Km}\right) \uvec{0} + \left(\frac{B}{Km}\right) \left(\left(\frac{m}{B}\right) \uovec{m} + \left(1-\frac{m}{B}\right)\uoupvec{m}\right) \\
& = \left(\frac{m(K-2)+1}{Km}\right) \uvec{0} + \left(\frac{1}{Km}\right) \left(m\uovec{m} + (m-1)\uoupvec{m}\right) \\
& = \left(\frac{m(K-2)}{Km}\right) \uvec{0} + \left(\frac{1}{Km}\right) \left(\uo{m} + \uoup{m} + \uvec{0}\right) \\
& = \left(\frac{1}{Km}\right)\left(m(K-2) \uvec{0} + \left(\uo{m} + \ue{m-1}\right)\right).
\end{align*}
Let $Y^m = m(K-2) \uvec{0} + \left(\uo{m} + \ue{m-1}\right)$.
To show that $\bm{x}$ is $(B,K)$-feasible, we will construct a $m \times K$ matrix such that $Y^m$ is a cardinality vector for it. Since $K \geq 2$ so $m(K-2)\uvec{0}$ is $(0,K-2)$-feasible and can be implemented by $m \times (K-2)$ matrix that contains zeros only.
Thus it is enough to construct a $m \times 2$ matrix $\bm{S}(m)$ that $(2m-1,2)$-implements $Z(m) = \uo{m} + \ue{m-1}$.
Notice that $Z(m) = [1,\ldots,1,0]$, that is $Z_{i} = 1$, for $i \in \{0,\ldots,2m-1\}$ and $Z_{2m} = 0$.
Let $\bm{S}(m)$ be $m \times 2$ matrix defined as follows:
\begin{displaymath}
\bm{S}(m) = \left[ \begin{array}{c c}
                       0 & 2m-1 \\
                       1 & 2m-2 \\
                       \vdots & \vdots \\
                       m-2 & m+1 \\
                       m-1 & m
                       \end{array}
                 \right],
\end{displaymath}
that is the $i$'th row, $i \in \{1,\ldots,m\}$, $s_{i} = \left[ \begin{array}{c c} i-1 & 2m-i \end{array} \right]$.
It is immediate to see that $\bm{S}(m)$ $(2m-1,2)$-implements $Z(m)$.
\end{proof}

In proof of Proposition~\ref{pr:aimpl4} we use the following theorem from~\cite{Dziubinski17}.

\begin{theorem}[\cite{Dziubinski17}]
\label{th:symmetric}
Let $A \geq 1$ and $K \geq 2$ be natural numbers such that $\divides{K}{A}$ and let $m = A/K$.
\begin{enumerate}
\item If $K$ is even, then $\lambda \uovec{m} + (1 - \lambda)\uevec{m}$ is $(A,K)$-feasible for all $\lambda \in [0,1]$.\label{p:symmetric:1}
\item If $K$ is odd and $A$ is odd, then $\lambda \uovec{m} + (1 - \lambda)\uevec{m}$ is $(A,K)$-feasible if and only if $\lambda \in \left[\frac{1}{K},1\right]$.\label{p:symmetric:2}
\item If $K$ is odd and $A$ is even, then $\lambda \uevec{m} + (1 - \lambda)\uovec{m}$ is $(A,K)$-feasible if and only if $\lambda \in \left[\frac{1}{K},1\right]$.\label{p:symmetric:3}
\end{enumerate}
\end{theorem}

With this theorem in hand we are ready to prove Proposition~\ref{pr:aimpl4}.

\begin{proof}[Proof of Proposition~\ref{pr:aimpl4}]
Suppose that $B$ is odd and let $r = B \bmod m$ and $L = \lfloor B/m \rfloor$ so that $B = Lm + r$. In addition, since $B \geq 2m$ so $L \geq 2$ and since $B \leq Km$ so either $L < K$ or $L = K$ and $r = 0$.

Notice that in the case of $L = K$ and $r = 0$, $B = Km$ and $\bm{y} = (1/K) \uovec{m} + ((K-1)/K)\uevec{m}$. Since $B$ is odd, so by points~\ref{p:symmetric:1} and~\ref{p:symmetric:2} of Theorem~\ref{th:symmetric}, $\bm{y}$ is $(B,K)$-feasible.
For the remaining part of the proof we assume that $L < K$. 

We start with rewriting $\bm{y}$ as follows:
\begin{align*}
\bm{y} & = \left(1-\frac{B}{Km}\right) \uvec{0} + \left(\frac{B}{Km}\right) \left(\left(\frac{m}{B}\right) \uovec{m} + \left(1-\frac{m}{B}\right)\uevec{m}\right) \\
& = \left(\frac{Km-B}{Km}\right) \uvec{0} + \left(\frac{1}{Km}\right) \left(m\uovec{m} + (B-m)\uevec{m}\right) \\
& = \left(\frac{Km-B}{Km}\right) \uvec{0} + \left(\frac{1}{Km}\right) \left(\uo{m} + \frac{B-m}{m+1}\ue{m}\right) \\
& = \left(\frac{1}{Km(m+1)}\right)\left((m+1)(m(K-L)-r) \uvec{0} + \left((m+1)\uo{m} + ((L-1)m+r)\ue{m}\right)\right).
\end{align*}
Let $Y^m = (m+1)(m(K-L)-r) \uvec{0} + ((m+1)\uo{m} + ((L-1)m+r)\ue{m})$.
To show that $\bm{y}$ is $(B,K)$-feasible, we will construct $m(m+1) \times K$ matrices such that
$Y^m$ is a cardinality vector for them. 

We consider two cases separately: (i)~$L\geq 3$ and $L$ odd and (ii)~$L\geq 2$ and $L$ even.

\noindent\textbf{Case~(i): $L \geq 3$ and $L$ odd.}

In this case $Y^m$ can be rewritten as $Y^m = (m-r)(m+1)  \uvec{0} + (m+1) \uo{m} + (2m+r) \ue{m} + m(m+1)(K-L-1)  \uvec{0} + (L-3)m\ue{m}$.
Since $L \geq 3$ and $L$ is odd so, by point~\ref{p:hart:2} of Proposition~\ref{pr:hart}, $(L-3)m\ue{m}$ is $((L-3)m,L-3)$-feasible and can be implemented by $m(m+1) \times (L-3)$ matrix $\sslash_{i = 1}^{m} |_{j = 1}^{\frac{L-3}{2}} \bm{E}(m)$.
Also, since $K > L$ so $m(m+1)(K-L-1)  \uvec{0}$ is $(0,K-L-1)$-feasible and can be implemented by $m(m+1) \times (K-L-1)$ matrix that contains zeros only.
Thus it is enough to construct a $m(m+1) \times 3$ matrix 
$\tilde{\bm{S}}(m)$ 
that $(3m+r,4)$-implements $Z(m) = (m-r)(m+1)  \uvec{0} + (m+1) \uo{m} + (2m+r) \ue{m}$.
Notice that $Z(m) = [(m-r)(m+1) + 2m+r,m+1,2m+r,\ldots,m+1,2m+r]$, that is $Z_{2i+1} = m+1$, for $i \in \{0,\ldots,m-1\}$, $Z_{0} = (m-r)(m+1) + 2m+r$, and $Z_{2i} = 2m+r$, for $i \in \{1,\ldots,m\}$.
Moreover, since $B$ is odd and $L$ is odd so $r \equiv (m+1) \pmod 2$.

Let 
$\tilde{\bm{S}}(m,r)$ be a $m(m+1) \times 4$ matrix obtained from matrix $\bm{S}(m,r+1)$, defined in proof of Proposition~\ref{pr:aimpl3}, by decreasing by $1$ all the values in the first column.

\noindent\textbf{Case~(ii): $L\geq 2$ and $L$ even.}
In this case $Y^m$ can be rewritten as $Y^m = (m-r)(m+1)  \uvec{0} + (m+1) \uo{m} + (m+r) \ue{m} + m(m+1)(K-L-1)  \uvec{0} + (L-2)m\ue{m}$.
Since $L \geq 2$ and $L$ is even so, by point~\ref{p:hart:2} of Proposition~\ref{pr:hart}, $(L-2)m\ue{m}$ is $((L-2)m,L-2)$-feasible and can be implemented by $m(m+1) \times (L-2)$ matrix $\sslash_{i = 1}^{m} |_{j = 1}^{\frac{L-2}{2}} \bm{E}(m)$.
Also, since $K > L$ so $m(m+1)(K-L-1)  \uvec{0}$ is $(0,K-L-1)$-feasible and can be implemented by $m(m+1) \times (K-L-1)$ matrix that contains zeros only.
Thus it is enough to construct a $m(m+1) \times 3$ matrix $\tilde{\bm{T}}(m)$ that $(2m+r,3)$-implements $Z(m) = (m-r)(m+1)  \uvec{0} + (m+1) \uo{m} + (m+r) \ue{m}$.
Notice that $Z(m) = [(m-r)(m+1) + m+r,m+1,m+r,\ldots,m+1,m+r]$, that is $Z_{2i+1} = m+1$, for $i \in \{0,\ldots,m-1\}$, $Z_{0} = (m-r)(m+1) + m+r$, and $Z_{2i} = m+r$, for $i \in \{1,\ldots,m\}$.
Moreover, since $B$ is odd and $L$ is even so $r$ is odd.

Let ${\tilde{\bm{T}}}(m,r)$ be a $m(m+1) \times 3$ matrix obtained from matrix $\bm{T}(m,r+1)$, defined in proof of Proposition~\ref{pr:aimpl3}, by decreasing by $1$: 
\begin{enumerate} 
\item all the values in the second column of block $\bm{T}^{\mathrm{I}}(m,r+1)$,
\item all the values in the first column of every odd numbered row of each of the remaining blocks (numbering rows in each block starting from $1$),
\item all the values in the second column of every even numbered row of each of the remaining blocks (numbering rows in each block starting from $1$).
\end{enumerate}
\end{proof}

\section{Proofs for the cases beyond General Lotto}

\begin{proof}[Proof of Proposition~\ref{pr:general:constr}]
Take any integer $m \geq 1$ and real $a =m + \alpha$, $b \leq m$, $\alpha \in (0,1)$, and $c \in (0,b/(m+1)]$.
Let $\vvvec{m} = \frac{1}{m}\sum_{i = 1}^{m} \vvec{i}{m}$, $\delta = (2m+1)/(m+1)$, and $\sigma = (2m+1)/m$.

To prove the proposition with first compute lower bounds on the values of $H(X,Y)$ for non-negative integer valued random variable $X$ distributed according to strategy $\bm{x}$, defined in the proposition, and a non-negative integer valued random variable $Y$, distributed with probability distribution $\bm{p}$ such that $\sum_{i \geq 0} p_{2i+1} = c$ and $\Ex(Y) = b$.

As was shown in~\cite{Hart08}, if $Z$ is a random variable distributed with $\uovec{m}$ then 
\begin{equation*}
H(Z,Y) = 1 - \left(\frac{1}{m}\right)\sum_{i=1}^{2m} \Pr(Y \geq i)  = 1 - \frac{b}{m} + \left(\frac{1}{m}\right)\sum_{i\geq 2m+1} \Pr(Y \geq i)
\end{equation*}
and since
\begin{equation*}
\sum_{i\geq 2m+1} \Pr(Y \geq i) = \sum_{i\geq 2m+1} (i-2m)p_i = p_{2m+1} + \sum_{i\geq 2m+2} (i-2m)p_i \geq p_{2m+1}
\end{equation*}
with equality if and only if $\sum_{i\geq 2m+2} p_i = 0$ so
\begin{equation*}
H(Z,Y) \geq 1 - \frac{b}{m} + \frac{p_{2m+1}}{m}.
\end{equation*}
Similarly, if $Z$ is a random variable distributed with $(1-\alpha)\uovec{m} + \alpha\uovec{m+1}$ then 
\begin{align*}
H(Z,Y) & = 1 - \frac{(1-\alpha) b}{m} - \frac{\alpha b}{m+1} + \frac{1-\alpha}{m}\sum_{i \geq 2m+1} \Pr(Y \geq i)+ \frac{\alpha}{m+1}\sum_{i \geq 2m+2} \Pr(Y \geq i)\\
& \geq 1 - \frac{(1-\alpha) b}{m} - \frac{\alpha b}{m+1} + \frac{(1-\alpha)p_{2m+1}}{m}
\end{align*}
with equality if and only if $\sum_{i\geq 2m+2} p_i = 0$.

As was shown in~\cite{Dziubinski12}, if $Z$ is a random variable distributed with $\vvec{j}{m}$ then
\begin{align*}
H(Z,Y) & = 1 - \left(\frac{1}{2m+1}\right)\left(- p_{2j-1} + 2\sum_{i=1}^{2m+1} \Pr(Y \geq i) \right)\\
& = 1 - \frac{2b}{2m+1} + \frac{p_{2j-1}}{2m+1} + \frac{2}{2m+1}\sum_{i \geq 2m+2} \Pr(Y \geq i).
\end{align*}
Using that, if $Z$ is a random variable distributed with $\vvvec{m}$ then 
\begin{align*}
H(Z,Y) & = 1 - \frac{2b}{2m+1} + \frac{\sum_{i = 1}^m p_{2i-1}}{m(2m+1)} + \frac{2}{2m+1}\sum_{i \geq 2m+2} \Pr(Y \geq i)\\
& = 1 - \frac{2b}{2m+1} + \frac{c}{m(2m+1)} + \frac{2}{2m+1}\sum_{i \geq 2m+2} \Pr(Y \geq i) 
- \frac{\sum_{i \geq m} p_{2i+1}}{m(2m+1)}
\end{align*}
and since
\begin{align*}
& \frac{2}{2m+1}\sum_{i \geq 2m+2} \Pr(Y \geq i) - \frac{\sum_{i \geq m} p_{2i+1}}{m(2m+1)}\\
&\qquad = \frac{1}{2m+1}\left(2\sum_{i \geq 2m+2} (i-2m-1) p_i - \frac{p_{2m+1}}{m} - \frac{1}{m}\sum_{i \geq m+1} p_{2i+1}\right) \\
&\qquad = \frac{1}{2m+1}\left(2\sum_{i \geq m+1} (2i-2m-1) p_{2i}  
+ \sum_{i \geq m+1} \left(4i-4m-\frac{1}{m}\right) p_{2i+1}\right) - \frac{p_{2m+1}}{m(2m+1)}\\
& \qquad \geq - \frac{p_{2m+1}}{m(2m+1)}
\end{align*}
with equality if and only if $\sum_{i \geq 2m+2} p_i = 0$ so
\begin{align*}
H(Z,Y) & \geq 1 - \frac{2b}{2m+1} + \frac{c}{m(2m+1)} - \frac{p_{2m+1}}{m(2m+1)}
\end{align*}
with equality if and only if $\sum_{i \geq 2m+2} p_i = 0$.

Using these derivations, if $X$ is a random variable distributed with $\bm{x}  = \alpha \delta \vvvec{m} + (1 - \alpha \delta) \uovec{m}$ then
\begin{equation}
\label{eq:pr:general:constr:1}
\begin{aligned}
H(X,Y) & \geq 1 - \frac{2b\alpha}{m+1} + \frac{\alpha c}{m(m+1)} - \frac{(m+1-(2m+1)\alpha)b}{m(m+1)}\\
& - \frac{\alpha p_{2m+1}}{m(m+1)} + \frac{(m+1-(2m+1)\alpha)p_{2m+1}}{m(m+1)}  \\
& = 1  - \frac{(1-\alpha)b}{m} - \frac{\alpha b}{m+1} + \frac{\alpha c}{m(m+1)} + \frac{(1-2\alpha)p_{2m+1}}{m},
\end{aligned}
\end{equation}
with equality if and only if $\sum_{i \geq 2m+2} p_i = 0$.

Similarly, if $X$ is a random variable distributed with 
\begin{align*}
\bm{x} & = \left(\frac{1-\alpha}{\alpha}\right) \left(\alpha \delta \vvvec{m} + (1 - \alpha \delta) \uovec{m}\right) + \left(\frac{2\alpha - 1}{\alpha}\right) \left( (1 - \alpha) \uovec{m} + \alpha \uovec{m+1} \right)
\end{align*}
then
\begin{equation}
\label{eq:pr:general:constr:2}
\begin{aligned}
H(X,Y) & \geq 1 - \frac{(1-\alpha) b}{m} - \frac{\alpha b}{m+1} + \frac{(1-\alpha) c}{m(m+1)}\\
& \qquad\qquad {} - \frac{(1-\alpha)(2\alpha-1)p_{2m+1}}{m\alpha} + \frac{(1-\alpha)(2\alpha-1)p_{2m+1}}{m\alpha}\\
& = 1 - \frac{(1-\alpha) b}{m} - \frac{\alpha b}{m+1} + \frac{(1-\alpha) c}{m(m+1)}.
\end{aligned}
\end{equation}
with equality if and only if $\sum_{i \geq 2m+2} p_i = 0$.

Next, we compute lower bounds on the values of $H(Y,X)$ for non-negative integer valued random variable $Y$ distributed according to strategy $\bm{y}$, defined in the proposition, and non-negative integer valued random variable $X$, distributed with probability distribution $\bm{q}$ such that and $\Ex(X) = m+\alpha$.

As was shown in~\cite{Hart08}, if $Z$ is a random variable distributed with $\uevec{m}$ then 
\begin{align*}
H(Z,X) & = 1 - \left(\frac{1}{m+1}\right)\left(1+\sum_{i=1}^{2m+1} \Pr(X \geq i)\right) \geq 1 - \frac{m+1+\alpha}{m+1}
\end{align*}
with equality if and only if $\sum_{i \geq 2m+2} q_i = 0$.
Thus if $Z$ is a random variable distributed with $\lambda \uovec{m} + (1-\lambda)\uevec{m}$
then 
\begin{equation*}
H(Z,X) \geq 1 - \frac{\lambda (m+\alpha)}{m} + \frac{\lambda q_{2m+1}}{m} - \frac{(1-\lambda)(m+1+\alpha)}{m+1} \geq - \frac{\alpha}{m+1} - \frac{\lambda\alpha}{m(m+1)},
\end{equation*}
with equality if and only if $\sum_{i \geq 2m+1} q_i = 0$.
Similarly, if $Z$ is a random variable distributed with $\lambda \uovec{m+1} + (1-\lambda)\uevec{m}$
then 
\begin{equation*}
H(Z,X) \geq 1 - \frac{\lambda (m+\alpha)}{m+1} - \frac{(1-\lambda)(m+1+\alpha)}{m+1} = \frac{\lambda-\alpha}{m+1},
\end{equation*}
with equality if and only if $\sum_{i \geq 2m+2} q_i = 0$.

Using these derivations, if $Y$ is a random variable distributed with
\begin{equation*}
\bm{y} = \left(1 - \frac{b}{m}\right)\uvec{0} + \frac{b}{m}\left(\frac{cm}{b} \uovec{m} + \left(1-\frac{cm}{b}\right)\uevec{m}\right)
\end{equation*}
then
\begin{equation}
\label{eq:pr:general:constr:3}
\begin{aligned}
H(Y,X) & \geq \left(1 - \frac{b}{m}\right)(q_0-1)- \frac{b}{m}\left(\frac{\alpha}{m+1} + \frac{c\alpha}{b(m+1)}\right)\\
& \geq -1+\frac{b}{m} - \frac{b\alpha}{m(m+1)} - \frac{c\alpha}{m(m+1)}\\
& = -1+\frac{\alpha b}{m+1} + \frac{(1-\alpha)b}{m} - \frac{c\alpha}{m(m+1)}\\
\end{aligned}
\end{equation}
with equality if and only if $q_0 = 0$ and $\sum_{i \geq 2m+2} q_i = 0$.

Similarly, if $Y$ is a random variable distributed with 
\begin{equation*}
\bm{y} = \left(1 - \frac{b-c}{m}\right)\uvec{0} + \frac{b-c}{m}\left(\frac{mc}{b-c} \uovec{m+1} + \left(1 - \frac{mc}{b-c}\right)\uevec{m}\right)
\end{equation*}
then
\begin{equation}
\label{eq:pr:general:constr:4}
\begin{aligned}
H(Y,X) & \geq \left(1 - \frac{b-c}{m}\right)(q_0-1) + \frac{(m+\alpha)c - b \alpha}{m(m+1)}\\
& \geq -1 + \frac{b-c}{m} + \frac{(m+\alpha)c - b \alpha}{m(m+1)}\\
& = -1 + \frac{\alpha b}{m+1} + \frac{(1-\alpha)b}{m} - \frac{(1-\alpha)c}{m(m+1)},
\end{aligned}
\end{equation}
with equality if and only if $q_0 = 0$ and $\sum_{i \geq 2m+2} q_i = 0$.

Now we are ready to prove optimality of the strategies defined in the proposition.
Suppose that $\alpha \in (0,1/2)$. By~\eqref{eq:pr:general:constr:1}, if $X$ is a non-negative integer valued random variable distributed with optimal strategy $\bm{x}$ of player A, defined in point~\ref{p:general:constr:1} of the proposition, then 
\begin{align*}
H(X,Y) & \geq 1  - \frac{(1-\alpha)b}{m} - \frac{\alpha b}{m+1} + \frac{\alpha c}{m(m+1)} + \frac{(1-2\alpha)p_{2m+1}}{m}\\
       & = 1  - \frac{(1-\alpha)b}{m} - \frac{\alpha b}{m+1} + \frac{\alpha c}{m(m+1)},
\end{align*}
with equality if and only if $\sum_{i \geq 2m+1} p_i = 0$, for a non-negative integer valued random variable $Y$ distributed with any strategy of player B. On the other hand, by~\eqref{eq:pr:general:constr:3}, if $Y$ is a non-negative integer valued random variable distributed with optimal strategy $\bm{y}$ of player B, defined in point~\ref{p:general:constr:1} of the proposition, then 
\begin{equation*}
H(Y,X) \geq -1+\frac{\alpha b}{m+1} + \frac{(1-\alpha)b}{m} - \frac{c\alpha}{m(m+1)},
\end{equation*}
for a non-negative integer valued random variable $X$ distributed with any strategy of player A. Since the game is zero-sum, this proves that $\bm{x}$ and $\bm{y}$ are optimal strategies of players A and B, respectively, when $\alpha \in (0,1/2)$.

Suppose that $\alpha \in [1/2,1)$. 
By~\eqref{eq:pr:general:constr:2}, if $X$ is a non-negative integer valued random variable distributed with optimal strategy $\bm{x}$ of player A, defined in point~\ref{p:general:constr:2} of the proposition, then 
\begin{equation*}
H(X,Y) \geq 1 - \frac{(1-\alpha) b}{m} - \frac{\alpha b}{m+1} + \frac{(1-\alpha) c}{m(m+1)},
\end{equation*}
for a non-negative valued random variable $Y$ distributed with any strategy of player B. 
By~\eqref{eq:pr:general:constr:4}, if $Y$ is a non-negative integer valued random variable distributed with optimal strategy $\bm{y}$ of player B, defined in point~\ref{p:general:constr:2} of the proposition, then 
\begin{equation*}
H(Y,X) \geq -1 + \frac{\alpha b}{m+1} + \frac{(1-\alpha)b}{m} - \frac{(1-\alpha)c}{m(m+1)},
\end{equation*}
for non-negative integer valued random variable $X$ distributed with any strategy of player A. Since the game is zero-sum, this proves that $\bm{x}$ and $\bm{y}$ are optimal strategies of players A and B, respectively, when $\alpha \in [1/2,1)$.
\end{proof}

\begin{proof}[Proof of Proposition~\ref{pr:aimpl6}]
Suppose that $B$ is odd and let $r = (B-1) \bmod m$ and $L = \lfloor (B-1)/m \rfloor$ so that $B = Lm+r+1$. In addition, since $B \geq 2m+1$ so $L \geq 2$ and since $B \leq Km$ so $L < K$.

We start by rewriting $\bm{y}$ as follows:
\begin{align*}
\bm{y} & = \left(1-\frac{B-1}{Km}\right) \uvec{0} + \left(\frac{B-1}{Km}\right) \left(\left(\frac{m}{B-1}\right) \uovec{m+1} + \left(1-\frac{m}{B-1}\right)\uevec{m}\right) \\
& = \left(\frac{Km-B+1}{Km}\right) \uvec{0} + \left(\frac{1}{Km}\right) \left(m\uovec{m+1} + (B-m-1)\uevec{m}\right) \\
& = \left(\frac{Km-B+1}{Km}\right) \uvec{0} + \left(\frac{1}{Km}\right) \left(\frac{m}{m+1}\uo{m+1} + \frac{B-m-1}{m+1}\ue{m}\right) \\
& = \left(\frac{1}{Km(m+1)}\right)\left((m+1)(Km-B+1) \uvec{0} + \left(m\uo{m+1} + (B-m-1)\ue{m}\right)\right) \\
& = \left(\frac{1}{Km(m+1)}\right)\left((m+1)(m(K-L)-r) \uvec{0} + \left(m\uo{m+1} + ((L-1)m+r)\ue{m}\right)\right)
\end{align*}
Let $Y^m = (m+1)(m(K-L)-r) \uvec{0} + (m\uo{m+1} + ((L-1)m+r)\ue{m})$.
To show that $\bm{y}$ is $(B,K)$-feasible, we will construct $m(m+1) \times K$ matrices such that
$Y^m$ is the cardinality vector for them. 

We consider two cases separately: (i)~$L\geq 3 $ and $L$ odd and (ii)~$L\geq 2$ and $L$ even.

\noindent\textbf{Case~(i): $L$ odd.}

In this case $Y^m$ can be rewritten as $Y^m = (m-r)(m+1)  \uvec{0} + m \uo{m+1} + (2m+r) \ue{m} + m(m+1)(K-L-1)  \uvec{0} + (L-3)m\ue{m}$.
Since $L \geq 3$ and $L$ is odd so, by point~\ref{p:hart:2} of Proposition~\ref{pr:hart}, $(L-3)m\ue{m}$ is $((L-3)m,L-3)$-feasible and can be implemented by $m(m+1) \times (L-3)$ matrix $\sslash_{i = 1}^{m} |_{j = 1}^{\frac{L-3}{2}} \bm{E}(m)$.
Also, since $K > L$ so $m(m+1)(K-L-1)  \uvec{0}$ is $(0,K-L-1)$-feasible and can be implemented by $m(m+1) \times (K-L-1)$ matrix that contains zeros only.
Thus it is enough to construct a $m(m+1) \times 3$ matrix $\tilde{\bm{S}}(m)$ that $(3m+r+1,4)$-implements $Z(m) = (m-r)(m+1)  \uvec{0} + m \uo{m+1} + (2m+r) \ue{m}$.
Notice that $Z(m) = [(m-r)(m+1) + 2m+r,m,2m+r,\ldots,m,2m+r,m]$, that is $Z_{2i+1} = m$, for $i \in \{0,\ldots,m\}$, $Z_{0} = (m-r)(m+1) + 2m+r$, and $Z_{2i} = 2m+r$, for $i \in \{1,\ldots,m\}$.
Moreover, since $B$ is odd and $L$ is odd so $r \equiv m \pmod 2$.

Let 
$\tilde{\bm{S}}(m,r)$ be a $m(m+1) \times 4$ matrix obtained from matrix $\bm{S}(m,r)$, defined in proof of Proposition~\ref{pr:aimpl3}, by increasing by $1$
\begin{enumerate} 
\item the value in the first column of the first row and the value in the third column of the second row of each part of block $\bm{S}^{\mathrm{I}}(m,r)$ (parts of this block consist of two rows, the value in the third column is $0$ in every row),
\item the value in the third column of the first row and the values in the first column of the remaining rows of every part of block $\bm{S}^{\mathrm{IV}}(m,r)$ (the value in the third column is $0$ in every row),
\item all the values in the first column of the remaining blocks.
\end{enumerate}

\noindent\textbf{Case~(ii): $L\geq 2$ and $L$ even.}
In this case $Y^m$ can be rewritten as $Y^m = (m-r)(m+1)  \uvec{0} + m \uo{m+1} + (m+r) \ue{m} + m(m+1)(K-L-1)  \uvec{0} + (L-2)m\ue{m}$.
Since $L \geq 2$ and $L$ is even so, by point~\ref{p:hart:2} of Proposition~\ref{pr:hart}, $(L-2)m\ue{m}$ is $((L-2)m,L-2)$-feasible and can be implemented by $m(m+1) \times (L-2)$ matrix $\sslash_{i = 1}^{m} |_{j = 1}^{\frac{L-2}{2}} \bm{E}(m)$.
Also, since $K > L$ so $m(m+1)(K-L-1)  \uvec{0}$ is $(0,K-L-1)$-feasible and can be implemented by $m(m+1) \times (K-L-1)$ matrix that contains zeros only.
Thus it is enough to construct a $m(m+1) \times 3$ matrix $\tilde{\bm{T}}(m)$ that $(2m+r+1,3)$-implements $Z(m) = (m-r)(m+1)  \uvec{0} + m \uo{m+1} + (m+r) \ue{m}$.
Notice that $Z(m) = [(m-r)(m+1) + m+r,m,m+r,\ldots,m,m+r,m]$, that is $Z_{2i+1} = m$, for $i \in \{0,\ldots,m\}$, $Z_{0} = (m-r)(m+1) + m+r$, and $Z_{2i} = m+r$, for $i \in \{1,\ldots,m\}$.
Moreover, since $B$ is odd and $L$ is even so $r$ is even.

Let $\tilde{\bm{T}}(m,r)$ be a $m(m+1) \times 3$ matrix obtained from matrix $\bm{T}(m,r)$, defined in proof of Proposition~\ref{pr:aimpl3}, by increasing by $1$: 
\begin{enumerate} 
\item the values in the first column of the first two rows and the values in the second column of the remaining rows of parts $i \in \{1,\ldots,\frac{r}{2}-1\}$ of block $\bm{T}^{\mathrm{I}}(m,r)$ (the values in the first column are $0$ for every part and the value in the second row of part $i$ is $r+2i$),
\item the values in the first column of the first row and the values in the second column of the remaining rows of parts $i = 0$ and $i \in \{\frac{r}{2},\ldots,m-\frac{r}{2}\}$ of block $\bm{T}^{\mathrm{I}}(m,r)$
(like in the previous point, the values in the first column are $0$ for every part and the value in the second row of part $i$ is $r+2i$),
\item the values in the second column of the two rows of part $i = j$ and the value in the first column of the first row and the value in the second column of the second row of each of the remaining parts $i \in \{0,\ldots,m-\frac{r}{2}\}\setminus \{j\}$ of block $\bm{T}^{\mathrm{II},j}(m,r)$, for each $j \in \{1,\ldots,\frac{r}{2}-1\}$ (the values in the first column are $2j$ for every part and the value in the second row of part $i$ is $r+2i$),
\item all the values in the first column of every odd numbered row of each of the remaining blocks (numbering rows in each block starting from $1$),
\item all the values in the second column of every even numbered row of each of the remaining blocks (numbering rows in each block starting from $1$).
\end{enumerate}

\end{proof}

\end{appendix}

\end{document}